\documentclass[11pt,letterpaper]{article}

\usepackage{amssymb,amsmath,amsfonts}
\usepackage{graphicx,xcolor,enumitem}
\usepackage{epsfig}
\usepackage{amsthm}
\usepackage{bm}
\usepackage{subcaption}
\usepackage{nicematrix}
\usepackage[round]{natbib}
\usepackage{csquotes}

\usepackage{multirow}
\usepackage{algorithm, algorithmic}

\usepackage{diagbox}
	
\renewcommand{\mkbegdispquote}[2]{\itshape}
\usepackage[twoside, hmarginratio=1:1, vmarginratio=1:1, left=1in,top=1in]{geometry}

\RequirePackage[breaklinks=true, hidelinks]{hyperref}
\usepackage{breakcites}

\newcommand{\cA}{\mathcal{A}}

\newcommand{\E}{\mathbb{E}}
\newcommand{\R}{\mathbb{R}}
\newcommand{\p}{\mathbb{P}}

\newcommand{\cL}{{\mathcal L}}
\newcommand{\cS}{{\mathcal S}}
\newcommand{\cF}{{\mathcal F}}

\newcommand{\cU}{{\mathcal U}}
\newcommand{\cV}{{\mathcal V}}
\newcommand{\cM}{{\mathcal M}}
\newcommand{\cN}{{\mathcal N}}
\newcommand{\cI}{{\mathcal I}}

\newcommand{\tr}{{\rm{tr}}}

\newcommand{\one}{\mathbf{1}}

\newtheorem{theorem}{Theorem}

\newtheorem{assumption}[theorem]{Assumption}

\newtheorem{definition}[theorem]{Definition}

\newtheorem{lemma}[theorem]{Lemma}

\newtheorem{proposition}[theorem]{Proposition}
\newtheorem{remark}[theorem]{Remark}

\theoremstyle{definition}

\numberwithin{equation}{section}
\numberwithin{theorem}{section}

\begin{document}

\title{Goal-based portfolio selection with mental accounting\footnote{The authors thank Professor Thaleia Zariphopoulou, Professor Panagiotis E. Souganidis, and seminar participants at the Byrne B2A2 (Back to Ann Arbor) Conference for their helpful comments. Erhan Bayraktar is partially supported by the National Science Foundation and by the Susan M. Smith chair.  Bingyan Han is partially supported by the Guangzhou-HKUST(GZ) Joint Funding Scheme (No. 2025A03J3858).}}

\author{
	Erhan Bayraktar \thanks{Department of Mathematics, University of Michigan, Ann Arbor, Email: erhan@umich.edu.}
	\and Bingyan Han \thanks{Thrust of Financial Technology, The Hong Kong University of Science and Technology (Guangzhou), Email: bingyanhan@hkust-gz.edu.cn.}
}

\date{May 9, 2026}
\maketitle

\begin{abstract}
	We present a continuous-time portfolio selection framework that reflects goal-based investment principles and mental accounting behavior. In this framework, an investor with multiple investment goals constructs separate portfolios, each corresponding to a specific goal, with penalties imposed on fund transfers between these goals, referred to as mental costs. By applying the stochastic Perron's method, we demonstrate that the value function is the unique constrained viscosity solution of a Hamilton-Jacobi-Bellman equation system. Numerical analysis reveals several key features: the free boundaries exhibit complex shapes with bulges and notches; the optimal strategy for one portfolio depends on the wealth level of another; investors must diversify both among stocks and across portfolios; and they may postpone reallocating surplus from an important goal to a less important one until the former's deadline approaches.
	\\[2ex] 
	\\[2ex]
\end{abstract}

\section{Introduction}
The classical mean-variance portfolio selection and expected utility theory have two fundamental assumptions. First, investors possess a clear understanding of risk aversion and can specify precise risk aversion values for their investment decisions. Second, the principle of fungibility holds. Money is a fungible asset, as any \$1 can be exchanged for another. However, behavioral research and recent developments in wealth management challenge these assumptions \citep{thaler1985mental,barberis2001mental,choi2009mental,das2010portfolio,das2022dynamic,capponi2024}.  

Estimating risk aversion parameters can be a challenging task. In contrast, investors are better at stating the thresholds and relative importance of their goals than their risk aversion coefficients. For instance, investors are typically aware of the exact amount and deadline for their children's college tuition fees, and can clearly identify this educational goal as significantly more important than a vacation. Goal-based investing aligns investors' needs with portfolio design, receiving attention from both the academic community \citep{das2010portfolio,das2022dynamic,capponi2024,gargano2024goal} and the retail robo-advising industry. \cite{capponi2024} highlights that investment platforms like Schwab and Betterment can specify various goals, including retirement, house downpayments, and car purchases.

On the other hand, for the principle of fungibility, \cite{thaler1985mental} mentioned the following anecdote:

``Mr. and Mrs. L and Mr. and Mrs. H went on a fishing trip in the northwest and caught some salmon. They packed the fish and sent it home on an airline, but the fish were lost in transit. They received \$300 from the airline. The couples take the money, go out to dinner and spend \$225. They had never spent that much at a restaurant before."

\cite{thaler1985mental} commented that the dinner would not have occurred if the \$300 had come from a salary increase. This lavish dinner, exceeding their typical budget, resulted from perceiving the \$300 as a ``windfall gain". \cite{thaler1985mental} refers to this cognitive phenomenon as mental accounting. Money is mentally categorized based on its intended purpose, creating the illusion of different values and thereby violating the principle of fungibility. Additionally, distinct financial accounts also exist physically to serve specific objectives, such as the 529 College Savings Plan and Coverdell Education Savings Accounts for educational expenses, or 401(k) and IRA accounts for retirement savings.

Motivated by the observations above, we propose a continuous-time portfolio selection framework that aligns with the goal-based investment philosophy and mental accounting behaviors. Our primary assumption is that individuals, considering multiple investment goals, construct several portfolios, each linked to a specific goal. Consequently, funds are labeled according to these goals. To reflect their priorities, investors can assign different weights to each goal. Moreover, penalties can be imposed when transferring funds between goals, with a higher penalty applied to transfers from more important goals to less important ones. We refer to these penalties as mental costs, which resemble transaction costs, although investors do not lose actual money in transfers. Each portfolio concludes upon the expiration of its corresponding goal, and the terminal value of each portfolio reflects the wealth amount to fulfill that goal.

Our technical contributions are as follows. The problem inherently involves multiple wealth variables due to its multiple goals. Moreover, the value function loses one state variable after the deadline for each goal, which is a distinctive feature of our problem. Through a heuristic derivation, we obtain a Hamilton-Jacobi-Bellman (HJB) equation system. Our main technical result is that the value function is the unique constrained viscosity solution of this HJB equation system, which is also continuous when restricted to the time horizon between consecutive goal deadlines. Our proof technique is the stochastic Perron's method; see \cite{bayraktar2013stochastic,bayraktar2015stochastic}. Three main challenges arise in the proof. First, we modify the stochastic Perron's method to establish the viscosity solution property at the goal deadline when the value function loses a state variable. Second, there is a state constraint that the wealth processes are nonnegative, which complicates the proof of the viscosity solution property on the boundary. We adopt the concept of constrained viscosity solutions as in \cite{zariphopoulou1994consumption,bentahar2007}. Third, the mental costs lead to gradient constraints from two sides as follows:
\begin{equation*}
	- \theta_{i}  \leq \partial_{K+1} V_k - \partial_{i} V_k \leq \lambda_{i}, \quad i = k, \ldots, K.
\end{equation*}
Consequently, the proof of the comparison principle requires modification, revealing an interesting connection with the technique in \citet[Proposition 2.2]{hynd2012eigenvalue}. Furthermore, it is also harder to find a strict classical subsolution in Lemma \ref{lem:strict_l}. 


We apply a standard numerical scheme based on the finite difference discretization and the penalization method; see \citet[Section 2.3]{azimzadeh2017}. Given many parameters and factors influencing optimal decisions, we conduct a detailed numerical analysis of optimal behaviors in Section \ref{sec:num}:
\begin{enumerate}[label={(\arabic*)}]	
	\item The free boundaries exhibit complex shapes with bulges and notches,  indicating that optimal transfer decisions between goals are nonlinear in wealth variables. These boundaries are dynamic, reflecting the trade-off between current and future goals. When two goals are active, the optimal investment in one portfolio also depends on the wealth level of another portfolio.
	\item Our model introduces a new diversification problem. Agents must diversify not only among stocks but also across portfolios. A portfolio with surplus can use its positions to hedge risks in another underfunded portfolio. The optimal hedging method is strongly influenced by stock return parameters such as correlation.
	\item If an agent prioritizes one goal over another, she will allocate all available funds to the more important goal at its deadline until the target is met, causing the continuation region to shrink to line segments. If the important goal has enough funds already, the agent may still delay allocating surplus to the less important goal until the important goal's deadline approaches.
\end{enumerate}

Our paper contributes to the growing literature on goal-based portfolio selection. \cite{das2010portfolio} explore separate portfolios for different goals and impose different thresholds on the failure probability of each goal. They focus on a single time period, with all goals sharing the same deadline. \cite{das2022dynamic} extend this by allowing different deadlines and capturing the competition among goals, though their model assumes finite states for both strategy and wealth. Using data from a FinTech app, \cite{gargano2024goal} show that setting savings goals increases individual savings rates. \cite{capponi2024} propose a continuous-time framework for multiple-goal wealth management, which is also solved using the HJB equation method. They construct a unified portfolio that encompasses all goals, making withdrawals when a decision is made to fulfill a goal. Our model reduces to \cite{capponi2024} when mental cost parameters are set to zero. Thus, our results can be seen as interpolating between two extremes: zero mental cost, as in \cite{capponi2024}, and infinite mental cost, where portfolios are managed separately for each goal.

Another relevant stream of literature is on the stochastic Perron's method. Following the early developments by \cite{bayraktar2012linear,bayraktar2013stochastic,bayraktar2014Dynkin}, \cite{bayraktar2015stochastic} study the lifetime ruin problem under proportional transaction costs. \cite{belak2019utility} consider the utility maximization problem with constant and proportional transaction costs. \cite{belak2022optimal} further explore the case with strictly positive transaction costs.  \cite{rokhlin2014} extends the stochastic Perron's method to optimal control problems with state constraints. Other applications can be found in \cite{sirbu2014stochastic,bayraktar2016robust,bayraktar2016targetgame,bayraktar2016target,belak2017impulse,bayraktar2022vanish}, to name a few.

Finally, our mental cost formulation is closely related to the extensive literature on transaction costs, pioneered by \cite{magill1976portfolio,davis1990portfolio,davis1993european,shreve1994optimal}. However, unlike transaction cost problems where costs appear in the state process, our proportional mental costs appear in the objective function. This leads to significant differences in the shape of the continuation regions. Our value function does not exhibit homotheticity as in \cite{shreve1994optimal}. Another distinctive feature is the reduction in the number of state variables by one after each deadline. Moreover, the PDE in the continuation region still has the Hamiltonian operator and results in a degenerate nonlinear equation, motivating the use of viscosity solution methods.

The rest of the paper is organized as follows. Section \ref{sec:form} introduces the motivation and formulation for goal-based portfolio selection with mental accounting. Section \ref{sec:HJB} derives the HJB equation and presents the concepts of viscosity solutions. Section \ref{sec:num} details the numerical analysis\footnote{The visualization of the free boundaries can be found at \url{https://github.com/hanbingyan/goaldemo}.}. The proof of the main Theorem \ref{thm:viscosity} then proceeds in three steps: the construction of stochastic supersolutions and a viscosity subsolution in Section \ref{sec:sto_sup}, the construction of stochastic subsolutions and a viscosity supersolution in Section \ref{sec:sto_sub}, and the comparison principle in Section \ref{sec:compare}. {\color{black} Section \ref{sec:con} concludes the paper.} Finally, the Appendix provides the proofs of the main results.

\section{Formulation}\label{sec:form}
{\color{black} \subsection{Empirical motivation}

In addition to the anecdote from \cite{thaler1985mental}, empirical evidence from \cite{choi2009mental} illustrates mental account segregation in a retirement savings context. The authors analyzed a firm that changed its 401(k) matching rules in March 2003.
\begin{enumerate}[label={(\arabic*)}]	
	\item Before March 2003, the firm directed 100\% of matching contributions into employer stock. Employees chose the allocation only for their own contributions, although they were free to transfer the employer-match balances out of employer stock after receiving them.
	\item After March 2003, the firm removed this restriction. New enrollees were required to actively choose the asset allocation for both their own contributions and the employer match.
\end{enumerate}

If employees fully integrated the two accounts, they could have offset the employer-stock exposure in the match account by reducing employer-stock purchases in their own-contribution account. The data suggest that many employees did not make this offsetting adjustment. Employees who enrolled in February 2003 allocated 23.2\% of their own-contribution flows to employer stock; combined with the employer match, their total 401(k) allocation to employer stock was 55.9\%; see \citet[Table 1]{choi2009mental}. By contrast, employees who enrolled in March 2003, after the allocation choices for own contributions and employer matches were made jointly, had a total combined allocation to employer stock of 22.5\%. Notably, the own-contribution allocations to employer stock in the first regime precisely match the combined allocation of both contribution types in the second regime.

This example is consistent with the mental accounting framework. When the plan design presented the components as distinct accounts, employees behaved as if they were partially segregated, illustrating a flypaper effect where money remains in its initial category. In the second regime, the firm made the two allocation decisions jointly salient. This institutional change reduced the cognitive and psychological cost of integrating the two accounts, and the resulting employer-stock exposure fell substantially.

This study suggests that the cognitive effort required to break down mental boundaries acts as a substantial friction, which we refer to as a mental cost. These frictions affect the investor's decision process even though they do not mechanically reduce the total financial wealth. The relevant friction in the first regime was not a monetary fee, but rather the cognitive burden of recognizing the combined exposure, deciding how to offset it, and executing the reallocation.

Beyond cognitive frictions, substantial industry practices and structural incentives also support goal-based segregation. While a frictionless, unified portfolio is theoretically optimal for pure risk-adjusted returns, managing separate portfolios provides distinct practical advantages. Different goals possess heterogeneous characteristics, including distinct horizons, loss tolerances, tax treatments, and feasible investment sets. Earmarked accounts, such as Health Savings Accounts for medical expenses or 529 plans for education, enforce financial discipline and protect funds from being diverted. Consequently, the practical trade-off is that separate accounts may sacrifice some diversification efficiency in order to improve salience, commitment, and goal discipline.

\subsection{Model setup}
Motivated by the empirical evidence above, we consider separate portfolios for different goals.} Suppose the investor has $K + 1$ goals. Each goal $k \in \{1, \ldots, K+1\}$ requires a target wealth amount $G_k$ by a deterministic deadline $T_k$. For simplicity, suppose the deadlines are different and sorted as $T_1 < \ldots < T_k < \ldots < T_{K+1}$. For notational convenience, denote $T_0 := 0$ and $T := T_{K+1}$. Hence, the investment problem spans the time horizon $[0, T]$. The investor allocates the budget across portfolios $X_k$ for each goal $G_k$. The last goal $K+1$, referred to as the fundamental goal or portfolio, addresses essential needs or bequests and has the longest deadline, distinguishing it from the other $K$ goals. 

Let $(\Omega, \cF, \p)$ be a probability space supporting an $\R^d$-valued Brownian motion $W := \{W(t) : t \in [0, T] \}$. The filtration $\mathbb{F} := \{\cF_t : t \in [0, T] \}$ denotes the completion of the natural filtration of the Brownian motion $W$, and $\mathbb{F}$ satisfies the usual conditions. Consider a financial market with a risk-free asset and $n$ risky assets (stocks). Denote $r$ as the constant risk-free interest rate. Stock prices $\{ S(u): u \in [t, T]\}$ are driven by an $m$-dimensional factor process $\{ Y(u) : u \in [t, T]\}$. The dynamics of $S$ and $Y$ are given by
\begin{align}
	d Y(u) & = \mu_Y(Y(u)) du + \sigma_Y (Y(u)) dW(u), \quad Y(t) = y, \label{Y} \\
	d S(u) & = \text{diag}(S(u)) [\mu(Y(u)) du + \sigma(Y(u)) dW(u)]. \label{stock}
\end{align}
The factor process $Y$ has a drift coefficient $\mu_Y: \R^m \rightarrow \R^m$ and diffusion coefficient $\sigma_Y: \R^m \rightarrow \R^{m \times d}$. The stock price process $S$ has a drift $\mu: \R^m \rightarrow \R^n $ and volatility $\sigma: \R^m \rightarrow \R^{n \times d}$.

The wealth process for goal $k$ is denoted by $X_k$. In accordance with the financial market model \eqref{Y}-\eqref{stock}, the wealth process $X_{K+1}(u)$ for the fundamental portfolio $K+1$ follows   
\begin{equation}\label{Xlast}
	\begin{aligned}
		d X_{K+1}(u) = & r X_{K+1}(u) du + (\mu(Y(u)) - r \mathbf{1})^\top \alpha_{K+1}(u) X_{K+1}(u) du \\
		& + X_{K+1}(u) \alpha^\top_{K+1}(u) \sigma(Y(u)) dW(u) \\ 
		& - \sum^K_{k=1} {\bf 1}_{ \{ u \leq T_k\}} d L_k(u) + \sum^K_{k=1} {\bf 1}_{ \{ u \leq T_k\}} d M_k (u), \quad u \in[t, T], \\
		X_{K+1}(t-) = & x_{K+1} \geq 0.
	\end{aligned}
\end{equation}
For the remaining portfolios with $k = 1, \ldots, K$, the wealth process $X_k(u)$ follows
\begin{equation}\label{Xk}
	\begin{aligned}
		d X_k (u)  = & r X_k(u) du +  (\mu(Y(u)) - r \mathbf{1})^\top \alpha_k(u) X_k(u) du + X_k(u) \alpha^\top_k(u) \sigma(Y(u)) dW(u) \\
		& + d L_k(u) - d M_k(u), \quad u \in[t, T_k], \\
		X_k(t-) =& x_k \geq 0 . 
	\end{aligned}
\end{equation}

The $n$-dimensional vector $\alpha_k$ represents the proportion of $X_k$ invested in the stocks. The process $L_k$ tracks the cumulative amount of money withdrawn from portfolio $X_{K+1}$ to $X_k$ for goal $k$, while $M_k$ records the cumulative amount transferred back from portfolio $X_k$ to $X_{K+1}$. For simplicity, we assume that the investor can only transfer funds between the fundamental portfolio $X_{K+1}$ and the portfolio $X_k$ to adjust the budget for different goals. {\color{black}  Admissible controls are defined as follows.

\begin{definition}[Admissible controls]\label{def:admissible}
	\begin{enumerate}[label={(\arabic*)}]
		\item For each $k=1, \ldots, K+1$, the control process $\{\alpha_k(u): u \in[0, T_k]\}$ is progressively measurable. No short selling or borrowing is allowed; that is, $\alpha_k(u)$ takes values in the $n$-dimensional simplex $\cA$ a.s. 
		
		\item For each $k=1, \ldots, K$, the processes $(L_k, M_k)$ are right-continuous with left limits (RCLL), nonnegative, nondecreasing, and $\mathbb{F}$-adapted. Both are initialized as $L_k(0-) = M_k(0-) = 0$, meaning immediate transfers at time $0$ are permitted. The portfolio $X_k$, processes $L_k$ and $M_k$ exist over the interval $[0, T_k]$.
		
		\item A no-bankruptcy condition is enforced on the state processes, meaning $X_k(u) \geq 0$, a.s., for all $u \in [0, T_k]$ and $k = 1, \ldots, K+1$.
	\end{enumerate} 
	The conditions above define admissibility from the initial time $0$. When the current data are $(t, x_{k:K+1}, y)$, where $t \in [T_{k-1}, T_k]$ and $x_{k:K+1} := (x_{k}, \ldots, x_{K+1})$, we denote $\cU(t, x_{k: K+1}, y)$ as the set of admissible controls $(\alpha_{k:K+1}, L_{k:K}, M_{k:K})$ on $[t, T]$ that satisfy the analogous conditions above (e.g., initialized at $L_k(t-) = M_k(t-) = 0$) and ensure the no-bankruptcy condition given $X_{k:K+1}(t-) = x_{k:K+1}$.
\end{definition}	

Note that the set $\cU(T_k, x_{k: K+1}, y)$ contains admissible controls $(\alpha_{k:K+1}, L_{k:K}, M_{k:K})$, whereas $\cU(T_k, x_{k+1: K+1}, y)$ includes only $(\alpha_{k+1:K+1}, L_{k+1:K}, M_{k+1:K})$. Specifically, $\cU(T_k, x_{k:K+1}, y)$ represents the control set immediately prior to the expiration of goal $k$ and accommodates final transfers $(\Delta L_k(T_k), \Delta M_k(T_k))$ for goal $k$. In contrast, $\cU(T_k, x_{k+1: K+1}, y)$ is utilized right after goal $k$ expires. This distinction is critical when defining the value functions $V_k(T_k, x_{k: K+1}, y)$ and $V_{k+1}(T_k, x_{k+1: K+1}, y)$ later in \eqref{eq:val_func}.

}

 If the investor prioritizes essential goals, such as tuition fees, over non-essential ones like vacations, reallocating funds from the tuition-fee portfolio to the vacation portfolio should be penalized more heavily than the reverse. This leads to the following objective at time $0$:
\begin{align}
	\inf_{\substack{(\alpha_{1:K+1}, L_{1:K}, M_{1:K}) \\ \in \cU(0, x_{1: K+1}, y)} }\E \Big[ &\sum^{K+1}_{k=1} w_k e^{-\beta T_k}(G_k - X_k(T_k))^+ \label{obj0} \\
	& + \sum^K_{k=1} \lambda_k \int^{T_k}_0 e^{-\beta t} d L_k(t) + \sum^K_{k=1} \theta_k \int^{T_k}_0 e^{-\beta t} d M_k(t) \Big]. \nonumber
\end{align}  
 Here, $\beta \geq 0$ is the discount rate, and $w_k \geq 0$ represents the relative importance of each goal. As a benchmark level, the weight for goal $K+1$ is set as $w_{K+1} =1.0$. Constants $\lambda_k, \theta_k \geq 0$ represent mental costs that penalize transfers between portfolios. We focus solely on the linear preference \eqref{obj0} and omit the all-or-nothing preference in \cite{capponi2024}. The linear preference \eqref{obj0} is consistent with dimensional analysis since all terms are in dollars. In fact, our technical results rely only on the continuity and boundedness of the linear preference. We omit potential extensions to more general preferences for simplicity. Other relevant problems and formulations include shortfall risk \citep{cvitanic1999dynamic,cvitanic2000SICON,pham2002shortfall}, lifetime shortfall \citep{bayraktar2009shortfall}, lifetime minimum wealth \citep{bayraktar2007minwealth}, bequest goals \citep{bayraktar2016bequest}, survival problems \citep{browne1997survival}, and reaching goals by a deadline \citep{browne1999deadline}.

During the time $t \in [T_{k-1}, T_{k}]$ with $k = 1, \ldots, K+1$, the value function is defined as
\begin{equation}\label{eq:val_func}
	\begin{aligned}
	V_k(t, x_{k:K+1}, y) := \inf_{\substack{ (\alpha_{k:K+1}, L_{k:K}, M_{k:K} ) \\ \in \cU(t, x_{k:K+1}, y) } }  \E \Big[& \sum^{K+1}_{i = k} w_i e^{-\beta (T_i-t)} (G_i - X_i(T_i))^+  \\
	 & + \sum^K_{i=k} \lambda_i \int^{T_i}_t e^{-\beta(s-t)} d L_i(s) + \sum^K_{i=k} \theta_i \int^{T_i}_t e^{-\beta (s-t)} d M_i(s) \\
	&\Big| X_{k:K+1}(t-) = x_{k:K+1}, Y(t) = y \Big].
	\end{aligned}
\end{equation}
The value function $V_k(t, x_{k:K+1}, y)$ applies when goals $k, \ldots, K+1$ are active. In particular, when defining $V_k(T_{k-1}, x_{k:K+1}, y)$ using \eqref{eq:val_func}, the condition $X_{k:K+1}(T_{k-1} -) = x_{k:K+1}$ should be interpreted as transfers between $X_{k:K+1}$ at $T_{k-1}$ are allowed, while no transfers involving $X_{k-1}$ occur, as goal $k-1$ has expired. In contrast, for $V_{k-1}(T_{k-1}, x_{k-1:K+1}, y)$, transfers between $X_{k-1:K+1}$ at $T_{k-1}$ are allowed.

In the last period $t \in [T_K, T]$, we interpret $L_{k:K}$ and $M_{k:K}$ as null for $k=K+1$, where only the last goal $K+1$ remains active. In this case, the value function simplifies to
\begin{align*}
	V_{K+1}(t, x_{K+1}, y)  := \inf_{\alpha_{K+1} \in \cA} \E \Big[ & e^{-\beta(T-t)}(G_{K+1} - X_{K+1}(T))^+ \Big| X_{K+1}(t) = x_{K+1}, Y(t) = y \Big].
\end{align*}

We impose the following standing assumption throughout the paper. For technical simplicity, we do not pursue the most general conditions in this paper.
\begin{assumption}\label{stand_assum}
	The coefficients $\mu_Y(\cdot)$ and $\sigma_Y(\cdot)$ of the stochastic factors are Lipschitz continuous. The coefficients  $\mu(\cdot)$ and $\sigma(\cdot)$ of the stock prices are bounded Lipschitz. The variance matrix $\sigma(y) \sigma(y)^\top$ is invertible, and the inverse is uniformly bounded for all $y \in \R^m$. The number of stocks does not exceed the number of Brownian motions, i.e., $n \leq d$. 
	
	For each goal $k$, the required cash amount $G_k$ is constant and thus bounded. The mental cost parameters satisfy $\lambda_k \geq 0 $, $\theta_k \geq 0$, with $\lambda_k + \theta_k > 0$.
\end{assumption}

{\color{black} Under Definition \ref{def:admissible} and Assumption \ref{stand_assum}, the total wealth process $X_{\text{total}} := \sum_{i=1}^{K+1} X_i$ satisfies $\E[\sup_u |X_{\text{total}}(u)|^2] < \infty$. Since $0 \le X_k(u) \le X_{\text{total}}(u)$, the individual state processes are also square-integrable.
	

The condition $\lambda_k+\theta_k>0$ rules out frictionless round-trip transfers between $X_k$ and $X_{K+1}$. Economically, it ensures that moving wealth out of and back into a goal account involves a strictly positive mental cost. Without this condition, the two accounts could be instantaneously rebalanced at zero net cost, making the separate state variable $x_k$ redundant. This condition also has a technical role. The strict inequality is used in the construction of the strict classical subsolution in Lemma~\ref{lem:strict_l}, which is a key step in the comparison argument and hence in the uniqueness of the constrained viscosity solution.

Before proceeding to the technical analysis, we briefly discuss the practical calibration of the model inputs. First, modern wealth management platforms usually require investors to specify concrete target amounts and deadlines. Although educational and retirement goals possess naturally fixed horizons, future extensions could incorporate random deadlines to capture more flexible objectives. Second, prioritization weights can be elicited through structured questionnaires. By establishing the fundamental portfolio as a baseline benchmark, financial advisors can utilize preference-ranking tools to evaluate the relative importance of secondary goals. Finally, mental cost parameters could potentially be reverse-engineered from observed behavioral data, such as the implicit frictions preventing optimal transfers documented in \cite{choi2009mental}. Because psychological frictions fluctuate in response to platform design or market stress, the dynamic calibration of these parameters is also an important topic for future research.
}

\section{The HJB equation}\label{sec:HJB}
The classical Merton's portfolio problem typically considers a single wealth variable. An exception is \cite{deelstra2001dual}, which maximizes the expected utility from terminal wealth in a multivariate financial market with transaction costs. Another example is \cite{belak2022optimal}, where two wealth variables represent a money market account and a stock account, respectively. In our framework, mental accounting naturally leads to a multivariate formulation. A distinctive feature of our problem is that the value function loses a variable $x_k$ when goal $k$ expires. The value function can be viewed as an array:
\begin{equation}\label{value}
	(\{ V_1(t, x_{1:K+1}, y) \}_{t \in [0, T_1]}, \ldots, \{ V_k(t, x_{k:K+1}, y) \}_{t \in [T_{k-1}, T_{k}]}, \ldots, \{ V_{K+1}(t, x_{K+1}, y) \}_{t \in [T_K, T]} ).
\end{equation}
At the deadline $T_k$, both $V_k(T_k, x_{k:K+1}, y)$ and $V_{k+1}(T_k, x_{k+1:K+1}, y)$ are defined, representing the optimal value of the objective right before and after the deadline $T_k$, respectively.

For $k=1, \ldots, K+1$, we define the following regions of the state space. The closed state space is denoted by $\overline{\cS}_k := [0, \infty)^{K-k+2} \times \R^m$, and its open interior by $\cS_k := (0, \infty)^{K-k+2} \times \R^m$. The set $\overline{\cS}^o_k: = \{ [0, \infty)^{K-k+2} \backslash \{0\} \} \times \R^m$ excludes the corner at $0$ for $x_{k:K+1}$. 

Through a heuristic derivation, we obtain the HJB equation system as follows:
\begin{enumerate}[label={(\arabic*)}]
	\item For time $t \in [T_{k-1}, T_{k})$ with $k = 1, \ldots, K$, the goals $k, \ldots, K+1$ are active. The Hamiltonian $H$ is defined as 
	\begin{equation}
		\begin{aligned}
		& H(x_{k:K+1}, y, \partial V_k, \partial^2 V_k) \\
		& \;  = \sum^{K+1}_{i=k} r x_{i} \partial_i V_k + \mu_Y(y)^\top \partial_y V_k  \label{H}  \\
		& \quad + \inf_{\alpha_{k:K+1} \in \cA^{K-k+2}} \Big\{  \sum^{K+1}_{i=k}  (\mu(y) - r \one )^\top \alpha_i x_i \partial_i V_k + \frac{1}{2} \tr\left[\Sigma(\alpha_{k:K+1}, x_{k:K+1}, y) \partial^2 V_k  \right]  \Big\}.
	    \end{aligned}
	\end{equation}

	In the Hamiltonian $H$, $\partial V_k$ and $\partial^2 V_k$ represent the first and second partial derivatives with respect to the state variables $(x_{k:K+1}, y)$, respectively. Specifically, $\partial V_k = (\partial_k V_k, \ldots, \partial_{K+1} V_k,$ $ \partial_y V_k)$. Here, $\partial_i V_k$ denotes the first partial derivative of $V_k$ with respect to $x_i$, and $\partial_y V_k$ is the $m$-dimensional partial derivative with respect to $y$. Let $\partial^2 V_k$ be the $(K-k+2+m) \times (K-k+2+m)$-dimensional Hessian matrix, where $(x_{k:K+1}, y)$ are the variables. The matrix $\Sigma$ in the Hamiltonian $H$ is given by
	\begin{align}
		& \Sigma(\alpha_{k:K+1}, x_{k:K+1}, y) \\
		& \quad := \begin{pmatrix}
			(\alpha x)^\top_{k:K+1} \sigma(y) \sigma(y)^\top (\alpha x)_{k:K+1} & (\alpha x)^\top_{k:K+1} \sigma(y) \sigma_Y(y)^\top \\
			\sigma_Y(y) \sigma(y)^\top (\alpha x)_{k:K+1} & \sigma_Y(y) \sigma_Y(y)^\top
		\end{pmatrix}, \nonumber
	\end{align} 
	where
	\begin{equation}
		(\alpha x)_{k:K+1} := (\alpha_{k} x_{k}, \ldots, \alpha_{K+1} x_{K+1}) \in \R^{n \times (K-k+2)}.
	\end{equation}
	
	Define the operator $F$ as follows:
	\begin{equation}
		\begin{aligned}
		& F(t, x_{k:K+1}, y, V_k, \partial_t V_k, \partial V_k, \partial^2 V_k) \\
		&\quad  := \max \Big\{  \beta V_k - \partial_t V_k - H(x_{k:K+1}, y, \partial V_k, \partial^2 V_k), \\
		&\quad \qquad \qquad - \lambda_{k} + \partial_{K+1} V_k - \partial_{k} V_k, \ldots, - \lambda_{K} + \partial_{K+1} V_k - \partial_{K} V_k, \\
		&\quad \qquad \qquad - \theta_{k} - \partial_{K+1} V_k + \partial_{k} V_k, \ldots, - \theta_K - \partial_{K+1} V_k + \partial_K V_k \Big\}.
		\end{aligned}
	\end{equation}
	
	Then the HJB equation with $t \in [T_{k-1}, T_{k})$ is
	\begin{equation}\label{F}
		F(t, x_{k:K+1}, y, V_k, \partial_t V_k, \partial V_k, \partial^2 V_k)  = 0.
	\end{equation}

	\item At time $T_k$, the boundary condition that connects the value function $V_k(T_k, x_{k:K+1}, y)$ with $V_{k+1}(T_k, x_{k+1:K+1}, y)$ is given by
	\begin{equation}\label{Tk_bc}
		\begin{aligned}
		\max \Big\{ & V_k(T_k, x_{k:K+1}, y) - w_k (G_k - x_k)^+ - V_{k+1}(T_k, x_{k+1:K+1}, y), \\
		& - \lambda_{k} + \partial_{K+1} V_k - \partial_{k} V_k, \ldots, - \lambda_{K} + \partial_{K+1} V_k - \partial_{K} V_k, \\
		& -\theta_{k} - \partial_{K+1} V_k + \partial_k V_k, \ldots, - \theta_K - \partial_{K+1} V_k + \partial_{K} V_k \Big\} = 0.
		\end{aligned}
	\end{equation}
	
	{\color{black} The first term in \eqref{Tk_bc} corresponds to the scenario where no instantaneous fund transfer is optimal. When this term is binding, the value just before the deadline, $V_k$, equals the post-deadline value, $V_{k+1}$, plus the penalty for any shortfall in goal $k$. The subsequent gradient terms define the intervention regions. If any of these terms equals zero, it indicates that an immediate transfer is optimal. Hence, \eqref{Tk_bc} ensures that the value function captures the choice between maintaining the current allocation and executing a final transfer exactly at the deadline. This intuitive explanation is rigorously verified in the subsequent proof.}
	
	\item When the time $t \in [T_K, T]$, only the last goal $K+1$ is active. Then the HJB equation for $V_{K+1}(t, x_{K+1}, y)$ is classical:
	\begin{equation}\label{Vlast}
		\begin{aligned}
		&\beta V_{K+1}(t, x_{K+1}, y) - \partial_t V_{K+1}(t, x_{K+1}, y) - H(x_{K+1}, y, \partial V_{K+1}, \partial^2 V_{K+1}) = 0, \\
		&  V_{K+1}(T, x_{K+1}, y) = (G_{K+1} - x_{K+1})^+.
		\end{aligned}
	\end{equation}
	
	\item When the wealth is zero for every portfolio, the boundary condition is given by
	\begin{equation}
		V_k(t, 0, y) = \sum^{K+1}_{i = k} w_i e^{-\beta (T_i-t)} G_i.
	\end{equation}
\end{enumerate}

Since the system described above is the only HJB equation system in this paper, we simply refer to it as {\it the HJB equation system}.

The main result of this paper characterizes the value function defined in \eqref{value} as the unique viscosity solution of the HJB equation system. We use the following notation from the theory of viscosity solutions. Given a locally bounded function $v_k(t, x_{k:K+1}, y): [T_{k-1}, T_k] \times \overline{\cS}_k \rightarrow \R$, the upper semicontinuous (USC) envelope of $v_k$ is defined as
\begin{equation*}
	v^*_k(t, x_{k:K+1}, y) := \limsup_{\substack{(s, z, w) \rightarrow (t, x_{k:K+1}, y) \\ (s, z, w) \in [T_{k-1}, T_k] \times \overline{\cS}_k} } v_k(s, z, w).
\end{equation*}
The lower semicontinuous (LSC) envelope of $v_k$ is defined as
\begin{equation*}
	v_{k, *}(t, x_{k:K+1}, y) := \liminf_{\substack{(s, z, w) \rightarrow (t, x_{k:K+1}, y) \\ (s, z, w) \in [T_{k-1}, T_k] \times \overline{\cS}_k} } v_k(s, z, w).
\end{equation*}
\cite{bentahar2007} defines the envelopes using interior points $(s, z, w) \in [T_{k-1}, T_k) \times \cS_k$. Our definition makes the proof easier with the stochastic envelopes $v_{k, +}$ and $v_{k, -}$ introduced later.  

 
\begin{definition}[Viscosity subsolution]\label{def:vis_sub} 
	Consider an array of functions
	\begin{equation}\label{vis_sub}
		(\{ v_1(t, x_{1:K+1}, y) \}_{t \in [0, T_1]}, \ldots, \{ v_k(t, x_{k:K+1}, y) \}_{t \in [T_{k-1}, T_{k}]}, \ldots, \{ v_{K+1}(t, x_{K+1}, y) \}_{t \in [T_K, T]} ),
	\end{equation}
 where $v_k(t, x_{k:K+1}, y): [T_{k-1}, T_k] \times \overline{\cS}_k \rightarrow \R, k=1,\ldots, K+1$ is locally bounded. The array \eqref{vis_sub} is a viscosity subsolution of the HJB equation system if 
\begin{enumerate}[label={(\arabic*)}]	
	\item for each $k=1,\ldots, K$,
	\begin{equation}\label{F_sub}
		\begin{aligned}
		F(& \bar{t}, \bar{x}_{k:K+1}, \bar{y}, v^*_k(\bar{t}, \bar{x}_{k:K+1}, \bar{y}), \\
		  &\partial_t \varphi(\bar{t}, \bar{x}_{k:K+1}, \bar{y}), \partial \varphi(\bar{t}, \bar{x}_{k:K+1}, \bar{y}), \partial^2 \varphi(\bar{t}, \bar{x}_{k:K+1}, \bar{y})) \leq 0
		\end{aligned} 
	\end{equation}
	for all $(\bar{t},\bar{x}_{k:K+1}, \bar{y}) \in [T_{k-1}, T_k) \times \cS_k$ and for all $\varphi \in C^{1, 2}([T_{k-1}, T_k) \times \cS_k)$ such that $(\bar{t},\bar{x}_{k:K+1}, \bar{y})$ is a maximum point of $v^*_k - \varphi$;
	
	\item for each $T_k$ with $k=1,\ldots, K$,
		\begin{equation}
			\begin{aligned}
				\max \Big\{ & v^*_k(T_k, \bar{x}_{k:K+1}, \bar{y}) - w_k (G_k - \bar{x}_k)^+ - v^*_{k+1}(T_k, \bar{x}_{k+1:K+1}, \bar{y}), \\
				& - \lambda_{k} + \partial_{K+1} \varphi(\bar{x}_{k:K+1}, \bar{y}) - \partial_{k} \varphi(\bar{x}_{k:K+1}, \bar{y}), \ldots, \\
				& - \lambda_{K} + \partial_{K+1} \varphi(\bar{x}_{k:K+1}, \bar{y}) - \partial_{K} \varphi(\bar{x}_{k:K+1}, \bar{y}), \\
				& -\theta_{k} - \partial_{K+1} \varphi (\bar{x}_{k:K+1}, \bar{y}) + \partial_k \varphi(\bar{x}_{k:K+1}, \bar{y}), \ldots, \\
				&- \theta_K - \partial_{K+1} \varphi(\bar{x}_{k:K+1}, \bar{y}) + \partial_{K} \varphi (\bar{x}_{k:K+1}, \bar{y}) \Big\} \leq 0
			\end{aligned}
		\end{equation}
		for all $(\bar{x}_{k:K+1}, \bar{y}) \in \cS_k$ and for all $\varphi \in C^{2}(\cS_k)$ such that $(\bar{x}_{k:K+1}, \bar{y})$ is a maximum point of $v^*_k(T_k, \cdot, \cdot) - \varphi$;
		\item 
		\begin{equation}\label{Vlast_sub}
			\begin{aligned}
				&\beta v^*_{K+1}(\bar{t},\bar{x}_{K+1}, \bar{y}) - \partial_t \varphi(\bar{t},\bar{x}_{K+1}, \bar{y}) \\
				&\quad - H(\bar{x}_{K+1}, \bar{y}, \partial \varphi(\bar{t},\bar{x}_{K+1}, \bar{y}), \partial^2 \varphi(\bar{t},\bar{x}_{K+1}, \bar{y})) \leq 0,
			\end{aligned}
		\end{equation}
		for all $(\bar{t},\bar{x}_{K+1}, \bar{y}) \in [T_K, T) \times \cS_{K+1}$ and for all $\varphi \in C^{1, 2}([T_K, T) \times \cS_{K+1})$ such that $(\bar{t},\bar{x}_{K+1}, \bar{y})$ is a maximum point of $v^*_{K+1} - \varphi$;
		\item
		 \begin{align}
			v^*_k(t, 0, y) & \leq \sum^{K+1}_{i = k} w_i e^{-\beta (T_i-t)} G_i, \quad \forall\;  t \in [0, T], y \in \R^m, k = 1, \ldots, K+1, \label{vissub_bd0}\\
			v^*_{K+1}(T, x_{K+1}, y) & \leq (G_{K+1} - x_{K+1})^+, \quad \forall \; (x_{K+1}, y) \in \cS_{K+1}. \label{visub_T}
		\end{align}
\end{enumerate}
\end{definition}

\begin{definition}[Viscosity supersolution]\label{def:vis_super} 
	Consider an array of functions
	\begin{equation}\label{vis_super}
		(\{ v_1(t, x_{1:K+1}, y) \}_{t \in [0, T_1]}, \ldots, \{ v_k(t, x_{k:K+1}, y) \}_{t \in [T_{k-1}, T_{k}]}, \ldots, \{ v_{K+1}(t, x_{K+1}, y) \}_{t \in [T_K, T]} ),
	\end{equation}
	where $v_k(t, x_{k:K+1}, y): [T_{k-1}, T_k] \times \overline{\cS}_k \rightarrow \R, k=1,\ldots, K+1$ is locally bounded. The array \eqref{vis_super} is a viscosity supersolution of the HJB equation system if 
	\begin{enumerate}[label={(\arabic*)}]	
		\item for each $k=1,\ldots, K$,
		\begin{equation}\label{F_super}
			\begin{aligned}
				F(& \bar{t}, \bar{x}_{k:K+1}, \bar{y}, v_{k, *}(\bar{t}, \bar{x}_{k:K+1}, \bar{y}), \\
				&\partial_t \varphi(\bar{t}, \bar{x}_{k:K+1}, \bar{y}), \partial \varphi(\bar{t}, \bar{x}_{k:K+1}, \bar{y}), \partial^2 \varphi(\bar{t}, \bar{x}_{k:K+1}, \bar{y})) \geq 0
			\end{aligned} 
		\end{equation}
		for all $(\bar{t},\bar{x}_{k:K+1}, \bar{y}) \in [T_{k-1}, T_k) \times \overline{\cS}^o_k$ and for all $\varphi \in C^{1, 2}([T_{k-1}, T_k) \times \overline{\cS}^o_k)$ such that $(\bar{t},\bar{x}_{k:K+1}, \bar{y})$ is a minimum point of $v_{k, *} - \varphi$;
		
		\item for each $T_k$ with $k=1,\ldots, K$,
		\begin{equation}
			\begin{aligned}
				\max \Big\{ & v_{k, *}(T_k, \bar{x}_{k:K+1}, \bar{y}) - w_k (G_k - \bar{x}_k)^+ - v_{k+1, *}(T_k, \bar{x}_{k+1:K+1}, \bar{y}), \\
				& - \lambda_{k} + \partial_{K+1} \varphi(\bar{x}_{k:K+1}, \bar{y}) - \partial_{k} \varphi(\bar{x}_{k:K+1}, \bar{y}), \ldots, \\
				& - \lambda_{K} + \partial_{K+1} \varphi(\bar{x}_{k:K+1}, \bar{y}) - \partial_{K} \varphi(\bar{x}_{k:K+1}, \bar{y}), \\
				& -\theta_{k} - \partial_{K+1} \varphi (\bar{x}_{k:K+1}, \bar{y}) + \partial_k \varphi(\bar{x}_{k:K+1}, \bar{y}), \ldots, \\
				&- \theta_K - \partial_{K+1} \varphi(\bar{x}_{k:K+1}, \bar{y}) + \partial_{K} \varphi (\bar{x}_{k:K+1}, \bar{y}) \Big\} \geq 0
			\end{aligned}
		\end{equation}
		for all $(\bar{x}_{k:K+1}, \bar{y}) \in \overline{\cS}^o_k$ and for all $\varphi \in C^{2}(\overline{\cS}^o_k)$ such that $(\bar{x}_{k:K+1}, \bar{y})$ is a minimum point of $v_{k, *}(T_k, \cdot, \cdot) - \varphi$;
		\item 
		\begin{equation}\label{Vlast_super}
			\begin{aligned}
				&\beta v_{K+1, *}(\bar{t},\bar{x}_{K+1}, \bar{y}) - \partial_t \varphi(\bar{t},\bar{x}_{K+1}, \bar{y}) \\
				&\quad - H(\bar{x}_{K+1}, \bar{y}, \partial \varphi(\bar{t},\bar{x}_{K+1}, \bar{y}), \partial^2 \varphi(\bar{t},\bar{x}_{K+1}, \bar{y})) \geq 0,
			\end{aligned}
		\end{equation}
		for all $(\bar{t},\bar{x}_{K+1}, \bar{y}) \in [T_K, T) \times \overline{\cS}^o_{K+1}$ and for all $\varphi \in C^{1, 2}([T_K, T) \times \overline{\cS}^o_{K+1})$ such that $(\bar{t},\bar{x}_{K+1}, \bar{y})$ is a minimum point of $v_{K+1, *} - \varphi$;
		\item 
		\begin{align}
			v_{k, *}(t, 0, y) & \geq \sum^{K+1}_{i = k} w_i e^{-\beta (T_i-t)} G_i, \; \forall\;  t \in [0, T], y \in \R^m, k = 1, \ldots, K+1, \label{vissup_bd0} \\
			v_{K+1, *}(T, x_{K+1}, y) & \geq (G_{K+1} - x_{K+1})^+, \quad \forall \; (x_{K+1}, y) \in \cS_{K+1}. \label{vissup_T}
		\end{align}
	\end{enumerate}

\end{definition}

\begin{definition}[Constrained viscosity solution]\label{def:cons}
	Denote an array of functions as
	\begin{equation}\label{vis_sol}
		(\{ v_1(t, x_{1:K+1}, y) \}_{t \in [0, T_1]}, \ldots, \{ v_k(t, x_{k:K+1}, y) \}_{t \in [T_{k-1}, T_{k}]}, \ldots, \{ v_{K+1}(t, x_{K+1}, y) \}_{t \in [T_K, T]} ),
	\end{equation}
	where $v_k(t, x_{k:K+1}, y): [T_{k-1}, T_k] \times \overline{\cS}_k \rightarrow \R, k=1,\ldots, K+1$ is locally bounded. The array \eqref{vis_sol} is a constrained viscosity solution of the HJB equation system if it is a viscosity subsolution under Definition \ref{def:vis_sub} and a viscosity supersolution under Definition \ref{def:vis_super}.
\end{definition}

{\color{black} In a standard viscosity solution on an unconstrained domain, both the subsolution and supersolution properties are tested at every point. In a state-constrained problem, however, the boundary possesses a different status because certain controls would push the state outside the admissible domain.} In Definition \ref{def:cons}, only the viscosity supersolution property holds on the boundary $\overline{\cS}^o_{k} \backslash \cS_k$, i.e., at least one $x_i$, for $i \in \{k, \ldots, K+1\}$, is zero, and $x_{k:K+1} \neq 0$. The viscosity subsolution property holds in the interior $\cS_k$. The boundary value at the corner $x_{k:K+1} = 0$ is known. {\color{black} To motivate the constrained viscosity solutions in our setup, we consider a smooth value function with an optimal control to simplify the discussion. On the boundary $\overline{\cS}^o_k \backslash \cS_k$, one or more portfolios possess zero wealth. It is not allowed to transfer funds out of these empty portfolios. If the subsolution property were required to hold on this boundary, the maximum operator would mandate that all terms be non-positive. This condition is too strong since it would improperly test gradient constraints corresponding to boundary-violating transfers. Conversely, the supersolution property remains valid on the boundary, as it requires only that at least one term within the maximum operator be non-negative, which is guaranteed by the optimal control. For an explicit example motivating constrained viscosity solutions, interested readers are directed to \citet[p. 553]{soner1986optimal}.}

The main result of this paper is stated as follows:
\begin{theorem}\label{thm:viscosity}
	The value function array \eqref{value} is the unique constrained viscosity solution of the HJB equation system. For each $k = 1, \ldots, K+1$, $V_k(t, x_{k:K+1}, y)$ is continuous and bounded on $[T_{k-1}, T_k] \times \overline{\cS}_k$.
\end{theorem}

To motivate the technical proof of Theorem \ref{thm:viscosity}, we present the numerical analysis first.

\section{Numerical analysis}\label{sec:num}
{\color{black} The viscosity characterization provides the theoretical basis for our numerical procedure. In the finite-difference approximation, the portfolio control is recovered by optimizing the discretized Hamiltonian, while transfer regions are identified when the discrete analogues of the gradient constraints become binding. Regarding theoretical convergence, standard results typically require monotonicity. However, discretizing cross-partial derivatives generally destroys this property, leaving a rigorous convergence proof as an open problem. Following common practice, we verify convergence empirically, adopting a validation routine similar to the companion paper \cite{tahar2010merton} to \cite{bentahar2007}. Furthermore, the scheme can be interpreted through the Markov chain approximation method detailed in \citet[Chapter IX]{fleming2006}. Therefore, the numerical solution characterizes the optimal policy for a financially meaningful discrete approximation of the continuous-time model.} The main numerical observations are summarized as follows:
\begin{enumerate}[label={(\arabic*)}]	
	\item Qualitative properties of free boundaries and optimal investment strategies: Free boundaries can exhibit complex shapes, such as bulges and notches, indicating that optimal transfer decisions vary nonlinearly with wealth levels. These free boundaries are also time-dependent, reflecting the trade-off between immediate and delayed transfers. For optimal investment strategies, Table \ref{tab:Goal2} provides the optimal proportions when only one goal remains, based on the linear preference $(G_i - X_i)^+$. When both goals are active, the optimal proportion vector for one portfolio depends on the wealth levels of both portfolios. For example, the optimal proportion vector $\alpha^*_1$ for portfolio $X_1$ is a function of both wealth variables $(x_1, x_2)$ rather than $x_1$ alone.
	\item Interaction between mental accounting and diversification: Unlike previous studies, our model needs to consider diversification between portfolios in addition to diversification between stocks. When one goal is fully funded and another is not, the portfolio with enough funds can either serve as a cash reserve or remain invested in the stock market to hedge risk. The optimal strategy depends on the financial environment, particularly the correlation between stock returns. Figures \ref{longeralpha_05} and \ref{shorteralpha_05} demonstrate the case of positive correlation, while Figures \ref{longeralpha_n09} and \ref{shorteralpha_n09} show the negative correlation case.
	\item Competition between goals: When one goal is important, the agent supports it in two ways. If this goal is underfunded as the deadline approaches, the agent reallocates all available funds until the target is met. In extreme cases, the continuation region in low-wealth areas reduces to line segments. Conversely, if the important goal is adequately funded while another is not, the agent still hesitates to shift resources to the less important goal, instead continuing to invest in the important portfolio and deferring allocation until its deadline.
\end{enumerate}

\subsection{Benchmark setting}
In all numerical investigations presented herein, we consider two goals $G_1$ and $G_2$ for simplicity. The goal $G_1$ has an expiration time of $T_1 = 1.0$, while $G_2$ expires at $T = 2.0$.  We refer to $G_1$ as the shorter-term goal and $G_2$ as the longer-term goal.

Suppose the financial market consists of two stocks, and the Brownian motion $W$ is two-dimensional. Consider deterministic parameters for the mean and covariance of returns. Hence, there is no factor process $Y$. Denote $\mu$ as the two-dimensional return vector and $\sigma$ as the $2 \times 2$ volatility matrix.  The investment proportion vectors are denoted as $\alpha_1 = (\alpha_{11}, \alpha_{12})^\top$ and $\alpha_2 = (\alpha_{21}, \alpha_{22})^\top$. Here, we interpret $\alpha_{ij}$ as the proportion of wealth from portfolio $i$ invested in stock $j$. The portfolio dynamics are expressed as follows:
\begin{equation}
	\begin{aligned}
		d X_1 (t)  = & r X_1(t) dt +  (\mu - r\one_2)^\top \alpha_1(t) X_1(t) dt + X_1(t)  \alpha^\top_1(t) \sigma dW(t) \\
		& + d L(t) - d M(t), \quad t \in[0, T_1], \quad X_1(0-) = x_1 \geq 0, \\
		d X_2(t) = & r X_2(t) dt + (\mu - r\one_2)^\top \alpha_2(t) X_2(t) dt + X_2(t)  \alpha^\top_2(t) \sigma dW(t) \\ 
		& -  {\bf 1}_{ \{ t \leq T_1\}} d L(t) + {\bf 1}_{ \{ t \leq T_1\}} d M (t), \quad t \in[0, T], \quad X_2(0-) = x_2 \geq 0.
	\end{aligned}
\end{equation} 

We specify the parameters as follows:
\begin{itemize}
	\item The shorter-term goal $G_1$ needs \$5,000, and the longer-term goal $G_2$ requires \$4,000. For simplicity, the monetary amounts mentioned later are all in thousands.
	\item The risk-free rate $r=0$ and the discount rate $\beta=0$.
	\item  Let $\mu = (0.2, 0.3)^\top$ be the expected return of stocks. The volatility matrix $\sigma$ is
	\begin{align}
		& \sigma = \begin{pmatrix}
			\sigma_1 & 0  \\
			\rho \sigma_2 & \sqrt{1 - \rho^2} \sigma_2 
		\end{pmatrix}, \nonumber
	\end{align}
	obtained from the Cholesky decomposition. We set $\sigma_1 = 0.3$ and $\sigma_2 = 0.4$. The correlation is specified later.
	\item The objective is given by
	\begin{align}
		\inf_{\alpha_1, \alpha_2, L, M} \mathbb{E} \Big[ & e^{-\beta T}(G_2 - X_2(T))^+ + w_1 e^{- \beta T_1}(G_1 - X_1(T_1))^+ \nonumber \\
		& + \lambda_1 \int^{T_1}_0 e^{-\beta t} d L(t) + \theta_1 \int^{T_1}_0 e^{- \beta t} d M(t) \Big].
	\end{align}    
	Unless specified otherwise, we set $w_1 = 1$, $\lambda_1 = 0.3$, and $\theta_1 = 0.1$.
\end{itemize}

We summarize the notations used for the two goals as follows. The variables $G_1$ and $G_2$ denote the shorter-term and longer-term goals, respectively. The portfolios associated with these goals are denoted by $X_1$ and $X_2$, respectively. The variables $x_1$ and $x_2$ correspond to the wealth amounts allocated for the shorter-term and longer-term goals, respectively.

\subsection{Two-stock market with correlation $\rho=0.5$}
In this section, we set the correlation $\rho =0.5$ as our benchmark configuration.
\subsubsection{After goal 1 expires: $t \in [T_1, T]$}
The HJB equation in this period is
\begin{align*}
	&\beta V_2(t, x_2) - \partial_t V_2(t, x_2) - H(x_2, \partial_2 V_2, \partial^2_{22} V_2) = 0, \quad V_2(T, x_2) = (G_2 - x_2)^+.
\end{align*}
When short-selling and borrowing are allowed, an explicit solution exists in the single-stock case; see \citet[Example 4.1]{cvitanic1999dynamic}. However, for the constrained two-stock setting considered here, no closed-form solution is available to our knowledge. Therefore, we employ a finite difference scheme to solve the problem numerically. The time step size is $\Delta t = 0.01$, the wealth step size is $\Delta x_2 = 0.2$, and the investment proportions have a step size of $0.01$.

When wealth $x_2$ is low, the optimal strategy suggests allocating all funds to the second stock, as it yields a higher expected return compared to the first one. If wealth $x_2$ surpasses a certain critical threshold, the portfolio is adjusted to diversify investments across both stocks. With further increases in wealth, the allocation to both risky assets diminishes, as outlined in Table \ref{tab:Goal2}. Indeed, when the wealth is close to the goal amount, the agent reduces the riskiness of the portfolio to meet the goal with certainty.
\begin{table}[H]
	\centering 
	\begin{tabular}{cccccccccccc}
		\hline 
		Wealth $x_2$ &  2.0 & 2.2  &  2.4  &  2.6   &  2.8   &  3.0 &  3.2  &  3.4  &  3.6  &  3.8  &  4.0 \\
		\hline 
		Proportion in Stock 1 & 0.0  & 0.0  &  0.07 & 0.22 &  0.33 &  0.42 & 0.48 & 0.38  & 0.28  &  0.16 &  0.0    \\
		\hline 
		Proportion in Stock 2 &  1.0  & 1.0  &  0.93 &  0.78 & 0.67 & 0.58 &  0.51  & 0.41  & 0.3 &  0.18   &  0.0 \\
		\hline 
	\end{tabular}
	\caption{The optimal investment proportion $\alpha^*_2$ after the shorter-term goal expires. The dollar amounts $x_2$ are in thousands.}\label{tab:Goal2}
\end{table}

\subsubsection{When goal 1 expires: $t = T_1$}

At the deadline $T_1$, the boundary condition connecting $V_1(T_1, x_1, x_2)$ with $V_2(T_1, x_2)$ is given by
\begin{align*}
	\max \Big\{ & V_1(T_1, x_1, x_2) - w_1 (G_1 - x_1)^+ - V_2(T_1, x_2), \\
	& - \lambda_1 + \partial_2 V_1(T_1, x_1, x_2) - \partial_1 V_1(T_1, x_1, x_2), \\
	& - \theta_1 - \partial_2 V_1(T_1, x_1, x_2)  + \partial_1 V_1(T_1, x_1, x_2) \Big\} = 0.
\end{align*}

For the numerical algorithm, we adopt a penalty scheme based on the finite difference discretization and the policy iteration; see \citet[Section 2.3]{azimzadeh2017} for an introduction. The wealth step size is $\Delta x_1 = \Delta x_2 = 0.2$.

\begin{figure}
	\centering
	\includegraphics[width=0.7\textwidth]{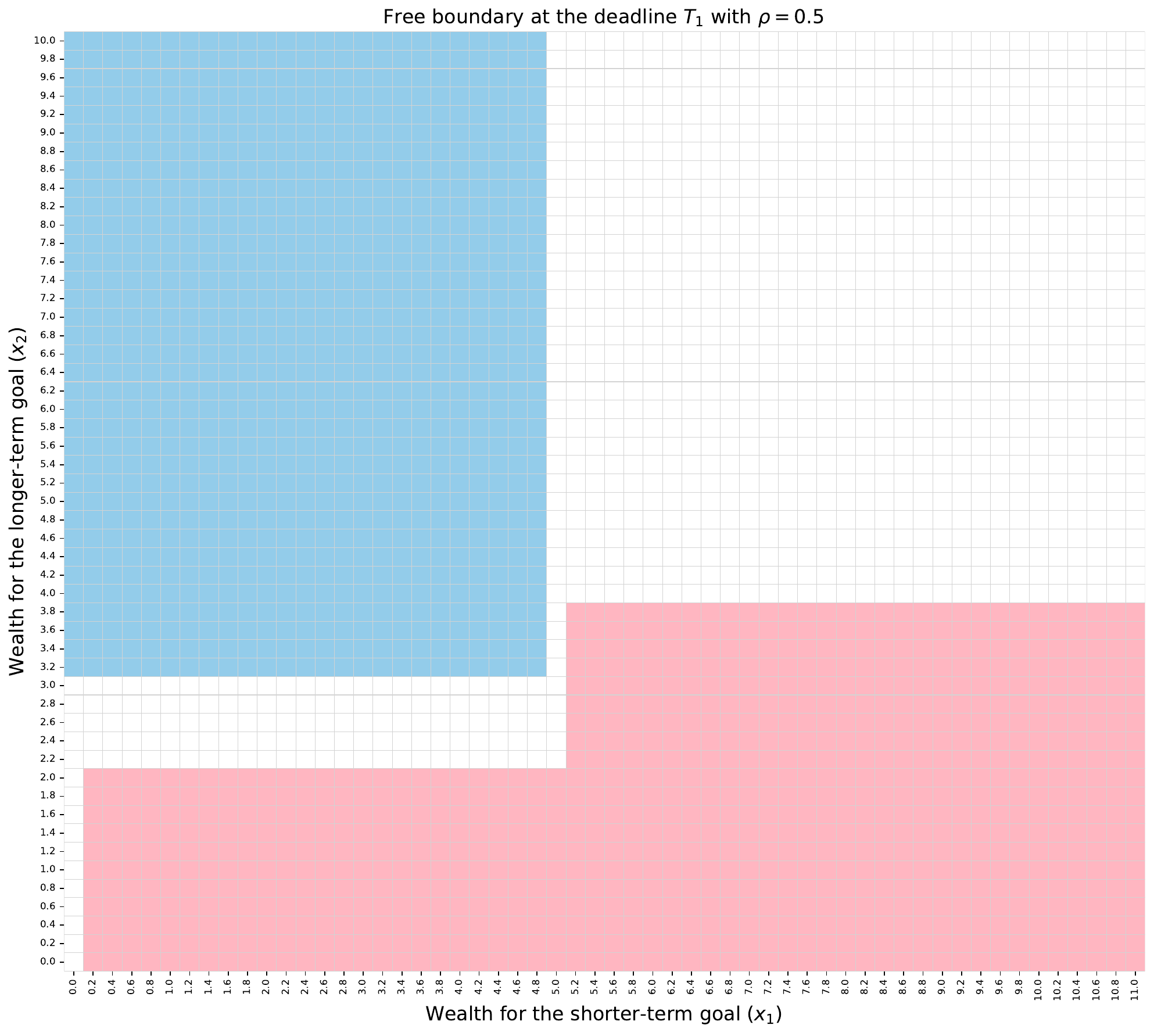}
	\caption{Free boundaries at the shorter-term goal deadline $T_1=1$. We set the correlation $\rho = 0.5$, $\lambda_1=0.3$, and $\theta_1=0.1$.  }\label{T1impulse_corr05}
\end{figure}

Figure \ref{T1impulse_corr05} illustrates the free boundary at the deadline $T_1$. The light blue region represents the transaction region from the longer-term to the shorter-term portfolio, while the light pink region indicates the reverse transaction. Since the transfer does not incur real costs, the wealth pair $(x_1, x_2)$ moves along the $(1, -1)$ or $(-1, 1)$ direction. When the longer-term wealth $x_2$ is high and the shorter-term wealth $x_1$ is low, the agent transfers from $x_2$ to $x_1$, but the longer-term wealth $x_2$ should not be lower than $3.0$. Conversely, when the longer-term wealth $x_2$ is low, the transfer from $x_1$ to $x_2$ aims to achieve the level $x_2=2.2$. Consequently, if the shorter-term wealth $x_1 \leq 2.2$, the optimal decision suggests leaving zero wealth amount in the shorter-term portfolio. Moreover, the threshold $x_2 = 2.2$ also acts as the stopping point of the transaction when the shorter-term wealth $x_1 \leq 7.2$. Instead, if $x_1 > 7.2$, the agent keeps the required amount $5.0$ for the shorter-term goal, allocating any surplus to the longer-term goal until $x_2 = 4.0$ is reached.

Importantly, the thresholds $2.2$ and $3.0$ above rely on the mental cost $\lambda_1$ and $\theta_1$. Furthermore, these thresholds are also influenced by the correlation between stocks, as evident from a comparison with Figure \ref{T1impulse_corrn09}. 

Under the current specification, the agent shows greater support for the longer-term portfolio due to two main reasons. First, the longer-term portfolio enjoys an extra year for stock investment earnings. Second, the mental cost $\theta_1$ associated with the transaction from the shorter-term to the longer-term portfolio is lower than the mental cost $\lambda_1$ for the opposite direction. The transfer from $x_1$ to $x_2$ halts until the marginal benefit of doing so becomes lower than retaining funds in $x_1$, as depicted in Figure \ref{T1impulse_corr05}.

\subsubsection{When both goals are active: $t \in [0, T_1]$}
Recall that the HJB equation is
\begin{align*}
	\max \Big\{ & \beta V_1(t, x_1, x_2) - \partial_t V_1(t, x_1, x_2) - H(x_{1:2}, \partial V_1, \partial^2 V_1), \\
	& - \lambda_1 + \partial_2 V_1 (t, x_1, x_2) - \partial_1 V_1 (t, x_1, x_2), - \theta_1 - \partial_2 V_1 (t, x_1, x_2)  + \partial_1 V_1 (t, x_1, x_2) \Big\} = 0.
\end{align*}

\begin{table}[H]
	\centering
	\small
	\begin{tabular}{c|ccccc}\hline
		\diagbox[width=10em]{Proportion \\ in Stock 2}{ Proportion in \\ Stock 1}& 0.0 & 0.25 & 0.50 & 0.75 & 1.0 \\ \hline
		0.0  &  0 & 5 &  9&  12 & 14 \\ 
		0.25  & 1 & 6 & 10  & 13  & \\ 
		0.50  & 2 & 7  & 11 &  & \\ 
		0.75  &  3& 8 &  &  & \\ 
		1.0  &  4 &  &  &  & \\ \hline
	\end{tabular}
	\caption{Strategy codes for investment proportions in stocks.}
	\label{code}
\end{table}

For the numerical algorithm, we use the penalty scheme with the wealth step size $\Delta x_1 = \Delta x_2 = 0.2$ and the time step size $\Delta t = 0.2$. The investment proportion vectors $\alpha_1$ and $\alpha_2$ are two-dimensional. Hence, to facilitate analysis, we flatten and label them as detailed in Table \ref{code}. For instance, the strategy code 4 corresponds to $(0.0, 1.0)$, indicating 100\% of the wealth invested in the second stock. In all figures of the investment strategies, such as Figures \ref{longeralpha_05} and \ref{shorteralpha_05}, values in each cell represent the codes listed in Table \ref{code}. Similar to Figure \ref{T1impulse_corr05}, the light blue and pink regions denote intervention zones where the transfer occurs and thus the investment strategy $\alpha_i$ is not binding there.

To understand the optimal investment strategies, we consider three scenarios involving different combinations of wealth levels, given the current time $t=0.8$ close to the deadline of the shorter-term goal:
\begin{enumerate}[label={(\arabic*)}]	
	\item {\bf Insufficient funds for both goals}
	
	{\bf The longer-term portfolio strategy}: Without the shorter-term goal, the optimal approach for the longer-term goal should align closely with the strategy outlined in Table \ref{tab:Goal2}. However, with two goals in mind, the longer-term portfolio strategy $\alpha^*_2$ also depends on the shorter-term wealth. As depicted in Figure \ref{longeralpha_05}, when the shorter-term wealth $x_1 \leq 2.4$, the longer-term portfolio $X_2$ allocates all wealth (Strategy Code 4) to the second stock. Only when $x_1 > 2.4$ does the longer-term portfolio $X_2$ begin to diversify across two stocks. If $x_1$ gets closer to the target level, the longer-term portfolio $X_2$ follows the strategy presented in Table \ref{tab:Goal2} more closely.
	
	An intriguing observation arises within $(x_1, x_2) \in [2.6, 3.2] \times [2.6, 2.8]$. The longer-term portfolio strategy $\alpha^*_2$ follows a pattern $(0.0, 1.0) \rightarrow (0.25, 0.75) \rightarrow (0.0, 1.0)$ (Strategy Code $4 \rightarrow 8 \rightarrow 4$) as the longer-term wealth $x_2$ increases, while the shorter-term wealth $x_1$ remains fixed. Notably, when the longer-term wealth $x_2$ increases from $2.8$ to $3.2$, the agent does not decrease the investment proportion but rather increases the weight on the second stock. This trend is particularly pronounced in the case of negative correlation $\rho = -0.9$; see Figure \ref{longeralpha_n09}. To elucidate this, we recall the strip of the continuation region depicted in Figure \ref{T1impulse_corr05}. When $x_2$ approaches $3.2$,  the longer-term portfolio $X_2$ transfers to the shorter-term one $X_1$. Observe that the time $t=0.8$ in Figure \ref{longeralpha_05} is near the deadline $T_1$. Consequently, even with a higher $x_2$, the agent should maintain an aggressive investment strategy rather than reducing investment as depicted in Table \ref{tab:Goal2}, due to the existence of another portfolio.

	{\bf The shorter-term portfolio strategy}: Based on Figure \ref{shorteralpha_05}, the shorter-term portfolio concentrates solely on the second stock initially, and diversifies only when wealth $x_1$ approaches the target level. This strategy is reasonable given the current time $t = 0.8$ is near the deadline. Therefore, the shorter-term portfolio $X_1$ is advised to pursue an aggressive investment approach.
	
	\item {\bf The longer-term goal has adequate funds, while the shorter-term goal does not}
	
	{\bf The longer-term portfolio strategy}:  As depicted in Figure \ref{longeralpha_05}, the longer-term portfolio $X_2$ avoids stock investments within this domain and acts as cash reserves. The surplus in the longer-term wealth $x_2$ is designated to support the shorter-term goal with the method elaborated in the following discussion on the bulges of intervention regions.
	
	{\bf The shorter-term portfolio strategy}: Illustrated in Figure \ref{shorteralpha_05}, when the longer-term wealth $x_2$ is high, the shorter-term portfolio follows the strategy $$(0.0, 1.0) \rightarrow (0.25, 0.75) \rightarrow (0.5, 0.5) \rightarrow (0.25, 0.5) \rightarrow (0.0, 0.0),$$ 
	or equivalently, $4 \rightarrow 8 \rightarrow 11 \rightarrow 7 \rightarrow 0$ in the strategy code. This approach mirrors the one delineated in Table \ref{tab:Goal2}. Intuitively, the agent manages the shorter-term portfolio as if it were the sole portfolio, utilizing the longer-term portfolio as a cash reservoir.

	Around the state $(x_1, x_2) = (4.0, 4.6)$, the shorter-term portfolio avoids aggressive strategies such as $\alpha_1 = (0.0, 1.0)$. Given that the shorter-term portfolio will receive funds from the longer-term counterpart, stock investments for $X_1$ adopt a conservative position. This behavior aligns with findings from \cite{capponi2024}, in which the optimal strategy becomes conservative when wealth approximates the target, but more aggressive under higher wealth levels.
	
	\item {\bf The shorter-term goal has adequate funds, while the longer-term goal does not} 
	
	This scenario closely mirrors the one described above, while the behaviors for the shorter-term and longer-term portfolios are interchanged.
\end{enumerate}

For the free boundaries, the main observation is that there exist two bulges in the intervention regions. We examine one of them around $(x_1, x_2) \in [3.4, 4.6] \times [3.8, 5.6]$. The other bulge can be interpreted similarly.
\begin{itemize}
	\item When $(x_1, x_2) \in [3.4, 4.6] \times [3.8, 5.6]$, it means the shorter-term wealth $x_1$ approaches the required level of $5.0$, while the longer-term wealth $x_2$ slightly exceeds the mandated level of $4.0$. The optimal control suggests transferring to the shorter-term goal. Both portfolios experience a reduction in stock investment within this region. Therefore, to achieve the shorter-term goal, it relies on wealth from the longer-term portfolio rather than investment returns from stocks. This strategy aims to mitigate risk and circumvent potential failure of the shorter-term goal, after taking the associated mental costs into account.
	
	It is crucial to highlight that the bulge phenomenon hinges on the correlation between two risky assets, with the current setting being $\rho = 0.5$. Consequently, the diversification benefit is limited if the surplus remains in the longer-term portfolio $X_2$, instead of transferring to portfolio $X_1$. We will later demonstrate that an extremely negative correlation eliminates the bulge, as the diversification benefit dominates; see Figure \ref{longeralpha_n09} and \ref{shorteralpha_n09}.
	
	\item In the area around $(x_1, x_2) \in [3.4, 4.8] \times [5.6, 10.0]$, we adopt a more ``patient" approach to support the shorter-term goal by waiting for the free boundary to shift rightward.
	\begin{itemize}
		\item[(a)]	If the shorter-term portfolio exhibits positive performance, causing an increase in $x_1$ such that the free boundary cannot catch up, we avoid transactions and mitigate associated mental costs.
		\item[(b)]  Conversely, if the shorter-term portfolio $X_1$ performs poorly, we support $x_1$ back to the free boundary by reallocating wealth from $x_2$. In cases of persistent underperformance, we transfer progressively to support the shorter-term portfolio $X_1$ towards the target level.
	\end{itemize}
\end{itemize}

Consider a scenario where both portfolios originate from regions of low wealth, such as $(x_1, x_2) = (1.4, 1.4)$. In this case, the wealth pair will move within the continuation region and rebound against free boundaries. Should both portfolios exhibit strong performance, they will traverse the ``strait" delineated by the two bulges and achieve both goals. It is less likely for one portfolio to substantially exceed the target level while the other remains unfunded.

The bulge phenomenon also aligns with the main finding in \cite{capponi2024}, where the agent reduces investment risk when approaching a goal and increases the risk when resources are more abundant. Differently, our framework unveils a more intricate interaction between mental accounting and diversification.

\begin{figure}[H]
	\centering
	\includegraphics[width=0.95\textwidth]{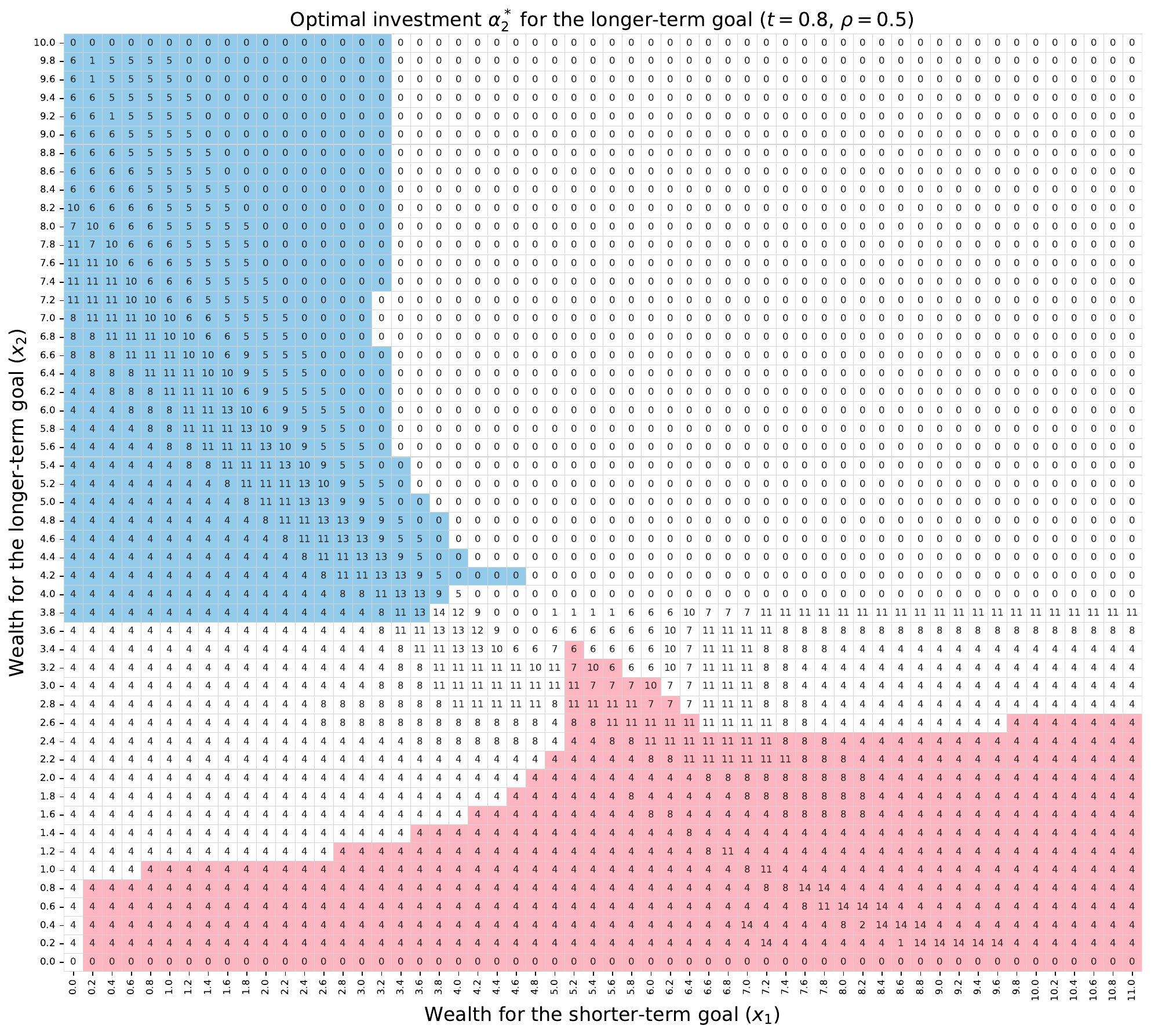}
	\caption{The optimal investment proportion $\alpha^*_2$ for the longer-term goal at time $t=0.8$. We set the correlation $\rho=0.5$ here.}\label{longeralpha_05}
\end{figure}

\begin{figure}[H]
	\centering
	\includegraphics[width=0.95\textwidth]{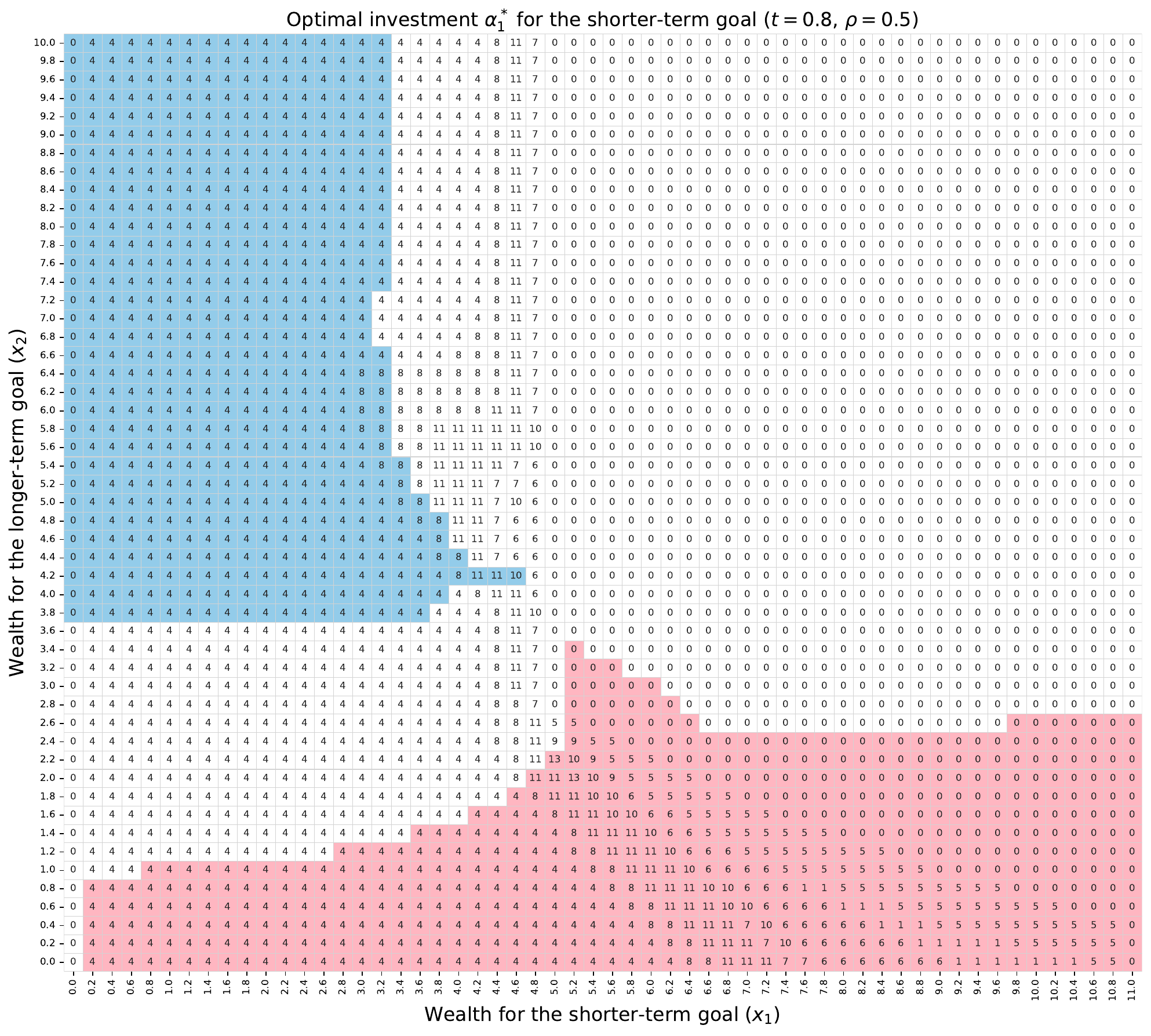}
	\caption{The optimal investment proportion $\alpha^*_1$ for the shorter-term goal at time $t=0.8$. We set the correlation $\rho=0.5$ here.}\label{shorteralpha_05}
\end{figure}

\subsection{Two-stock market with the correlation $\rho=-0.9$}

Commonly, diversification can reduce portfolio risk more effectively under a negative correlation. When we set $\rho=-0.9$ instead, it generates several effects as follows:
\begin{enumerate}[label={(\arabic*)}]	
	\item In the present configuration, a negative correlation yields a decline in the value function, especially when $t=0$. This decrease is observed in both the continuation and intervention regions, suggesting that the reductions stem from mitigating investment risks via diversification and diminishing mental costs of transactions between portfolios.
	\item A comparison between Figure \ref{T1impulse_corr05} and \ref{T1impulse_corrn09} reveals changes in the upper bound of the light pink region. When the shorter-term wealth $x_1 \leq 5.0$, the upper bound of the light pink region increases from $2.2$ to $2.6$, while it decreases from $4.0$ to $3.6$ when $x_1 > 5.0$. These changes can be interpreted as follows. At the deadline $T_1$, if both goals are unfunded, then it is better to transfer some money to the longer-term portfolio, which could still yield profit from the stock investment. Since a negative correlation makes the stock investment more attractive, we transfer more when both portfolios are in the low wealth region. Next, consider scenarios where funds for the shorter-term goal are abundant. With a negative correlation, stock investment risks decrease significantly. Investment returns can achieve the required level when the longer-term portfolio exceeds certain critical levels. Consequently, taking the mental costs into account, the transfer is not needed and the shorter-term goal may retain more funds than necessary.
	
	Under the current conditions, the light blue region remains unchanged, although this is not always the case when parameters are different.
	
	\item In Figure \ref{longeralpha_n09} or \ref{shorteralpha_n09}, the bulges disappear. To explain this, we note that the relevant terms in the negative Hamiltonian $-H$ are $- \rho \alpha_{i, 1} \alpha_{j, 2} \sigma_1 \sigma_2 x_i x_j \partial^2_{ij} V_1$, where $i, j \in \{1, 2\}$ are portfolio identifiers. Under the current parameters, we numerically find that $-H$ becomes larger with a negative correlation, leading the continuous HJB equation to be satisfied with equality.
	
	\item When an individual portfolio surpasses the target level, it continues to invest in the stock market, preparing for another unfunded portfolio. This contrasts with the positive correlation scenario depicted in Figures \ref{longeralpha_05} and \ref{shorteralpha_05}. In Figure \ref{shorteralpha_n09}, with $x_1 > 5.0$ and $x_2 < 4.0$, the optimal strategy for the shorter-term portfolio still involves investing in the first stock. This is because the longer-term portfolio predominantly invests in the second stock within the same area, seeking a higher return. Given the negative correlation between stocks, it is preferable for the shorter-term portfolio to diversify investments to mitigate risk, rather than relying solely on the money account. This phenomenon also holds for the longer-term portfolio strategy when $x_1 < 5.0$ and $x_2 > 4.0$ in Figure \ref{longeralpha_n09}.
\end{enumerate}

\begin{figure}[H]
	\centering
	\includegraphics[width=0.7\textwidth]{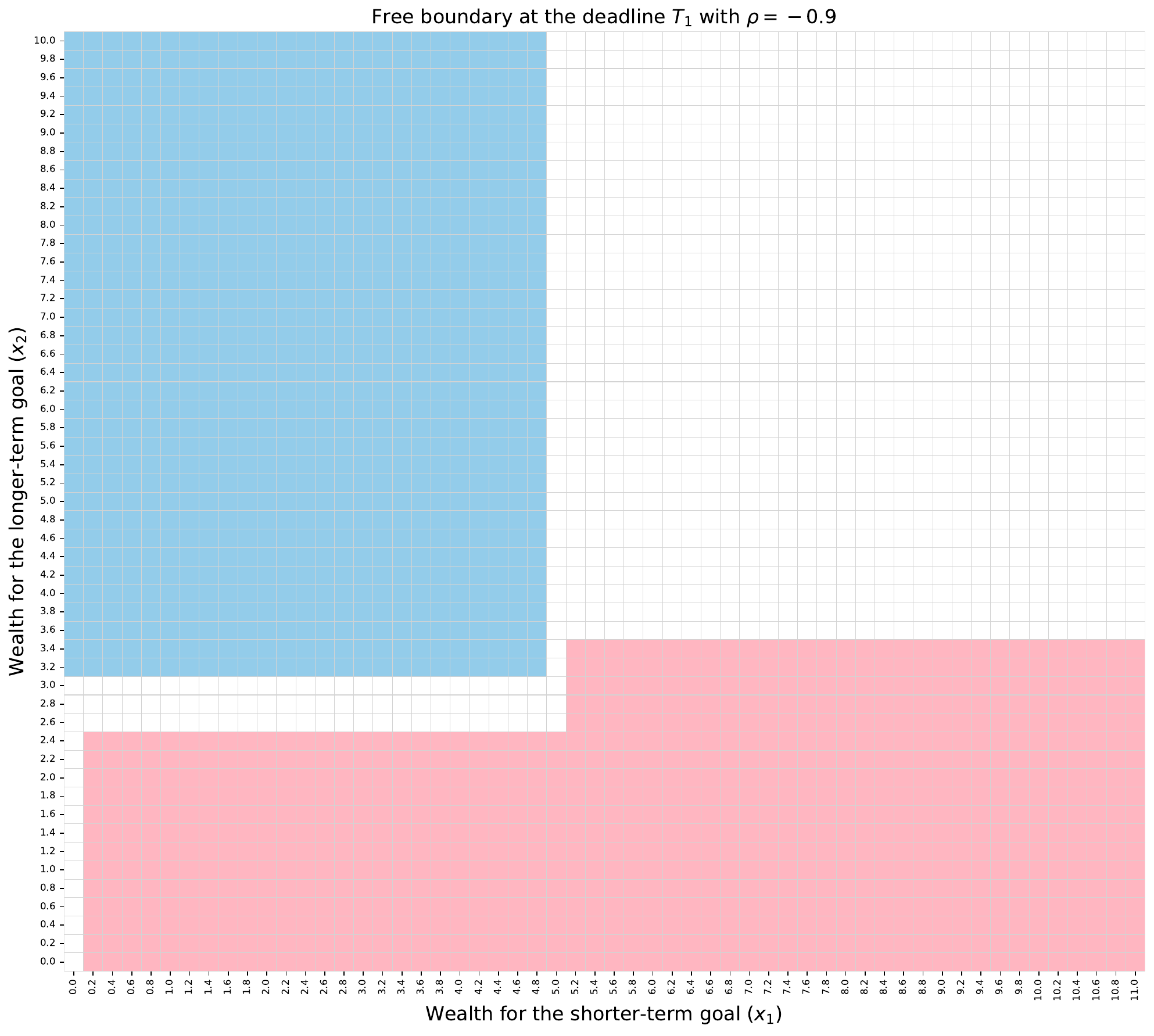}
	\caption{Free boundaries at the shorter-term goal deadline $T_1$. The correlation $\rho = -0.9$ here. }\label{T1impulse_corrn09}
\end{figure}

\begin{figure}[H]
	\centering
	\includegraphics[width=0.95\textwidth]{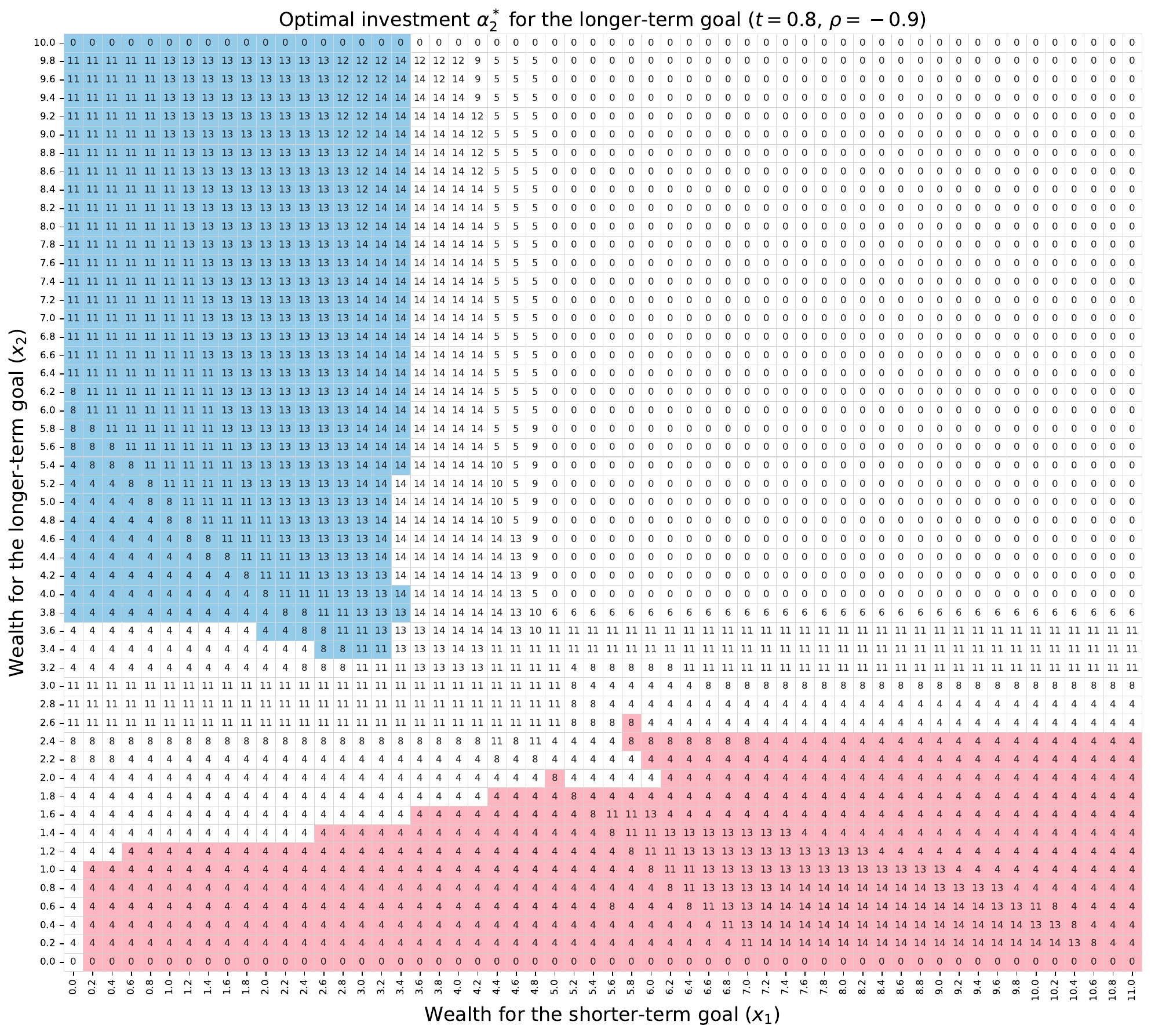}
	\caption{The optimal investment proportions $\alpha^*_2$ for the longer-term goal at time $t=0.8$. The correlation $\rho=-0.9$ here.}\label{longeralpha_n09}
\end{figure}

\begin{figure}[H]
	\centering
	\includegraphics[width=0.95\textwidth]{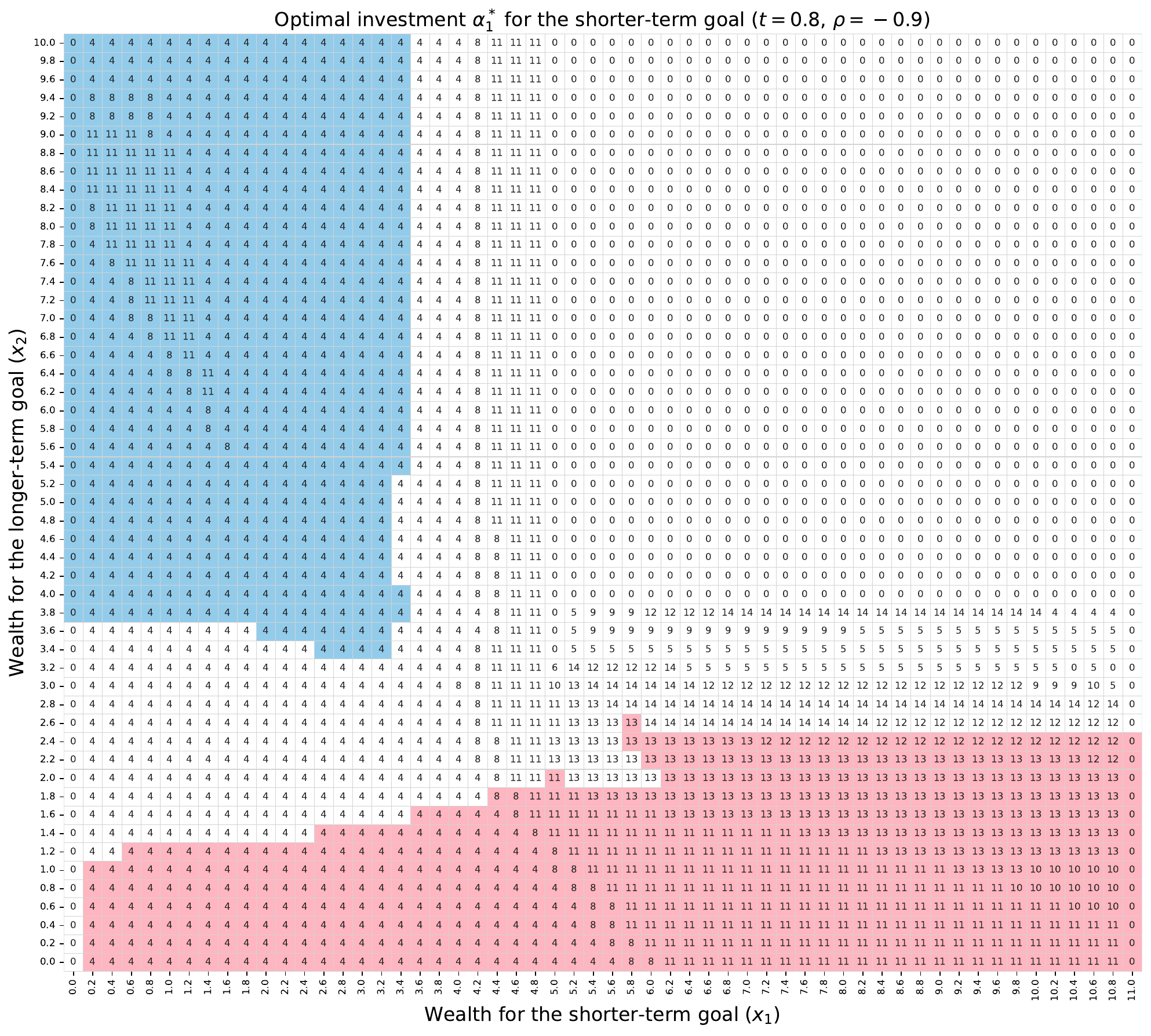}
	\caption{The optimal investment proportions $\alpha^*_1$ for the shorter-term goal at time $t=0.8$. The correlation $\rho=-0.9$ here.}\label{shorteralpha_n09}
\end{figure}

\subsection{When the shorter-term goal is important}
\begin{figure}[H]
	\centering
	\includegraphics[width=0.7\textwidth]{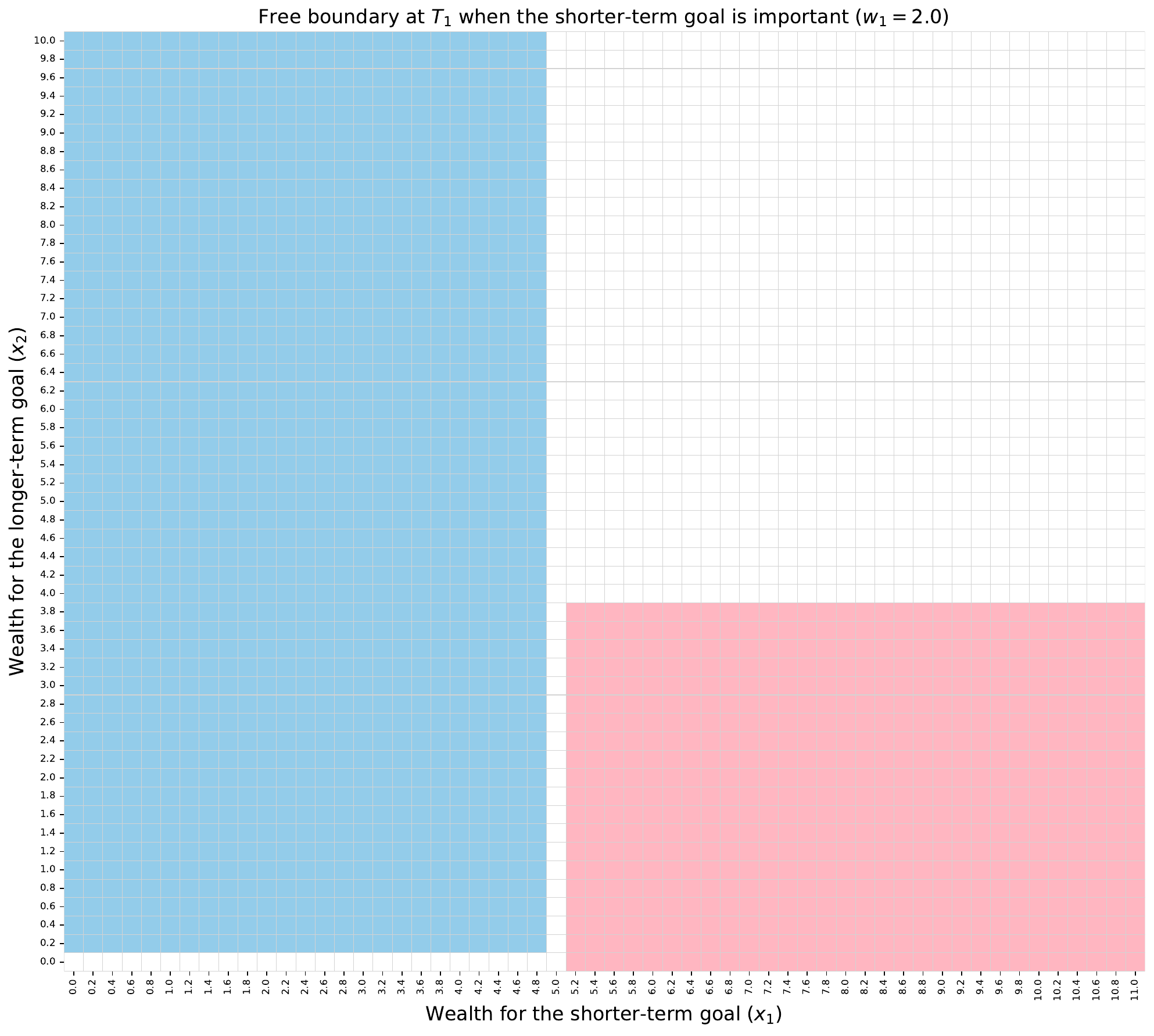}
	\caption{Free boundaries at the shorter-term goal deadline $T_1$. The weight $w_1=2.0$ and the correlation $\rho = 0.5$ here. }\label{T1impulse_important}
\end{figure}	

Suppose the shorter-term goal is prioritized by setting $w_1 = 2.0$, while other parameters are unchanged and $\rho = 0.5$. As illustrated in Figure \ref{T1impulse_important}, at the deadline $T_1$, the agent allocates all available funds toward the shorter-term goal. The continuation region in lower wealth areas is reduced to the line segments $(x_1, x_2) \in [0.0, 5.0] \times \{ 0.0 \}$ and $(x_1, x_2) \in \{ 5.0 \} \times [0.0, 4.0]$. A comparison of Figures \ref{shorter_important} and \ref{shorteralpha_05} reveals similar investment strategies in the stocks, while the figure for the longer-term case is omitted for simplicity.

At time $t = 0.8$, we observe that the bulges persist, and a ``notch" appears in $(x_1, x_2) \in [5.2, 6.6] \times [0.0, 1.6]$. This indicates that, despite the shorter-term portfolio exceeding the required level of 5.0, the agent postpones any transfers until $T_1$. In this region, the wealth in the longer-term portfolio is notably low, with full investment in the second stock. The shorter-term portfolio maintains moderate risk exposure in the stock market, with proportions $(0.25, 0.75)$ (Strategy Code 8) or $(0.5, 0.5)$ (Strategy Code 11), as shown in Figure \ref{shorter_important}. The notch phenomenon can be explained similarly to the bulge. The agent supports the longer-term portfolio $X_2$ by waiting for the free boundary to shift leftward. If the longer-term portfolio $X_2$ performs favorably, this strategy reduces the transfer amount needed to support $X_2$. Conversely, if the portfolio $X_2$ performs poorly, the agent will only support $X_2$ if there is a surplus not required by $G_1$ at time $T_1$, due to the importance of the shorter-term goal.

\begin{figure}[H]
	\centering
	\includegraphics[width=0.95\textwidth]{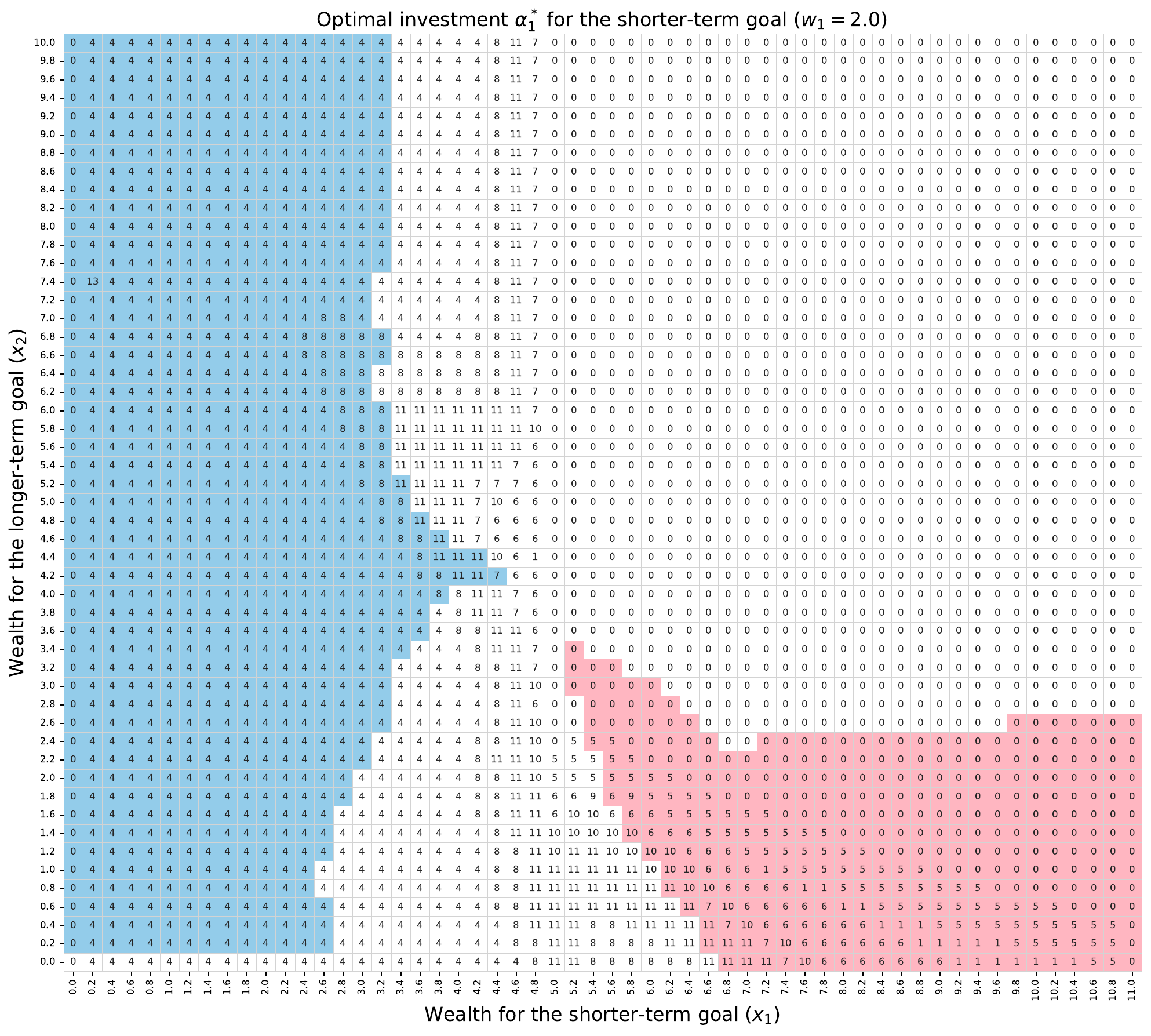}
	\caption{The optimal investment proportions $\alpha^*_1$ for the shorter-term goal at time $t=0.8$.  The weight $w_1=2.0$ and the correlation $\rho = 0.5$ here. }\label{shorter_important}
\end{figure}	

In addition to the figures presented in this paper, we provide visualizations of the free boundaries for the three settings discussed above at \url{https://github.com/hanbingyan/goaldemo}. These visualizations illustrate the evolution of the free boundaries when both goals are active.

\section{Stochastic supersolution}\label{sec:sto_sup}

The remaining sections present the proof of Theorem \ref{thm:viscosity}. The main approach is based on the stochastic Perron's method; see \cite{bayraktar2013stochastic, bayraktar2015stochastic}. The proof proceeds in three steps: 
\begin{enumerate}[label={(\arabic*)}]	
	\item In Section \ref{sec:sto_sup}, we introduce a family of stochastic supersolutions which bound the value function from above. The infimum of this family is called the upper stochastic envelope and is shown to be a viscosity subsolution. 
	
	\item In Section \ref{sec:sto_sub}, we introduce a family of stochastic subsolutions which bound the value function from below. The supremum of this family is called the lower stochastic envelope and is shown to be a viscosity supersolution. 
	
	\item Finally, a comparison argument in Section \ref{sec:compare} completes the proof. 
\end{enumerate}

The stochastic Perron's method bypasses the need to first establish the dynamic programming principle (DPP), instead deriving it after proving that the value function is a viscosity solution. This approach circumvents the complicated and sometimes incomplete proofs of the DPP.

Before presenting the details, we outline the main technical challenges and the techniques used to address them:
\begin{enumerate}[label={(\arabic*)}]	
	\item The value function loses one state variable $x_k$ after the deadline of goal $k$. Then we must define stochastic supersolutions and subsolutions properly in Definitions \ref{def:sto_super} and \ref{def:sto_sub}, respectively. The boundary condition in equation \eqref{Tk_bc} at the deadline $T_k$ links the value function $V_k(T_k, x_{k:K+1}, y)$ to $V_{k+1}(T_k, x_{k+1:K+1}, y)$, which is a novel problem. The stochastic Perron's method is adapted to prove the viscosity solution property at $T_k$, as detailed in Step 3 of the proof of Propositions \ref{prop:vissub} and \ref{prop:vissuper}.
	\item Another technical challenge is the state constraint that the wealth processes are nonnegative. Similar to \citet[Equation (4.3)]{bentahar2007} and \citet[Theorem 1.2]{rokhlin2014}, Lemma \ref{lem:boundary} shows that the upper stochastic envelope $v_+$ also provides information about the value function on the boundary $\overline{\cS}^o_{k} \backslash \cS_k$. Together with Lemma \ref{lem:conti}, we can relax the continuity assumption imposed in the \citet[Theorem 7.9]{crandall1992user} when proving the comparison principle. A similar idea is also used in \citet[Lemma 4.1]{bentahar2007}. 
	\item The mental accounting formulation leads to gradient constraints from two sides as follows:
	\begin{equation*}
		- \theta_{i}  \leq \partial_{K+1} V_k - \partial_{i} V_k \leq \lambda_{i}, \quad i = k, \ldots, K.
	\end{equation*}
	This requires a modified approach to proving the comparison principle, similar to the technique in \citet[Proposition 2.2]{hynd2012eigenvalue}. The gradient constraints also complicate the construction of a strict classical subsolution as in Lemma \ref{lem:strict_l}. In fact, if we extend the domain of $x_{1:K+1}$ from $[0, \infty)^{K+1}$ to $\R^{K+1}$ and try to avoid the difficulty from state constraints like \cite{capponi2024}, then we could not find a strict classical subsolution on $\R^{K+1}$. Therefore, this paper follows the techniques mentioned in item (2) instead.
\end{enumerate}

For $k=1, \ldots, K+1$, a tuple $(\tau, \xi_{k:K+1}, \xi_y)$ is called a {\it random initial condition} for the state process \eqref{Y}, \eqref{Xlast}, and \eqref{Xk}, if $\tau \in [T_{k-1}, T]$ is an $\mathbb{F}$-stopping time and $(\xi_{k:K+1}, \xi_y)$ is an $\cF_{\tau}$-measurable random vector with $\p((\xi_{k:K+1}, \xi_y) \in \overline{\cS}_k) = 1$. It is possible that some elements in $\xi_{k:K+1}$ may be unnecessary. For example, if $\p( \tau > T_k) = 1$, then the random initial value $\xi_k$ is not needed since the portfolio $X_k$ terminates after the deadline $T_k$. However, we keep the redundant $\xi_k$ for notational convenience.

Denote $\{(X_{k:K+1}, Y)(t; \tau, \xi_{k:K+1}, \xi_y, U) \}_{t \in [\tau, T]}$ as the solution of the state process \eqref{Y}, \eqref{Xlast}, and \eqref{Xk}, with a control $U$ and random initial condition $(\tau, \xi_{k:K+1}, \xi_y)$ in the sense that
\begin{equation*}
	(X_{k:K+1}, Y)(\tau-; \tau, \xi_{k:K+1}, \xi_y) = (\xi_{k:K+1}, \xi_y).
\end{equation*}
Here, $U:=(\alpha_{k:K+1}, L_{k:K}, M_{k:K})$ is called $(\tau, \xi_{k:K+1}, \xi_y)$-admissible if 
\begin{equation*}
\p((X_{k:K+1}, Y)(t; \tau, \xi_{k:K+1}, \xi_y, U) \in \overline{\cS}_k, \; \tau \leq t \leq T ) = 1.
\end{equation*}
We adopt the convention that $X_k(t) = X_k(T_k)$ for $t > T_k$, i.e., the state ceases at the deadline $T_k$. Similarly, some elements in $X_{k:K+1}$ and $U$ may be redundant if $\p(\tau > T_k) = 1$.

Similar to \citet[Lemma 2.2]{bayraktar2015stochastic}, we establish the following result on the concatenation of admissible controls. The proof is also similar to that of \citet[Lemma 2.2]{bayraktar2015stochastic} and we omit the details for brevity.
\begin{lemma}\label{lem:glue}
	For $k=1, \ldots, K+1$, consider a random initial condition $(\tau, \xi_{k:K+1}, \xi_y)$ with $\tau \in [T_{k-1}, T_k]$.
	\begin{enumerate}
		\item If $(\alpha^1_{k:K+1}, L^1_{k:K}, M^1_{k:K})$ and $(\alpha^2_{k:K+1}, L^2_{k:K}, M^2_{k:K})$ are $(\tau, \xi_{k:K+1}, \xi_y)$-admissible and $A$ is an $\cF_\tau$-measurable set, then
		\begin{align*}
			&(\alpha_{k:K+1}(t), L_{k:K}(t), M_{k:K}(t)) \\
			& := \one_{ \{ \tau \leq t \} \cap A} (\alpha^1_{k:K+1}(t), L^1_{k:K}(t) - L^1_{k:K}(\tau-), M^1_{k:K}(t) - M^1_{k:K}(\tau-)) \\
			& \qquad + \one_{ \{ \tau \leq t \} \cap A^c} (\alpha^2_{k:K+1}(t), L^2_{k:K}(t) - L^2_{k:K}(\tau-), M^2_{k:K}(t) - M^2_{k:K}(\tau-))
		\end{align*}
		is also $(\tau, \xi_{k:K+1}, \xi_y)$-admissible.
		\item Let $U^1 := (\alpha^1_{k:K+1}, L^1_{k:K}, M^1_{k:K})$ be a $(\tau, \xi_{k:K+1}, \xi_y)$-admissible control, $\rho \in [\tau, T]$ be an $\mathbb{F}$-stopping time, and $(\xi^1_{k:K+1}, \xi^1_y) := (X_{k:K+1}, Y)(\rho; \tau, \xi_{k:K+1}, \xi_y, U^1)$. Then $(\rho, \xi^1_{k:K+1}, \xi^1_y)$ is a random initial condition. If $U^2 := (\alpha^2_{k:K+1}, L^2_{k:K}, M^2_{k:K})$ is $(\rho, \xi^1_{k:K+1}, \xi^1_y)$-admissible, then
		\begin{align*}
			&(\alpha_{k:K+1}(t), L_{k:K}(t), M_{k:K}(t)) \\
			& := \one_{\{\tau \leq  t < \rho \}} (\alpha^1_{k:K+1}(t), L^1_{k:K}(t), M^1_{k:K}(t)) \\
			& \qquad + \one_{ \{ \rho \leq t \leq T\}} (\alpha^2_{k:K+1}(t), L^2_{k:K}(t) - L^2_{k:K}(\rho-) + L^1_{k:K}(\rho), \\
			& \hspace{3.0cm} M^2_{k:K}(t) - M^2_{k:K}(\rho-) + M^1(\rho))
		\end{align*}
		is $(\tau, \xi_{k:K+1}, \xi_y)$-admissible.
	\end{enumerate}
\end{lemma}
Here, if the stopping time $\rho$ is after the deadline of goal $i$ with $i \geq k$, i.e., $\p(\rho > T_i) = 1$, then the initial conditions $\xi^1_{k:i}$ become redundant. Again, we include $\xi^1_{k:i}$ for notational convenience.

\begin{definition}[Stochastic supersolution]\label{def:sto_super}
	Denote an array of functions as
	\begin{equation}\label{sto_super}
		(\{ v_1(t, x_{1:K+1}, y) \}_{t \in [0, T_1]}, \ldots, \{ v_k(t, x_{k:K+1}, y) \}_{t \in [T_{k-1}, T_{k}]}, \ldots, \{ v_{K+1}(t, x_{K+1}, y) \}_{t \in [T_K, T]} ),
	\end{equation}
	where $v_k(t, x_{k:K+1}, y): [T_{k-1}, T_k] \times \overline{\cS}_k \rightarrow \R, k=1,\ldots, K+1$ is bounded and USC. The array \eqref{sto_super} is a stochastic supersolution of the HJB equation system if
	\begin{enumerate}[label={(\arabic*)}]	
		\item for each $k=1, \ldots K+1$, consider any random initial condition $(\tau, \xi_{k:K+1}, \xi_y)$ with $\tau \in [T_{k-1}, T_k]$, $(\xi_{k:K+1}, \xi_y) \in \cF_\tau$ and $\p((\xi_{k:K+1}, \xi_y) \in \cS_k) = 1$. There exists an admissible control $U:= (\alpha_{k:K+1}, L_{k:K}, M_{k:K}) \in \cU(\tau, \xi_{k:K+1}, \xi_y)$, such that
		\begin{equation*}
			\p((X_{k:K+1}, Y)(t; \tau, \xi_{k:K+1}, \xi_y, U) \in \cS_k, \; \tau \leq t \leq T) = 1,
		\end{equation*}
		and for any stopping time $\rho \in [\tau, T]$,
		\begin{equation}
			v_k(\tau, \xi_{k:K+1}, \xi_y) \geq \E \big[ \cM \big([\tau, \rho], v_{k:K+1}, (X_{k:K+1}, Y)(\cdot; \tau, \xi_{k:K+1}, \xi_y, U) \big) \big| \cF_\tau \big],
		\end{equation}
		where {\color{black} for $k =1, \ldots, K$,}
		\begin{align}
			& \cM \big([\tau, \rho], v_{k:K+1}, (X_{k:K+1}, Y)(\cdot; \tau, \xi_{k:K+1}, \xi_y, U) \big) \label{M} \\
			& \quad := \Big\{ e^{-\beta(\rho - \tau)} v_{k}(\rho, (X_{k:K+1}, Y)(\rho; \tau, \xi_{k:K+1}, \xi_y, U)) \nonumber \\
			& \quad \qquad + \lambda_k \int^{\rho}_\tau e^{-\beta(s-\tau)} d L_k(s) + \theta_k \int^{\rho}_\tau e^{-\beta (s - \tau)} d M_k(s) \Big \} \one_{\{\tau \leq \rho < T_k\}} \nonumber \\		
			& \qquad + \sum^{K-1}_{l=k} \Big\{ e^{-\beta(\rho - \tau)} v_{l+1}(\rho, (X_{l+1:K+1}, Y)(\rho; \tau, \xi_{k:K+1}, \xi_y, U)) \nonumber \\
			&\quad \qquad \qquad +  \sum^{l}_{i = k} w_i e^{-\beta (T_i - \tau)} (G_i - X_i(T_i; \tau, \xi_{k:K+1}, \xi_y, U))^+  \nonumber \\
			& \quad \qquad \qquad + \sum^{l+1}_{i=k} \lambda_i \int^{\rho \wedge T_i}_\tau e^{-\beta(s-\tau)} d L_i(s) + \sum^{l+1}_{i=k} \theta_i \int^{\rho \wedge T_i}_\tau e^{-\beta (s - \tau)} d M_i(s) \Big \} \one_{\{T_l \leq \rho < T_{l+1} \}} \nonumber \\
			&\qquad + \Big\{ e^{-\beta(\rho - \tau)} v_{K+1}(\rho, (X_{K+1}, Y)(\rho; \tau, \xi_{k:K+1}, \xi_y, U)) \nonumber \\
			&\qquad \qquad +  \sum^{K}_{i = k} w_i e^{-\beta (T_i - \tau)} (G_i - X_i(T_i; \tau, \xi_{k:K+1}, \xi_y, U))^+  \nonumber \\
			& \qquad \qquad + \sum^{K}_{i=k} \lambda_i \int^{T_i}_\tau e^{-\beta(s-\tau)} d L_i(s) + \sum^{K}_{i=k} \theta_i \int^{T_i}_\tau e^{-\beta (s - \tau)} d M_i(s) \Big \} \one_{\{T_K \leq \rho \leq T\}}. \nonumber
		\end{align}	 
		{\color{black}
		For $k = K+1$,  the functional $\cM$ simplifies to
		\begin{equation}\label{M:K+1}
		\begin{aligned}
			& \cM \big([\tau, \rho], v_{K+1}, (X_{K+1}, Y)(\cdot; \tau, \xi_{K+1}, \xi_y, U) \big) \\
			& \quad := e^{-\beta(\rho - \tau)} v_{K+1}(\rho, (X_{K+1}, Y)(\rho; \tau, \xi_{K+1}, \xi_y, U)).
		\end{aligned}
		\end{equation} 
		}
		 We refer to $U$ as a {\bf suitable} control for $v_{k}$ (and $v_{k+1}, \ldots, v_{K+1}$ in  \eqref{sto_super}) with the random initial condition $(\tau, \xi_{k:K+1}, \xi_y)$.  
	\item  
	\begin{align}
		v_k(t, 0, y) & \geq \sum^{K+1}_{i = k} w_i e^{-\beta (T_i-t)} G_i, \quad \forall\;  t \in [0, T], y \in \R^m, k = 1, \ldots, K+1,\\
		v_{K+1}(T, x_{K+1}, y) & \geq (G_{K+1} - x_{K+1})^+, \quad \forall \; (x_{K+1}, y) \in \cS_{K+1}.
	\end{align}
	\end{enumerate}
	Denote the set of stochastic supersolutions as $\cV^+$. 
\end{definition}
\begin{remark}
  The first argument $[\tau, \rho]$ in $\cM$ represents the interval of integration. The second $\tau$ is in the random initial condition of state processes. In general, the random initial condition can be earlier than the lower limit of the integral. Then two different stopping times can appear in these two places. See \eqref{Q1M} for an example.
\end{remark}
{\color{black} Clearly, the array with functions
\begin{equation}\label{ex:sto_super}
	v_k(t, x_{k:K+1}, y) := \sum^{K+1}_{i = k} w_i e^{-\beta (T_i-t)} G_i,
\end{equation} 
is a stochastic supersolution with the suitable control $U = (0, 0, 0)$ that does not invest in stocks and makes no transfer.
}

Following a similar argument as in \citet[Remark 3.2]{bayraktar2015stochastic}, for $k=1, \ldots, K+1$, any stochastic supersolution in $\cV^+$ dominates the value function in the following sense:
\begin{align*}
	V_{k}(t, x_{k:K+1}, y)  & \leq v_{k}(t, x_{k:K+1}, y) \quad \text{ on } \quad [T_{k-1}, T_k] \times (\cS_k \cup (\{ 0 \} \times \R^m)).
\end{align*}

Lemma \ref{lem:two_stochsup} below shows that the family $\cV^+$ of stochastic supersolutions is stable under minimum. The proof of Lemma \ref{lem:two_stochsup} is similar to \citet[Lemma 3.2]{bayraktar2015stochastic} and omitted here.
\begin{lemma}\label{lem:two_stochsup}
	If $ (v^1_1, \ldots, v^1_k, \ldots, v^1_{K+1})$ and $ (v^2_1, \ldots, v^2_k, \ldots, v^2_{K+1})$ are stochastic supersolutions, then $(v^1_1 \wedge v^2_1, \ldots, v^1_k \wedge v^2_k, \ldots, v^1_{K+1} \wedge v^2_{K+1})$ is also a stochastic supersolution.
\end{lemma}

Given $k = 1, \ldots, K+1$ and $(t, x_{k:K+1}, y) \in [T_{k-1}, T_k] \times \overline{\cS}_k$, define
\begin{align}
	&v_{k, +}(t, x_{k:K+1}, y) := \inf \big\{ v_k(t, x_{k:K+1}, y) |\; v_k \text{ is the $k$-th element of some } v \in \cV^+ \big\}.  \label{v+}
\end{align}
Here, $v \in \cV^+$ means that $v := (v_1, \ldots, v_k, \ldots, v_{K+1})$ is a stochastic supersolution. Importantly, the infimum in $v_{k, +}$ only considers $v_k$ when it can construct a stochastic supersolution $v$ with some $(v_1, \ldots, v_{k-1}, v_{k+1}, \ldots, v_{K+1})$. We write the upper stochastic envelope as $v_+ := (v_{1, +}, \ldots, v_{k, +}, \ldots, $ $v_{K+1, +})$.
\begin{proposition}\label{prop:vissub}
	The upper stochastic envelope $v_+$ is a viscosity subsolution of the HJB equation under Definition \ref{def:vis_sub}.
\end{proposition}
The proof of Proposition \ref{prop:vissub} is given in Appendix \ref{sec:proof_stosuper}.

Lemma \ref{lem:boundary} below shows that the upper stochastic envelope $v_+$ also provides information about the value function on the boundary $\overline{\cS}^o_{k} \backslash \cS_k$. A similar claim can be found in \citet[Equation (4.3)]{bentahar2007} and \citet[Theorem 1.2]{rokhlin2014}. We do not have to consider the case with $k = K+1$, since the set $\overline{\cS}^o_{K+1} \backslash \cS_{K+1}$ is empty.

\begin{lemma}\label{lem:boundary}
	For each $k=1, \ldots, K$, consider $t \in [T_{k-1}, T_k]$ and $(x^o_{k:K+1}, y) \in \overline{\cS}^o_{k} \backslash \cS_k$. $v_{k, +}$ in the upper stochastic envelope $v_+$ has the following property:
	\begin{align}\label{ineq:uscbd}
		V_k(t, x^o_{k:K+1}, y) \leq \limsup_{\substack{ x'_{k:K+1} \rightarrow x^o_{k:K+1} \\ x'_j > 0, \, j = k,\ldots,K+1}} v_{k, +}(t, x'_{k:K+1}, y),
	\end{align}
	where $V_k$ is the value function.
\end{lemma}
We introduce some notation used in the proofs of Lemma \ref{lem:boundary} and also Lemma \ref{lem:conti}. Given some $k=1, \ldots, K$, for any $(x_{k:K+1}, y) \in \overline{\cS}^o_{k}$, there exists at least one index $i \in \{k, \ldots, K+1\}$ such that $x_i > 0$. Without loss of generality, consider the smallest index $i$ with $x_i > 0$. Suppose $N$ elements in $x_{k:K+1}$ are zero. If $(x_{k:K+1}, y) \in \overline{\cS}^o_{k}\backslash \cS_k$, then $N > 0$. Instead, if $(x_{k:K+1}, y) \in \cS_k$, then $N=0$ and no transfer is needed. We introduce a small constant $\zeta > 0$. Define a new wealth state $Z_{k:K+1}(x_{k:K+1}, \zeta)$:
\begin{equation}\label{z}
	Z_i(x_{k:K+1}, \zeta) = x_i - N \zeta \text{ and for } j \neq i, \; 
	Z_j(x_{k:K+1}, \zeta) = \left\{ 
	\begin{array}{ c l }
		\zeta, & \text{ if } x_j = 0, \\
		x_j , & \text{ if } x_j > 0.
	\end{array}
	\right.
\end{equation}
From the state $x_{k:K+1}$ on the boundary to the state $Z_{k:K+1}(x_{k:K+1}, \zeta)$ in the interior, we reallocate $\zeta$ dollars from $x_i$ to each portfolio $j$ with $x_j = 0$. In general, the index $i$ may not be the fundamental portfolio $K+1$. Hence, the transfer is achieved via the route $x_i \rightarrow x_{K+1} \rightarrow x_j$. Evidently, there exists a unique transfer plan $(\zeta \Delta L_{k:K}(x_{k:K+1}), \zeta \Delta M_{k:K}(x_{k:K+1}))$ such that
\begin{equation}\label{transfer}
	\begin{aligned}
		& x_{K+1} - \sum_{j=k}^{K} \zeta \Delta L_j (x_{k:K+1}) +  \sum_{j=k}^{K} \zeta \Delta M_j(x_{k:K+1})  = Z_{K+1}(x_{k:K+1}, \zeta), \\
		& x_j + \zeta \Delta L_j(x_{k:K+1}) - \zeta \Delta M_j(x_{k:K+1}) = Z_j(x_{k:K+1}, \zeta), \text{ for } j = k, \ldots, K.
	\end{aligned}
\end{equation}
Elements of $\Delta L_{k:K}(x_{k:K+1})$ and $\Delta M_{k:K}(x_{k:K+1})$ are nonnegative integers. For each $j=1, \ldots, K$, $\Delta L_j(x_{k:K+1})$ and $\Delta M_j(x_{k:K+1})$ are not positive at the same time.

The proof of Lemma \ref{lem:boundary} is given in the Appendix \ref{sec:proof_stosuper}. Thanks to Lemma \ref{lem:boundary}, we can modify $v_+$ as follows. When $k=1, \ldots, K$ and $t \in [T_{k-1}, T_k]$, set
\begin{equation}\label{eq:modified}
	 \tilde{v}_{k, +} (t, x_{k:K+1}, y) = \left\{ 
	\begin{array}{ c l }
		\limsup_{\substack{ x'_{k:K+1} \rightarrow x_{k:K+1} \\ x'_j > 0, \, j = k,\ldots,K+1}} v_{k, +}(t, x'_{k:K+1}, y), & \text{ if } (x_{k:K+1}, y) \in \overline{\cS}^o_{k} \backslash \cS_k, \\
		v_{k, +} (t, x_{k:K+1}, y), & \text{ otherwise}.
	\end{array}
	\right.
\end{equation}
This modification is unnecessary for $v_{K+1, +}$. Additionally, we note that $\tilde{v}_+ := (\tilde{v}_{1, +}, \ldots, \tilde{v}_{k, +}, \ldots, $ $\tilde{v}_{K, +}, v_{K+1, +})$ is a viscosity subsolution of the HJB equation system. In Section \ref{sec:compare}, it will be used to validate the condition \eqref{modi_bd_u} in the comparison principle.

\section{Stochastic subsolution}\label{sec:sto_sub}

\begin{definition}[Stochastic subsolution]\label{def:sto_sub}
	Denote an array of functions as
	\begin{equation}\label{sto_sub}
		(\{ v_1(t, x_{1:K+1}, y) \}_{t \in [0, T_1]}, \ldots, \{ v_k(t, x_{k:K+1}, y) \}_{t \in [T_{k-1}, T_{k}]}, \ldots, \{ v_{K+1}(t, x_{K+1}, y) \}_{t \in [T_K, T]} ),
	\end{equation}
	where $v_k(t, x_{k:K+1}, y): [T_{k-1}, T_k] \times \overline{\cS}_k \rightarrow \R, k=1,\ldots, K+1$ is bounded and LSC. The array \eqref{sto_sub} is a stochastic subsolution of the HJB equation system if
	\begin{enumerate}[label={(\arabic*)}]	
		\item for each $k=1, \ldots K+1$, consider any random initial condition $(\tau, \xi_{k:K+1}, \xi_y)$ with $\tau \in [T_{k-1}, T_k]$, $(\xi_{k:K+1}, \xi_y) \in \cF_\tau$ and $\p((\xi_{k:K+1}, \xi_y) \in \overline{\cS}^o_k) = 1$. For any admissible control $U:= (\alpha_{k:K+1}, L_{k:K}, M_{k:K}) \in \cU(\tau, \xi_{k:K+1}, \xi_y)$ and any stopping time $\rho \in [\tau, T]$,
		\begin{equation}
			v_k(\tau, \xi_{k:K+1}, \xi_y) \leq \E \big[ \cM \big([\tau, \rho], v_{k:K+1}, (X_{k:K+1}, Y)(\cdot; \tau, \xi_{k:K+1}, \xi_y, U) \big) \big|\cF_\tau \big],
		\end{equation}
		where $\cM \big([\tau, \rho], v_{k:K+1}, (X_{k:K+1}, Y)(\cdot; \tau, \xi_{k:K+1}, \xi_y, U) \big)$ is defined in \eqref{M} and \eqref{M:K+1}.
		\item  
		\begin{align}
			v_k(t, 0, y) & \leq \sum^{K+1}_{i = k} w_i e^{-\beta (T_i-t)} G_i, \quad \forall\;  t \in [0, T], y \in \R^m, k = 1, \ldots, K+1,\\
			v_{K+1}(T, x_{K+1}, y) & \leq (G_{K+1} - x_{K+1})^+, \quad \forall \; (x_{K+1}, y) \in \cS_{K+1}.
		\end{align}
	\end{enumerate}
	Denote the set of stochastic subsolutions as $\cV^-$.
\end{definition}

Using a similar argument to \citet[Remark 4.2]{bayraktar2015stochastic}, we can demonstrate that any stochastic subsolution in $\cV^-$ is dominated by the value function as follows:
\begin{align}
	v_{k}(t, x_{k:K+1}, y) & \leq V_{k}(t, x_{k:K+1}, y) \quad \text{ on } \quad [T_{k-1}, T_k] \times \overline{\cS}_k, \quad k=1, \ldots, K+1.
\end{align}

Lemma \ref{lem:two_sub} below shows that the family $\cV^-$ of stochastic subsolutions is stable under maximum. We omit the proof, since it follows directly from the definition of stochastic subsolutions and is similar to \citet[Lemma 4.1]{bayraktar2015stochastic}.
\begin{lemma}\label{lem:two_sub}
	If $ (v^1_1, \ldots, v^1_k, \ldots, v^1_{K+1})$ and $ (v^2_1, \ldots, v^2_k, \ldots, v^2_{K+1})$ are stochastic subsolutions, then $(v^1_1 \vee v^2_1, \ldots, v^1_k \vee v^2_k, \ldots, v^1_{K+1} \vee v^2_{K+1})$ is also a stochastic subsolution.
\end{lemma}

To address the boundary conditions of the viscosity supersolution, we consider an important example of a stochastic subsolution: the {\it value} function without mental accounting, i.e., $\lambda_{k} = \theta_k = 0$. In this context, \cite{capponi2024} considered a single portfolio $X$ for all goals. Between two goal deadlines, the single wealth process $X$ satisfies
\begin{equation}\label{eq:no_account}
	\begin{aligned}
		d X(u) &=  r X(u) du + (\mu(Y(u)) - r \mathbf{1})^\top \alpha(u) X(u) du + X(u) \alpha^\top(u) \sigma(Y(u)) dW(u), \\
		X(t) &= x \geq 0,
	\end{aligned}
\end{equation}
where the initial wealth $x$ is one-dimensional. At the goal deadline $T_k$, the agent withdraws $C_k$ to fund goal $k$. Hence, $X(T_k) = X(T_k-) - C_k$.

During time $t \in [T_{k-1}, T_{k}]$ with $k = 1, \ldots, K+1$, the value function without mental accounting is defined as 
\begin{equation}\label{obj:no_acc}
	\begin{aligned}
		J(t, x, y) := \inf_{\substack{ (\alpha, C_{1:K+1} ) \\ \in \cU(t, x, y) } }  \E \Big[& \sum^{K+1}_{i = 1} w_i e^{-\beta (T_i-t)} (G_i - C_i)^+ \one_{\{ T_i > t \}} \Big| X(t) = x, Y(t) = y \Big].
	\end{aligned}
\end{equation}
{\color{black} Here, $\cU(t, x, y)$ denotes the set of admissible controls consisting solely of the investment strategy and the withdrawal sequence, as cross-account transfers are no longer applicable.} Importantly, we use $X(t) = x$ instead of $X(t-) = x$ in \eqref{obj:no_acc}. Hence, $J(T_k, x, y)$ is the value function right after funding the goal $k$. Based on $J(t, x, y)$, we introduce
\begin{align*}
	\Gamma_k(t, x_{k:K+1}, y) & := J(t, x_k + \ldots + x_{K+1}, y) \quad \text{ for } t \in [T_{k-1}, T_k), \\
	\Gamma_k(T_k, x_{k:K+1}, y) & := J(T_k- , x_k + \ldots + x_{K+1}, y).
\end{align*}

\begin{lemma}\label{lem:no_mental}
	$\Gamma := (\Gamma_1, \ldots, \Gamma_k, \ldots, \Gamma_{K+1})$ is a stochastic subsolution.
\end{lemma}

Given $k = 1, \ldots, K+1$ and $(t, x_{k:K+1}, y) \in [T_{k-1}, T_k] \times \overline{\cS}_k$, we define
\begin{align}
	&v_{k, -}(t, x_{k:K+1}, y) := \sup \big\{ v_k(t, x_{k:K+1}, y) |\; v_k \text{ is the $k$-th element of some } v \in \cV^- \big\}.  \label{v-}
\end{align}
The lower stochastic envelope is denoted as $v_- := (v_{1, -}, \ldots, v_{k, -}, \ldots, v_{K+1, -})$.

\begin{proposition}\label{prop:vissuper}
	The lower stochastic envelope $v_-$ is a viscosity supersolution of the HJB equation under Definition \ref{def:vis_super}.
\end{proposition}

Proofs of Lemma \ref{lem:no_mental} and Proposition \ref{prop:vissuper} are provided in Appendix \ref{sec:proof_stosub}.

\section{Comparison principle}\label{sec:compare}
In this section, we establish a comparison principle to prove the continuity and uniqueness of viscosity solutions of the HJB equation system. Detailed proofs of all results in this section can be found in Appendix \ref{sec:proof_comp}.

The first step is to construct a strict classical subsolution as follows.
\begin{lemma}\label{lem:strict_l}
	Let constants $c_1, c_2>0$, $a_i > 0, i = k, \ldots, K+1$, $q \in (0, 1)$, $\gamma>0$, and define
	\begin{equation}
		l(t, x_{k:K+1}, y) := - c_1 e^{\gamma(T_k-t)} \Big(1 + \sum^{K+1}_{i=k} a_i x_i \Big)^q - c_2 e^{\gamma(T_k - t)} (1+|y|^2),
	\end{equation}
	with $(t, x_{k:K+1}, y) \in [T_{k-1}, T_k] \times \overline{\cS}_k$. 
	
	Suppose
	\begin{equation}
		\begin{aligned}
			 - \theta_i < (a_i - a_{K+1}) < \lambda_i, \; i = k, \ldots K, \quad 0 < c_1 q e^{\gamma(T_k - t)} < 1,  \quad t \in [T_{k-1}, T_k],
		\end{aligned} 
	\end{equation}
	and $\gamma$ is large enough. Then $l$ is a strict classical subsolution of \eqref{F} on $[T_{k-1}, T_k) \times \overline{\cS}_k$.
\end{lemma}

Next, we present the following lemma, which is crucial to relax the continuity assumption imposed in \citet[Theorem 7.9]{crandall1992user}. A similar result can be found in \citet[Lemma 4.1]{bentahar2007}.

For notational simplicity, denote the vector $\zeta g(x_{k:K+1}) := x_{k:K+1} - Z_{k:K+1}(x_{k:K+1}, \zeta)$ with $Z_{k:K+1}(x_{k:K+1}, \zeta)$ defined in \eqref{z}.    
\begin{lemma}\label{lem:conti}
	Let $u \in \text{USC}([T_{k-1}, T_k] \times \overline{\cS}_k)$ such that for $(x^o_{k:K+1}, y) \in \overline{\cS}^o_{k} \backslash \cS_k$,
	\begin{equation}
		u(t, x^o_{k:K+1}, y) = \limsup_{\substack{ x'_{k:K+1} \rightarrow x^o_{k:K+1} \\ x'_j > 0, \, j = k,\ldots,K+1}} u(t, x'_{k:K+1}, y).
	\end{equation}
	Suppose
	\begin{equation}\label{u_sub}
		\begin{aligned}
			\max \Big\{ & - \lambda_{k} + \partial_{K+1} u - \partial_{k} u, \ldots, - \lambda_{K} + \partial_{K+1} u - \partial_{K} u, \\
			& - \theta_{k} - \partial_{K+1} u + \partial_{k} u, \ldots, - \theta_K - \partial_{K+1} u + \partial_K u \Big\} \leq 0
		\end{aligned}
	\end{equation}
	holds on $[T_{k-1}, T_k] \times \cS_k$ in the viscosity sense. Then
	\begin{equation}\label{conti}
		\begin{aligned}
			& \limsup_{\zeta \downarrow 0} u(t, x_{k:K+1} - \zeta g(x_{k:K+1}), y) = \liminf_{\zeta \downarrow 0} u(t, x_{k:K+1} - \zeta g(x_{k:K+1}), y) \\
			& = u(t, x_{k:K+1}, y)
		\end{aligned}
	\end{equation}
	holds for any $t \in [T_{k-1}, T_k]$ and $(x_{k:K+1}, y) \in \overline{\cS}^o_{k}$.
\end{lemma}

In the proof of the comparison principle, one technical difficulty is that the constraints on $\partial_{K+1} v - \partial_{k} v$ are from both the above and below. To address this, we perturb $u - v$ by introducing $\delta u - v$, where the constant $\delta$ is close to $1$. The same idea also appears in \citet[Proposition 2.2]{hynd2012eigenvalue}.
\begin{proposition}[Comparison principle: $t \in [T_{k-1}, T_{k})$]\label{p:comp_btw}
	Given some $k=1, \ldots, K$, suppose
	\begin{enumerate}[label={(\arabic*)}]	
		\item  $u_k \in USC([T_{k-1}, T_k] \times \overline{\cS}_{k})$ is a viscosity subsolution of \eqref{F} on $[T_{k-1}, T_k) \times \cS_k$, that is, the USC function $u_k$ satisfies condition \eqref{F_sub} in Definition \ref{def:vis_sub}. Moreover, for $(x^o_{k:K+1}, y) \in \overline{\cS}^o_{k} \backslash \cS_k$ and $t \in [T_{k-1}, T_k)$,
		\begin{equation}\label{modi_bd_u}
			u_k(t, x^o_{k:K+1}, y) = \limsup_{\substack{ x'_{k:K+1} \rightarrow x^o_{k:K+1} \\ x'_j > 0, \, j = k,\ldots,K+1}} u_k(t, x'_{k:K+1}, y);
		\end{equation}
		\item $v_k \in LSC([T_{k-1}, T_k] \times \overline{\cS}_{k})$ is a viscosity supersolution of \eqref{F} on $[T_{k-1}, T_k) \times \overline{\cS}^o_k$, that is, the LSC function $v_k$ satisfies condition \eqref{F_super} in Definition \ref{def:vis_super};
		\item
		\begin{align}
			u_k(t, 0, y) & \leq v_k(t, 0, y), \quad \forall\;  t \in [T_{k-1}, T_k], \; y \in \R^m,\\
			u_k(T_k, x_{k:K+1}, y) & \leq v_k(T_k, x_{k:K+1}, y), \quad \forall\; (x_{k:K+1}, y) \in \overline{\cS}_{k};
		\end{align}
	\item $u_k$ and $v_k$ are bounded,
	\end{enumerate}
	then $v_k$ dominates $u_k$ in the sense that
	\begin{equation}
		u_k(t, x_{k:K+1}, y) \leq v_k(t, x_{k:K+1}, y), \quad \forall\; (t, x_{k:K+1}, y) \in [T_{k-1}, T_k] \times \overline{\cS}_{k}.
	\end{equation}
\end{proposition}

In Proposition \ref{prop:Tk}, we omit the time variable $T_k$ for simplicity.
\begin{proposition}[Terminal comparison at $T_k$]\label{prop:Tk}
	Given some $k=1, \ldots K$, consider a continuous and bounded function $f: \overline{\cS}_{k+1} \rightarrow \R$ and the PDE for $h: \overline{\cS}_{k} \rightarrow \R$:
		\begin{equation}\label{Tk_f}
		\begin{aligned}
			& F_k (x_{k:K+1}, y, h(x_{k:K+1}, y), \partial_{k:K+1} h(x_{k:K+1}, y)) \\
			& \quad := \max \Big\{ h(x_{k:K+1}, y) - w_k (G_k - x_k)^+ - f(x_{k+1:K+1}, y), \\
			& \quad \qquad \qquad - \lambda_{k} + \partial_{K+1} h - \partial_{k} h, \ldots, - \lambda_{K} + \partial_{K+1} h - \partial_{K} h, \\
			& \quad \qquad \qquad -\theta_{k} - \partial_{K+1} h + \partial_k h, \ldots, - \theta_K - \partial_{K+1} h + \partial_{K} h \Big\} = 0.
		\end{aligned}
	\end{equation}
	Suppose
	\begin{enumerate}[label={(\arabic*)}]	
		\item  $u_k \in USC(\overline{\cS}_{k})$ is a viscosity subsolution of \eqref{Tk_f} on $\cS_k$, i.e.,
			\begin{equation*}
				F_k (\bar{x}_{k:K+1}, \bar{y}, u_k(\bar{x}_{k:K+1}, \bar{y}), \partial_{k:K+1} \varphi(\bar{x}_{k:K+1}, \bar{y})) \leq 0
			\end{equation*}
			for all $(\bar{x}_{k:K+1}, \bar{y}) \in \cS_k$ and for all $\varphi \in C^{2}(\cS_k)$ such that $(\bar{x}_{k:K+1}, \bar{y})$ is a maximum point of $u_k - \varphi$. Moreover, for $(x^o_{k:K+1}, y) \in \overline{\cS}^o_{k} \backslash \cS_k$,
			\begin{equation}
				u_k(x^o_{k:K+1}, y) = \limsup_{\substack{ x'_{k:K+1} \rightarrow x^o_{k:K+1} \\ x'_j > 0, \, j = k,\ldots,K+1}} u_k(x'_{k:K+1}, y);
			\end{equation} 
		\item $v_k \in LSC(\overline{\cS}_{k})$ is a viscosity supersolution of \eqref{Tk_f} on $\overline{\cS}^o_k$, i.e.,
		\begin{equation*}
			F_k (\bar{x}_{k:K+1}, \bar{y}, v_k(\bar{x}_{k:K+1}, \bar{y}), \partial_{k:K+1} \varphi(\bar{x}_{k:K+1}, \bar{y})) \geq 0
		\end{equation*}
		for all $(\bar{x}_{k:K+1}, \bar{y}) \in \overline{\cS}^o_k$ and for all $\varphi \in C^{2}(\overline{\cS}^o_k)$ such that $(\bar{x}_{k:K+1}, \bar{y})$ is a minimum point of $v_k - \varphi$;
		\item
		\begin{align}
			u_k(0, y) & \leq v_k(0, y), \quad \forall \;  y \in \R^m;
		\end{align}
		\item $u_k$ and $v_k$ are bounded,
	\end{enumerate}
	then
	\begin{equation}
		u_k(x_{k:K+1}, y) \leq v_k(x_{k:K+1}, y), \quad \forall\; (x_{k:K+1}, y) \in \overline{\cS}_{k}.
	\end{equation}
\end{proposition}

The comparison principle when $t \in [T_K, T)$ is classical and much easier, since there are no gradient constraints involved. However, we could not find a precise reference applicable to our case. Therefore, we include the proof in Appendix for readers' convenience. The assumptions in Proposition \ref{prop:comp_last} are not the most general ones. 

\begin{proposition}[Comparison principle: $t \in [T_K, T)$]\label{prop:comp_last}
	Consider the HJB equation \eqref{Vlast} for $V_{K+1}(t, x_{K+1}, y)$. Suppose
	\begin{enumerate}[label={(\arabic*)}]	
	\item  $u_{K+1} \in USC([T_K, T] \times \overline{\cS}_{K+1})$ is a viscosity subsolution of \eqref{Vlast} on $[T_K, T) \times \cS_{K+1}$, that is, the USC function $u_{K+1}$ satisfies condition \eqref{Vlast_sub} in Definition \ref{def:vis_sub};
	\item $v_{K+1} \in LSC([T_K, T] \times \overline{\cS}_{K+1})$ is a viscosity supersolution of \eqref{Vlast} on $[T_K, T) \times \cS_{K+1}$, that is, the LSC function $v_{K+1}$ satisfies condition \eqref{Vlast_super} in Definition \ref{def:vis_super};
	\item
	\begin{align}
		u_{K+1}(t, 0, y) & \leq v_{K+1}(t, 0, y), \quad \forall\;  t \in [T_K, T], \; y \in \R^m,\\
		u_{K+1}(T, x_{K+1}, y) & \leq v_{K+1}(T, x_{K+1}, y), \quad \forall\; (x_{K+1}, y) \in \overline{\cS}_{K+1};
	\end{align}
	\item $u_{K+1}$ and $v_{K+1}$ are bounded,
\end{enumerate}
then
\begin{equation}
	u_{K+1}(t, x_{K+1}, y) \leq v_{K+1}(t, x_{K+1}, y), \quad \forall\; (t, x_{K+1}, y) \in [T_K, T] \times \overline{\cS}_{K+1}.
\end{equation}
\end{proposition}

We now present the proof of Theorem \ref{thm:viscosity}.
\begin{proof}
	First, using a similar argument as in \citet[Remark 3.2 and 4.2]{bayraktar2015stochastic}, we have for $k=1, \ldots, K+1$:
	\begin{align}
		v_{k, -}(t, x_{k:K+1}, y) & \leq V_{k}(t, x_{k:K+1}, y) \quad \text{ on } \quad [T_{k-1}, T_k] \times \overline{\cS}_k, \\
		V_{k}(t, x_{k:K+1}, y)  & \leq v_{k, +}(t, x_{k:K+1}, y) \quad \text{ on } \quad [T_{k-1}, T_k] \times (\cS_k \cup (\{ 0 \} \times \R^m)).
	\end{align}
	Thanks to Lemma \ref{lem:boundary}, the modified function $ \tilde{v}_{k, +} (t, x_{k:K+1}, y)$ defined in \eqref{eq:modified} satisfies
	\begin{align}
		V_{k}(t, x_{k:K+1}, y)  & \leq \tilde{v}_{k, +}(t, x_{k:K+1}, y) \quad \text{ on } \quad [T_{k-1}, T_k] \times \overline{\cS}_k, \quad k = 1, \ldots, K.
	\end{align}
	
	Moreover, the modified function $\tilde{v}_+ = (\tilde{v}_{1, +}, \ldots, \tilde{v}_{k, +}, \ldots, \tilde{v}_{K, +}, v_{K+1, +})$ is a viscosity subsolution of the HJB equation under Definition \ref{def:vis_sub}. By Proposition \ref{prop:vissuper}, the lower stochastic envelope $v_-$ is a viscosity supersolution of the HJB equation under Definition \ref{def:vis_super}. To obtain the continuity and uniqueness of the viscosity solution, we apply the comparison principle in Propositions \ref{prop:comp_last}, \ref{prop:Tk}, and \ref{p:comp_btw} backward in time: 
	\begin{enumerate}[label={(\arabic*)}]	
		\item By Propositions \ref{prop:vissub} and \ref{prop:vissuper}, $v_{K+1, +}$ is USC and $v_{K+1, -}$ is LSC. Moreover, 
		\begin{align}
			v_{K+1, +}(t, 0, y) & \leq v_{K+1, -}(t, 0, y), \quad \forall\;  t \in [T_K, T], \; y \in \R^m,\\
			v_{K+1, +}(T, x_{K+1}, y) & \leq v_{K+1, -}(T, x_{K+1}, y), \quad \forall\; (x_{K+1}, y) \in \overline{\cS}_{K+1};
		\end{align}
		Then Proposition \ref{prop:comp_last} implies 
		$$v_{K+1, +} (t, x_{K+1}, y) \leq v_{K+1, -} (t, x_{K+1}, y)  \quad \text{on} \quad [T_{K}, T] \times \overline{\cS}_{K+1}.$$ 
		It leads to
		$$v_{K+1, +}(t, x_{K+1}, y) = v_{K+1, -}(t, x_{K+1}, y) = V_{K+1}(t, x_{K+1}, y) \text{ on } [T_{K}, T] \times \overline{\cS}_{K+1}.$$
		Hence, $V_{K+1}$ is continuous and unique;
		\item  At the deadline $T_{K}$, Propositions \ref{prop:vissub} and \ref{prop:vissuper} show that $\tilde{v}_{K, +}(T_K, \cdot, \cdot)$ is USC and $v_{K, -}(T_K, \cdot, \cdot)$ is LSC. Moreover, $\tilde{v}_{K, +}(T_K, 0, y) \leq v_{K, -}(T_K, 0, y)$. By Lemmas \ref{lem:boundary}, \ref{lem:conti}, and step (1) above, we can apply Proposition \ref{prop:Tk} and obtain 
		$$\tilde{v}_{K, +}(T_K, x_{K:K+1}, y) = v_{K, -}(T_K, x_{K:K+1}, y) = V_K(T_K, x_{K:K+1}, y) \text{ on } \overline{\cS}_K.$$
		Then $V_K(T_K, \cdot, \cdot)$ is continuous in $(x_{K:K+1}, y)$ and unique;
		\item On $[T_{K-1}, T_K)$, Propositions \ref{prop:vissub} and \ref{prop:vissuper} show that $\tilde{v}_{K, +}$ is USC, $v_{K, -}$ is LSC, and
		\begin{align}
			\tilde{v}_{K, +}(t, 0, y) & \leq v_{K, -}(t, 0, y), \quad \forall\;  t \in [T_{K-1}, T_K], \; y \in \R^m.
		\end{align}
		Step (2) above shows that
		\begin{align}
			\tilde{v}_{K, +}(T_K, x_{K:K+1}, y) & \leq v_{K, -}(T_K, x_{K:K+1}, y), \quad \forall\; (x_{K:K+1}, y) \in \overline{\cS}_{K}.
		\end{align} 
		
		Hence, by Propositions \ref{prop:vissub} and \ref{prop:vissuper}, together with Lemmas \ref{lem:boundary} and \ref{lem:conti}, we can apply Proposition \ref{p:comp_btw} and show that
		$$\tilde{v}_{K, +}(t, x_{K:K+1}, y) = v_{K, -}(t, x_{K:K+1}, y) = V_K(t, x_{K:K+1}, y) \quad \text{ on } \quad [T_{K-1}, T_K] \times \overline{\cS}_{K}. $$
		Then $V_K$ is continuous in $(t, x_{K:K+1}, y)$ and unique;
		\item We repeat the procedures from steps (2) and (3) down to the first interval $[0, T_1]$. In summary, for each $k = 1, \ldots, K+1$, $V_k(t, x_{k:K+1}, y)$ is continuous and bounded on $[T_{k-1}, T_k] \times \overline{\cS}_k$.
	\end{enumerate}	
\end{proof}

{\color{black} To conclude the technical analysis, we highlight several methodological features that may be useful for other singular control problems with state constraints. First, the stochastic Perron's method is modified to work on a constrained state space, following the state-constraint perspective of \citet{soner1986optimal}. This modification is needed because some transfers are inadmissible at the boundary of the solvency region. Second, the proof handles goal deadlines at which the value function changes from \(V_k\) to \(V_{k+1}\), and the boundary condition must incorporate both the terminal goal penalty and the possibility of an instantaneous transfer. Third, the singular control structure leads to bilateral gradient constraints, reflecting the possibility of transferring wealth in both directions between a goal account and the fundamental account. Finally, the comparison argument uses a perturbed doubling-variable technique, inspired by \citet{hynd2012eigenvalue}, to handle these bilateral gradient constraints. These features suggest that the approach may be applicable to other multidimensional singular control problems in which state constraints, switching or deadline conditions, and bilateral intervention constraints appear together.}

{\color{black} 
	
\section{Conclusion}\label{sec:con}
This paper studies a continuous-time goal-based portfolio selection problem with mental accounting. The investor maintains separate portfolios for distinct goals, while transfers across accounts are allowed subject to mental costs. This formulation captures partial segmentation across goals and provides a tractable framework for studying how investors allocate wealth when different objectives have different deadlines, target levels, and priorities.

On the theoretical side, we develop a modified stochastic Perron's method for an HJB system with state-dimension reduction, state constraints, and bilateral gradient constraints. The resulting viscosity characterization identifies the value function and provides the basis for the numerical computation of portfolio policies, transfer regions, and free boundaries.

Several limitations and open questions remain. First, the empirical interpretation and calibration of mental cost parameters represent an important practical challenge, especially because such costs may depend on institutional design, market conditions, and investor experience. Second, the present model assumes deterministic deadlines and target amounts, whereas many real-world goals have stochastic or flexible horizons. Extending the analysis to random deadlines would introduce additional randomness into the terminal conditions of the HJB system and is a natural direction for future research. Third, the current model focuses on internal frictions generated by mental accounting. Incorporating external market frictions, such as transaction costs, taxes, or liquidity constraints, together with internal mental costs may generate richer economic behavior and more intricate boundary structures. Finally, establishing a rigorous convergence theory for the numerical scheme remains an important open problem.}


\bibliographystyle{apalike}
\bibliography{ref.bib}

\appendix
\section{Proofs of results} 

\subsection{Results of the stochastic supersolution}\label{sec:proof_stosuper}

\begin{proof}[Proof of Proposition \ref{prop:vissub}]
	{\bf Step 1}. The boundary/terminal conditions in Definition \ref{def:vis_sub} (4): {\color{black} Recall the example of a stochastic supersolution in \eqref{ex:sto_super}.} Since $v_+$ is the infimum of stochastic supersolutions, $v_+$ satisfies the boundary condition \eqref{vissub_bd0} at $x_{k:K+1} = 0$. By a modification of the proof in \citet[Theorem 3.1]{bayraktar2013stochastic}, it follows that $v_{K+1, +}$ in $v_+$ satisfies the terminal condition \eqref{visub_T} at $T$.
	
	{\bf Step 2}. The viscosity subsolution property in Definition \ref{def:vis_sub} (3), when $t \in [T_K, T)$: Similarly, it follows from the proof of \citet[Theorem 3.1]{bayraktar2013stochastic}.
	
	{\bf Step 3}. The viscosity subsolution property in Definition \ref{def:vis_sub} (2), when $t=T_K$: Let $(\bar{x}_{K:K+1}, \bar{y}) \in \cS_K = (0, \infty)^2 \times \R^m$ and $\varphi(x_{K:K+1}, y) \in C^2(\cS_K)$ be a test function such that $v_{K, +}(T_K, \cdot, \cdot) - \varphi(\cdot, \cdot)$ attains a strict local maximum of zero at $(\bar{x}_{K:K+1}, \bar{y})$. We want to prove
	\begin{equation}
		\begin{aligned}
			\max \Big\{ & v_{K, +}(T_K, \bar{x}_{K:K+1}, \bar{y}) - w_K (G_K - \bar{x}_K)^+ - v_{K+1, +}(T_K, \bar{x}_{K+1}, \bar{y}), \\
			& - \lambda_{K} + \partial_{K+1} \varphi(\bar{x}_{K:K+1}, \bar{y}) - \partial_{K} \varphi(\bar{x}_{K:K+1}, \bar{y}), \\
			&- \theta_K - \partial_{K+1} \varphi(\bar{x}_{K:K+1}, \bar{y}) + \partial_{K} \varphi (\bar{x}_{K:K+1}, \bar{y}) \Big\} \leq 0.
		\end{aligned}
	\end{equation}
	Assume on the contrary that the left-hand side is strictly positive. There are three cases to consider: 
	
	{\bf Step 3: Case (1)}. $ v_{K, +}(T_K, \bar{x}_{K:K+1}, \bar{y}) - w_K (G_K - \bar{x}_K)^+ - v_{K+1, +}(T_K, \bar{x}_{K+1}, \bar{y}) > 0$. 
	
	As $v_{K+1, +}$ is USC, there exists a small $\varepsilon > 0$ such that
	\begin{equation*}
		v_{K, +}(T_K, \bar{x}_{K:K+1}, \bar{y}) \geq w_K (G_K - x_K)^+ + v_{K+1, +}(T_K, x_{K+1}, y) + \varepsilon
	\end{equation*}
	on the compact set $\overline{B(\bar{x}_{K:K+1}, \bar{y}, \varepsilon)}$, the closure of 
	\begin{equation}\label{Beps}
		B(\bar{x}_{K:K+1}, \bar{y}, \varepsilon) : = \{(x_{K:K+1}, y): |(x_{K:K+1}, y) - (\bar{x}_{K:K+1}, \bar{y})| < \varepsilon \}.
	\end{equation}
	Since $(\bar{x}_{K:K+1}, \bar{y})$ is in the interior $\cS_{K}$ and $\varepsilon$ is small enough, we can choose $\overline{B(\bar{x}_{K:K+1}, \bar{y},\varepsilon)} \subset \cS_K$. By \citet[Proposition 4.1]{bayraktar2012linear}, we obtain a nonincreasing sequence of stochastic supersolutions $v^n_{K+1} \searrow v_{K+1, +}$. Moreover, every $v^n_{K+1}$ has a corresponding stochastic supersolution $v^n = (v^n_1, \ldots, v^n_K, v^n_{K+1})$. Indeed, we can specify the function class $\mathcal{G}$ in \citet[Proposition 4.1]{bayraktar2012linear} to be functions $v_{K+1}$ that can form a stochastic supersolution $v = (v_1, \ldots, v_K, v_{K+1})$ with some $(v_1, \ldots, v_{K})$, which can be different for different $v_{K+1}$. By \citet[Lemma 2.4]{bayraktar2014Dynkin}, for the given compact set $\overline{B(\bar{x}_{K:K+1}, \bar{y},\varepsilon)}$ defined above, there exists a large enough $n_1$ such that $v^{n_1}_{K+1}$ is close enough to $v_{K+1, +}$ and yields 
	\begin{equation}\label{n1}
		v_{K, +}(T_K, \bar{x}_{K:K+1}, \bar{y}) \geq w_K (G_K - x_K)^+ + v^{n_1}_{K+1}(T_K, x_{K+1}, y) + \frac{\varepsilon}{2}
	\end{equation}  
	on $\overline{B(\bar{x}_{K:K+1}, \bar{y},\varepsilon)}$. Besides, $v^{n_1}_{K+1}$ corresponds to a stochastic supersolution $v^{n_1} = (v^{n_1}_1, \ldots$, $ v^{n_1}_K, v^{n_1}_{K+1})$. 
	
	Define sets
	\begin{align*}
		D(T_K, \bar{x}_{K:K+1}, \bar{y}, \varepsilon) & := (T_K - \varepsilon, T_K] \times B(\bar{x}_{K:K+1}, \bar{y}, \varepsilon), \\
		E(\varepsilon) & := \overline{D(T_K, \bar{x}_{K:K+1}, \bar{y},\varepsilon)} \backslash D(T_K, \bar{x}_{K:K+1}, \bar{y},\varepsilon/2). 
	\end{align*}

	Since $v_{K, +}$ is USC and $E(\varepsilon)$ is compact, $v_{K, +}$ is bounded from above on $E(\varepsilon)$. For a small enough $\eta > 0$, we obtain
	\begin{equation*}
		\sup_{(t, x_{K:K+1}, y) \in E(\varepsilon)} v_{K, +}(t, x_{K:K+1}, y) - v_{K, +} (T_K, \bar{x}_{K:K+1}, \bar{y}) < \frac{\varepsilon^2}{4 \eta} - \varepsilon. 
	\end{equation*}
	As this inequality is strict, we can find another $v^{n_2}_K$, which corresponds to a stochastic supersolution $v^{n_2} = (v^{n_2}_1, \ldots, v^{n_2}_K, v^{n_2}_{K+1})$, and
	\begin{equation}\label{n2}
		\sup_{(t, x_{K:K+1}, y) \in E(\varepsilon)} v^{n_2}_{K}(t, x_{K:K+1}, y) - v_{K, +} (T_K, \bar{x}_{K:K+1}, \bar{y}) < \frac{\varepsilon^2}{4 \eta} - \varepsilon. 
	\end{equation}
	Finally, we take
	$$v^n := (v^{n}_1, \ldots, v^{n}_{K+1}) := (v^{n_1}_1 \wedge v^{n_2}_1, \ldots, v^{n_1}_{K+1} \wedge v^{n_2}_{K+1}),$$
	which is a stochastic supersolution by Lemma \ref{lem:two_stochsup}. The inequalities \eqref{n1} and \eqref{n2} also hold for $v^n_{K+1}$ and $v^n_K$, respectively.
	
	We introduce an operator
	\begin{equation}
		\begin{aligned}
			& [L^{\alpha_{k:K+1}}V_k] (x_{k:K+1}, y) \\
			& \quad := \sum^{K+1}_{i=k} r x_{i} \partial_i V_k + \mu_Y(y)^\top \partial_y V_k  \\
			& \qquad +  \sum^{K+1}_{i=k}  (\mu(y) - r \one )^\top \alpha_i x_i \partial_i V_k + \frac{1}{2} \tr\left[\Sigma(\alpha_{k:K+1}, x_{k:K+1}, y) \partial^2 V_k  \right] ,
		\end{aligned}
	\end{equation}
	such that the Hamiltonian can be expressed as
	\begin{equation}
		\begin{aligned}
			& H(x_{k:K+1}, y, \partial V_k, \partial^2 V_k) = \inf_{\alpha_{k:K+1} \in \cA^{K-k+2}} [L^{\alpha_{k:K+1}}V_k] (x_{k:K+1}, y).
		\end{aligned}
	\end{equation}
	
	For $p > 0$, define
	\begin{equation*}
		\psi^{\varepsilon, \eta, p}(t, x_{K:K+1}, y) := v_{K,+} (T_K, \bar{x}_{K:K+1}, \bar{y}) + \frac{|(x_{K:K+1}, y) -  (\bar{x}_{K:K+1}, \bar{y})|^2}{\eta} + p(T_K - t).
	\end{equation*}
	With a large enough $p$,
	\begin{equation*}
		\beta \psi^{\varepsilon, \eta, p} - \partial_t \psi^{\varepsilon, \eta, p} - \inf_{\alpha_{K:K+1} \in \cA^{2}} L^{\alpha_{K:K+1}} \psi^{\varepsilon, \eta, p} > 0 \quad \text{ on } \overline{D(T_K, \bar{x}_{K:K+1}, \bar{y},\varepsilon)}.
	\end{equation*}
	Moreover, there exists a constant control $\bar{\alpha}_{K:K+1}$ such that
	\begin{equation}
		\beta \psi^{\varepsilon, \eta, p} - \partial_t \psi^{\varepsilon, \eta, p} - L^{\bar{\alpha}_{K:K+1}} \psi^{\varepsilon, \eta, p} > 0 \quad \text{ on } \overline{D(T_K, \bar{x}_{K:K+1}, \bar{y},\varepsilon)}.
	\end{equation}
	
	By the definition of set $E(\varepsilon)$, the inequality \eqref{n2}, and a large enough $p$,
	\begin{align}
		\psi^{\varepsilon, \eta, p}(t, x_{K:K+1}, y) & \geq v_{K, +} (T_K, \bar{x}_{K:K+1}, \bar{y}) + \frac{\varepsilon^2}{4 \eta} \nonumber \\
		& > \varepsilon + \sup_{(t, x_{K:K+1}, y) \in E(\varepsilon)} v^{n}_{K}(t, x_{K:K+1}, y) \nonumber \\
		& \geq \varepsilon + v^n_K (t, x_{K:K+1}, y) \quad \text{ on } E(\varepsilon). \label{Eineq}
	\end{align}
	
	Besides, for any $t \leq T_K$ and $(x_{K:K+1}, y) \in \overline{B(\bar{x}_{K:K+1}, \bar{y},\varepsilon)}$, \eqref{n1} leads to
	\begin{align}
		\psi^{\varepsilon, \eta, p}(t, x_{K:K+1}, y) & \geq v_{K, +} (T_K, \bar{x}_{K:K+1}, \bar{y}) \nonumber \\
		& \geq w_K (G_K - x_K)^+ + v^{n}_{K+1}(T_K, x_{K+1}, y) + \frac{\varepsilon}{2}. \label{psi_bd}
	\end{align}
	
	Let $0 <\delta < \frac{\varepsilon}{2}$ be small enough such that
	\begin{equation}
		\beta (\psi^{\varepsilon, \eta, p} - \delta) - \partial_t (\psi^{\varepsilon, \eta, p} - \delta) - L^{\bar{\alpha}_{K:K+1}} (\psi^{\varepsilon, \eta, p} - \delta) > 0 \quad \text{ on } \overline{D(T_K, \bar{x}_{K:K+1}, \bar{y},\varepsilon)}.
	\end{equation}
	Set
	\begin{equation}
		v^{\varepsilon, \eta, p, \delta}_{K} (t, x_{K:K+1}, y) := \left\{ 
		\begin{array}{ c l }
			v^n_K (t, x_{K:K+1}, y) \wedge (\psi^{\varepsilon, \eta, p}(t, x_{K:K+1}, y) - \delta) & \text{on } \overline{D(T_K, \bar{x}_{K:K+1}, \bar{y},\varepsilon)}, \\
			v^n_K (t, x_{K:K+1}, y), & \text{otherwise}.
		\end{array}
		\right.
	\end{equation}
	
	Next, we show that $(v^n_1, \ldots, v^{\varepsilon, \eta, p, \delta}_{K}, v^n_{K+1})$ is a stochastic supersolution. Then it leads to the following contradiction:
	$$v^{\varepsilon, \eta, p, \delta}_{K} (T_K, \bar{x}_{K:K+1}, \bar{y}) = v_{K, +} (T_K, \bar{x}_{K:K+1}, \bar{y})  - \delta < v_{K, +} (T_K, \bar{x}_{K:K+1}, \bar{y}).$$
	
	Clearly, $(v^n_1, \ldots, v^{\varepsilon, \eta, p, \delta}_{K}, v^n_{K+1})$ satisfies the boundary/terminal conditions in Definition \ref{def:sto_super} (2). For the supermartingale property in Definition \ref{def:sto_super} (1), we first verify it when the random initial condition $(\tau, \xi_{K:K+1}, \xi_y)$ satisfies $\tau \in [T_{K-1}, T_K]$.
	
	Define $\psi^{p, \delta} := \psi^{\varepsilon, \eta, p} - \delta$ and the event
	\begin{equation*}
		A := \{ (\tau, \xi_{K:K+1}, \xi_y) \in D(T_K, \bar{x}_{K:K+1}, \bar{y},\varepsilon/2) \} \cap \{\psi^{p, \delta}(\tau, \xi_{K:K+1}, \xi_y) < v^n_K(\tau, \xi_{K:K+1}, \xi_y)\}.
	\end{equation*}
	Then $A \in \cF_\tau$. 
	
	Let $U^0 := (\alpha^0_{K:K+1}, L^0_K, M^0_K)$ be a suitable control for $v^n_{K}$ with the random initial condition $(\tau, \xi_{K:K+1}, \xi_y)$. Define a new control $U^1 := (\alpha^1_{K:K+1}, L^1_K, M^1_K)$ by
	\begin{equation}\label{U1}
		\begin{aligned}
			\alpha^1_{K:K+1}(t) & := \bar{\alpha}_{K:K+1} \one_{A \cap \{ \tau \leq t\}} + \alpha^0_{K:K+1}(t) \one_{A^c \cap \{ \tau \leq t\}}, \\
			(L^1_K(t), M^1_K(t)) & := \one_{A^c \cap \{ \tau \leq t\}} (L^0_K(t) - L^0_K(\tau-), M^0_K(t) - M^0_K(\tau-)).
		\end{aligned}
	\end{equation}
	By Lemma \ref{lem:glue}, we have $U^1 \in \cU (\tau, \xi_{K:K+1}, \xi_y)$. Denote $\{ (X_{K:K+1}, Y)(t; \tau, \xi_{K:K+1}, \xi_y, U^1) \}_{t \in [\tau, T]}$ as the solution of the state process with the random initial condition $(\tau, \xi_{K:K+1}, \xi_y)$ under the control $U^1$. Then $\p((X_{K:K+1}, Y)(t; \tau, \xi_{K:K+1}, \xi_y, U^1) \in \cS_{K}, \; \tau \leq t \leq T ) = 1$.

	Moreover, if $A^c$ happens, then $v^{\varepsilon, \eta, p, \delta}_{K} (\tau, \xi_{K:K+1}, \xi_y) = v^n_{K} (\tau, \xi_{K:K+1}, \xi_y)$. Hence, under $A^c$, $U^1$ follows $U^0$, which is a suitable control for $v^n_K$ starting from time $\tau$.
	
	Let 
	\begin{align*}
		\tau^1 := \inf \{ t \in [\tau, T_K] \, | \,  (t,  (X_{K:K+1}, Y)(t; \tau, \xi_{K:K+1}, \xi_y, U^1)) \notin  D(T_K, \bar{x}_{K:K+1}, \bar{y}, \varepsilon/2) \} {\color{black} \wedge T_K}
	\end{align*}
	be the exit time of $D(T_K, \bar{x}_{K:K+1}, \bar{y},\varepsilon/2)$ and 
	\begin{align*}
		\xi^1 := (\xi^1_{K:K+1}, \xi^1_y) :=  ((X_{K:K+1}, Y)(\tau^1; \tau, \xi_{K:K+1}, \xi_y, U^1))
	\end{align*}
	be the exit position, which is $\cF_{\tau^1}$-measurable.  
	
	Since $\tau^1 \leq T_K$ and $\xi^1$ is $\cF_{\tau^1}$-measurable, then $(T_K, \xi^1_{K+1}, \xi^1_y)$ is a random initial condition. As $\p(\xi^1 \in \cS_{K}) = 1$, there exists a suitable control $U^2 := \{ \alpha^2_{K+1}(t) \}_{t \in [T_K, T]}$ for $v^n_{K+1}$ with the random initial condition $(T_K, \xi^1_{K+1}, \xi^1_y)$. $(\tau^1, \xi^1)$ is also a random initial condition. As $\p(\xi^1 \in \cS_{K}) = 1$, there is a suitable control $U^3 := ( \alpha^3_{K:K+1}, L^3_K, M^3_K)$ for $v^n_{K}$ with the random initial condition $(\tau^1, \xi^1)$. Finally, define a control $U := (\alpha_{K:K+1}, L_K, M_K)$ by
	\begin{equation}\label{eq:step3-1-U}
		\begin{aligned}
			& (\alpha_{K:K+1}(t), L_K(t), M_K(t)) \\
			& \quad  = \one_{\{ \tau \leq t < \tau^1 \}} U^1(t) + \one_{ \{ \tau^1 = T_K\} \cap \{ \tau^1 \leq t \leq T\} } (U^2(t), L^1_K(\tau^1), M^1_K(\tau^1)) \\
			& \qquad + \one_{\{ \tau^1 < T_K\}  \cap \{ \tau^1 \leq t \leq T\} } (\alpha^3_{K:K+1}(t), L^3_K(t) - L^3_K(\tau^1-) + L^1_K(\tau^1), \\ 
			& \hspace{4.5cm} M^3_K(t) - M^3_K(\tau^1-) + M^1_K(\tau^1) ).
		\end{aligned}
	\end{equation}

	By Lemma \ref{lem:glue}, the control $U \in \cU(\tau, \xi_{K:K+1}, \xi_y)$. Moreover, $\p((X_{K:K+1}, Y)(t; \tau, \xi_{K:K+1}, \xi_y, U) \in \cS_{K}, \; \tau \leq t \leq T ) = 1$. We verify that $U$ is suitable for $v^{\varepsilon, \eta, p, \delta}_{K}$ with $(\tau, \xi_{K:K+1}, \xi_y)$.
	
	 Consider a stopping time $\rho \in [\tau, T]$. Applying It\^o's formula to $\psi^{p, \delta}$ from $\tau$ to $\rho \wedge \tau^1$ under the event $A$, we obtain
	\begin{align}
		& \one_A v^{\varepsilon, \eta, p, \delta}_{K} (\tau, \xi_{K:K+1}, \xi_y) \nonumber \\
		& = \one_A \psi^{p, \delta} (\tau, \xi_{K:K+1}, \xi_y) \nonumber \\
		& = \one_A \psi^{p, \delta} (\tau, (X_{K:K+1}, Y)(\tau; \tau, \xi_{K:K+1}, \xi_y, U^1)) \nonumber \\
		 & \geq \E\Big[ \one_{A \cap \{ \rho < \tau^1\}} e^{- \beta(\rho - \tau)}  \psi^{p, \delta} (\rho, (X_{K:K+1}, Y)(\rho; \tau, \xi_{K:K+1}, \xi_y, U^1)) \label{psiAineq1} \\
		& \qquad + \one_{A \cap \{ \rho \geq \tau^1\}}  e^{- \beta(\tau^1 - \tau)}  \psi^{p, \delta} (\tau^1, \xi^1) \Big| \cF_\tau \Big]. \nonumber
	\end{align} 
	Moreover, \eqref{Eineq} and \eqref{psi_bd} lead to
	\begin{align}
		\one_{A \cap \{ \rho \geq \tau^1\}} \psi^{p, \delta} (\tau^1, \xi^1) \geq & \one_{A \cap \{ \rho \geq \tau^1\} \cap \{ \tau^1 < T_K\}} v^n_K(\tau^1, \xi^1) \label{psiAineq2} \\
		& + \one_{A \cap \{ \rho \geq \tau^1\} \cap \{ \tau^1 = T_K\}} \big( w_K(G_K - \xi^1_K)^+ + v^n_{K+1}(T_K, \xi^1_{K+1}, \xi^1_y) \big). \nonumber
	\end{align}
	Combining \eqref{psiAineq1} and \eqref{psiAineq2}, since $v^{\varepsilon, \eta, p, \delta}_{K} \leq \psi^{p, \delta}$ on $\overline{D(T_K, \bar{x}_{K:K+1}, \bar{y},\varepsilon)}$, we obtain
		\begin{align}
		& \one_A v^{\varepsilon, \eta, p, \delta}_{K} (\tau, \xi_{K:K+1}, \xi_y) \nonumber \\
		& \geq \E\Big[ \one_{A \cap \{ \rho < \tau^1\}} e^{- \beta(\rho - \tau)}  v^{\varepsilon, \eta, p, \delta}_{K} (\rho, (X_{K:K+1}, Y)(\rho; \tau, \xi_{K:K+1}, \xi_y, U^1)) \label{psiAineq} \\
		& \qquad + \one_{A \cap \{ \rho \geq \tau^1\} \cap \{ \tau^1 < T_K\}}  e^{- \beta(\tau^1 - \tau)}  v^n_K(\tau^1, \xi^1) \nonumber \\
		& \qquad + \one_{A \cap \{ \rho \geq \tau^1\} \cap \{ \tau^1 = T_K\}}  e^{- \beta(T_K - \tau)}  \big( w_K(G_K - \xi^1_K)^+ + v^n_{K+1}(T_K, \xi^1_{K+1}, \xi^1_y) \big)  \Big| \cF_\tau \Big]. \nonumber
	\end{align} 
	
	Under the event $A^c$, because $U^1$ is a suitable control for $v^n_{K}$ with the random initial condition $(\tau, \xi_{K:K+1}, \xi_y)$, we have
	\begin{align}
		& \one_{A^c} v^{\varepsilon, \eta, p, \delta}_{K} (\tau, \xi_{K:K+1}, \xi_y) = \one_{A^c} v^{n}_{K} (\tau, \xi_{K:K+1}, \xi_y) \nonumber \\
		& \geq \E\Big[ \one_{A^c \cap \{ \rho < \tau^1\}} e^{- \beta(\rho - \tau)}  v^{n}_{K} (\rho, (X_{K:K+1}, Y)(\rho; \tau, \xi_{K:K+1}, \xi_y, U^1))  \label{Acvnineq} \\
		& \qquad + \one_{A^c \cap \{ \rho \geq \tau^1\} \cap \{ \tau^1 < T_K \}} e^{- \beta(\tau^1 - \tau)}  v^{n}_{K} (\tau^1, \xi^1) \nonumber \\
		& \qquad + \one_{A^c \cap \{ \rho \geq \tau^1\} \cap \{ \tau^1 = T_K \}} e^{- \beta(T_K - \tau)}  \big( w_K(G_K - \xi^1_K)^+ + v^n_{K+1}(T_K, \xi^1_{K+1}, \xi^1_y) \big)  \nonumber \\
		& \qquad + \lambda_{K} \int^{\rho \wedge \tau^1}_\tau e^{-\beta (s - \tau)} dL^1_K(s) + \theta_{K} \int^{\rho \wedge \tau^1}_\tau e^{-\beta (s - \tau)} dM^1_K(s) \Big| \cF_\tau \Big]. \nonumber
	\end{align} 
	As $v^n_K \geq v^{\varepsilon, \eta, p, \delta}_{K}$ everywhere and $U = U^1$ when $\rho < \tau^1$,  \eqref{psiAineq} and \eqref{Acvnineq} yield
		\begin{align}
		& v^{\varepsilon, \eta, p, \delta}_{K} (\tau, \xi_{K:K+1}, \xi_y) \nonumber \\
		& \geq \E\Big[ \one_{\{ \rho < \tau^1\}} e^{- \beta(\rho - \tau)}  v^{\varepsilon, \eta, p, \delta}_{K} (\rho, (X_{K:K+1}, Y)(\rho; \tau, \xi_{K:K+1}, \xi_y, U))  \label{ineq:tau1} \\
		& \qquad + \one_{ \{ \rho \geq \tau^1\} \cap \{ \tau^1 < T_K \}} e^{- \beta(\tau^1 - \tau)}  v^{n}_{K} (\tau^1, \xi^1) \nonumber \\
		& \qquad + \one_{\{ \rho \geq \tau^1\} \cap \{ \tau^1 = T_K \}} e^{- \beta(T_K - \tau)}  \big( w_K(G_K - \xi^1_K)^+ + v^n_{K+1}(T_K, \xi^1_{K+1}, \xi^1_y) \big)  \nonumber \\
		& \qquad +\lambda_{K} \int^{\rho \wedge \tau^1}_\tau e^{-\beta (s - \tau)} dL^1_K(s) + \theta_{K} \int^{\rho \wedge \tau^1}_\tau e^{-\beta (s - \tau)} dM^1_K(s) \Big| \cF_\tau \Big]. \nonumber
	\end{align} 
	
	Since $U^2$ is a suitable control for $v^n_{K+1}$ with the random initial condition $(T_K, \xi^1_{K+1}, \xi^1_y)$ and $U^3$ is a suitable control for $v^n_{K}$ with the random initial condition $(\tau^1, \xi^1)$, \eqref{ineq:tau1} and the definition of $U$ yield the desired result:
	\begin{align*}
		v^{\varepsilon, \eta, p, \delta}_{K} (\tau, \xi_{K:K+1}, \xi_y) \geq \E\Big[& \one_{\{ \tau \leq \rho < T_K \}} \Big\{ e^{- \beta(\rho - \tau)}  v^{\varepsilon, \eta, p, \delta}_{K} (\rho, (X_{K:K+1}, Y)(\rho; \tau, \xi_{K:K+1}, \xi_y, U))  \\
		& \quad \qquad \qquad + \lambda_{K} \int^\rho_\tau e^{-\beta (s - \tau)} dL_K(s) + \theta_{K} \int^\rho_\tau e^{-\beta (s - \tau)} dM_K(s) \Big\} \nonumber \\
		& + \one_{\{ T_K \leq \rho \leq T\}} \Big\{ w_K e^{-\beta(T_K - \tau)} (G_K - X_K(T_K; \tau, \xi_{K:K+1}, \xi_y, U))^+ \nonumber \\
		& \hspace{2.5cm} + e^{-\beta(\rho - \tau)}  v^n_{K+1}(\rho, (X_{K+1}, Y)(\rho; \tau, \xi_{K:K+1}, \xi_y, U)) \nonumber \\
		& \hspace{2.5cm} +  \lambda_K \int^{T_K}_{\tau} e^{-\beta(s - \tau)} dL_K(s) \nonumber \\
		& \hspace{2.5cm} + \theta_K \int^{T_K}_{\tau}  e^{-\beta(s - \tau)} dM_K(s) \Big\} \Big| \cF_\tau \Big].
	\end{align*}
	It is straightforward to verify the supermartingale property when $\tau \in [T_{k-1}, T_k]$, $k \neq K$. We omit it here.

	{\bf Step 3: Case (2)}. $- \lambda_{K} + \partial_{K+1} \varphi(\bar{x}_{K:K+1}, \bar{y}) - \partial_{K} \varphi(\bar{x}_{K:K+1}, \bar{y}) > 0$.
	
	Since $\varphi$ is $C^2$, there exists a small closed ball $\overline{B(\bar{x}_{K:K+1}, \bar{y},\varepsilon)} \subset \cS_K$, such that
	\begin{equation*}
		- \lambda_{K} + \partial_{K+1} \varphi(x_{K:K+1}, y) - \partial_{K} \varphi(x_{K:K+1}, y) > 0 \quad \text{ on } \quad \overline{B(\bar{x}_{K:K+1}, \bar{y},\varepsilon)}.
	\end{equation*} 
	Since the maximum at $(\bar{x}_{K:K+1}, \bar{y})$ is strict, we have
	\begin{equation*}
		v_{K, +}(T_K, x_{K:K+1}, y) < \varphi(x_{K:K+1}, y) \quad \text{ on }  \overline{B(\bar{x}_{K:K+1}, \bar{y},\varepsilon)} \backslash \{ (\bar{x}_{K:K+1}, \bar{y}) \}.
	\end{equation*}
	Define
	\begin{equation*}
		\varphi^p(t, x_{K:K+1}, y) := \varphi(x_{K:K+1}, y) + p(T_K - t).
	\end{equation*}
	We introduce a function
	\begin{equation*}
		h(\delta) := \sup_{\substack{ T_K - \delta \leq t \leq T_K, \\ \frac{\varepsilon}{2} \leq | (x_{K:K+1}, y) -  (\bar{x}_{K:K+1}, \bar{y})| \leq \varepsilon }} \Big( v_{K, +}(t, x_{K:K+1}, y) - \varphi(x_{K:K+1}, y) \Big), \quad 0 < \delta < \delta_0,
	\end{equation*}
	with some small $\delta_0 > 0$. Since $v_{K, +} - \varphi$ is USC as a function of $(t, x_{K:K+1}, y)$, then the maximum in $h(\delta)$ is attained. Let a maximizer be $(t^\delta, x^\delta_{K:K+1}, y^\delta)$. By compactness, there exists a subsequence, still denoted as $\delta \searrow 0$, such that 
	\begin{equation*}
		(t^\delta, x^\delta_{K:K+1}, y^\delta) \rightarrow (T_K, x^*_{K:K+1}, y^*) \quad \text{and} \quad \frac{\varepsilon}{2} \leq | (x^*_{K:K+1}, y^*) -  (\bar{x}_{K:K+1}, \bar{y})| \leq \varepsilon.
	\end{equation*}
	Thus,
	\begin{align}
		\limsup_{\delta \searrow 0} h(\delta) & = \limsup_{\delta \searrow 0} \Big( v_{K, +}(t^\delta, x^\delta_{K:K+1}, y^\delta) - \varphi(x^\delta_{K:K+1}, y^\delta) \Big) \nonumber \\
		& \leq v_{K, +} (T_K, x^*_{K:K+1}, y^*) - \varphi(x^*_{K:K+1}, y^*) \label{usc_in}\\
		& \leq \sup_{\frac{\varepsilon}{2} \leq | (x_{K:K+1}, y) -  (\bar{x}_{K:K+1}, \bar{y})| \leq \varepsilon} \Big( v_{K, +}(T_K, x_{K:K+1}, y) - \varphi(x_{K:K+1}, y) \Big) < 0. \label{strict}
	\end{align}
	\eqref{usc_in} follows from the USC property. \eqref{strict} holds since the local maximum at $(\bar{x}_{K:K+1}, \bar{y})$ is assumed to be strict. 
	
	\eqref{strict} implies that there exists $\delta$ small enough such that $h(\delta) < 0$. For this fixed $\delta$, we define $\delta_h := - h(\delta) > 0$. When $T_K - \delta \leq t \leq T_K$ and $\frac{\varepsilon}{2} \leq | (x_{K:K+1}, y) -  (\bar{x}_{K:K+1}, \bar{y})| \leq \varepsilon$, we have
	\begin{equation}\label{region1} 
		v_{K, +}(t, x_{K:K+1}, y) - \varphi(x_{K:K+1}, y) \leq \delta_h < 0.
	\end{equation}
	
	When $T_K - \delta \leq t \leq T_K - \delta/2$, $|(x_{K:K+1}, y) -  (\bar{x}_{K:K+1}, \bar{y})| \leq \varepsilon/2$ and $p$ is large enough, we have
	\begin{align}
		& v_{K, +}(t, x_{K:K+1}, y) - \varphi^p(t, x_{K:K+1}, y) \nonumber \\
		& \quad = v_{K, +}(t, x_{K:K+1}, y) - \varphi(x_{K:K+1}, y) - p(T_K - t) \nonumber \\
		& \quad \leq \sup_{\substack{T_K - \delta \leq t \leq T_K - \delta/2, \\  | (x_{K:K+1}, y) -  (\bar{x}_{K:K+1}, \bar{y})| \leq \frac{\varepsilon}{2} }} \Big( v_{K, +}(t, x_{K:K+1}, y) - \varphi(x_{K:K+1}, y) \Big) - p(T_K - t) \leq - \delta_h. \label{region2} 
	\end{align}
	Therefore, \eqref{region1} and \eqref{region2} indicate that
	\begin{equation}\label{deltah_ineq}
		v_{K, +}(t, x_{K:K+1}, y) - \varphi^p(t, x_{K:K+1}, y) \leq - \delta_h 
	\end{equation} 
	holds on 
	\begin{align*}
		E(\delta, \varepsilon) :=&  \big\{ (t, x_{K:K+1}, y) \big|\, t \in [T_K - \delta, T_K], \, | (x_{K:K+1}, y) -  (\bar{x}_{K:K+1}, \bar{y})| \leq \varepsilon \big\} \\
		& - \big\{ (t, x_{K:K+1}, y) \big|\, t \in (T_K - \delta/2, T_K], \, | (x_{K:K+1}, y) -  (\bar{x}_{K:K+1}, \bar{y})| < \varepsilon/2 \big\}.
	\end{align*}
	
	Moreover, by letting $p$ even larger, we have
	\begin{equation*}
		\beta \varphi^p - \partial_t \varphi^p - H(x_{K:K+1}, y, \partial \varphi^p, \partial^2 \varphi^p) > 0 \quad \text{ on } [T_K - \delta, T_K] \times \overline{B(\bar{x}_{K:K+1}, \bar{y},\varepsilon)}.
	\end{equation*}
	Hence, there exists a constant control $\bar{\alpha}_{K:K+1}$ such that
	\begin{equation}\label{phip}
		\beta \varphi^p - \partial_t \varphi^p - L^{\bar{\alpha}_{K:K+1}} \varphi^p > 0 \quad \text{ on } [T_K - \delta, T_K] \times \overline{B(\bar{x}_{K:K+1}, \bar{y},\varepsilon)}.
	\end{equation}
	Using a Dini argument, \eqref{deltah_ineq} implies that there exists $v^n_K$, which is the $K$-th component of a stochastic supersolution $v^n := (v^n_1, \ldots, v^n_{K+1})$, such that
	\begin{equation}\label{phi_E}
		v^n_K(t, x_{K:K+1}, y) - \varphi^p(t, x_{K:K+1}, y) \leq -\delta_h/2 \quad \text{ on } E(\delta, \varepsilon).
	\end{equation}
	
	Let $0 < \eta < \delta_h/2$ be small enough and define $\varphi^{p, \eta} := \varphi^p - \eta$. \eqref{phip} also holds for $\varphi^{p, \eta}$. Define
	\begin{equation}
		v^{p, \eta}_{K} (t, x_{K:K+1}, y) := \left\{ 
		\begin{array}{ c l }
			v^n_K (t, x_{K:K+1}, y) \wedge \varphi^{p, \eta}(t, x_{K:K+1}, y) & \text{on } [T_K - \delta, T_K] \times \overline{B(\bar{x}_{K:K+1}, \bar{y},\varepsilon)}, \\
			v^n_K (t, x_{K:K+1}, y), & \text{otherwise}.
		\end{array}
		\right.
	\end{equation}
	We use a similar argument as in Case (1) to show $(v^n_1, \ldots, v^{p, \eta}_{K}, v^n_{K+1})$ is a stochastic supersolution. Only the supermartingale property is non-trivial. We first verify it when the random initial condition $(\tau, \xi_{K:K+1}, \xi_y)$ satisfies $\tau \in [T_{K-1}, T_K]$.
	
	Define the event 
	\begin{equation}\label{newA}
		\begin{aligned}
			A := & \{ (\tau, \xi_{K:K+1}, \xi_y) \in (T_K - \delta/2, T_K] \times B(\bar{x}_{K:K+1}, \bar{y},\varepsilon/2) \} \\
			& \cap \{\varphi^{p, \eta}(\tau, \xi_{K:K+1}, \xi_y) < v^n_K(\tau, \xi_{K:K+1}, \xi_y)\}.
		\end{aligned}
	\end{equation}
	Then $A \in \cF_\tau$. Define $U^1 := (\alpha^1_{K:K+1}, L^1_K, M^1_K)$ similarly to \eqref{U1}, but with the event $A$ in \eqref{newA} instead. We have $U^1 \in \cU (\tau, \xi_{K:K+1}, \xi_y)$ and $$\p((X_{K:K+1}, Y)(t; \tau, \xi_{K:K+1}, \xi_y, U^1) \in \cS_{K}, \; \tau \leq t \leq T ) = 1.$$
	
	Let 
	\begin{align*}
		\tau^1 := \inf \big\{ t \in [\tau, T_K] \, \big| &  (t,  (X_{K:K+1}, Y)(t; \tau, \xi_{K:K+1}, \xi_y, U^1)) \\
		& \notin (T_K - \delta/2, T_K] \times B(\bar{x}_{K:K+1}, \bar{y},\varepsilon/2) \big\} {\color{black} \wedge T_K}.
	\end{align*}
	Denote 
	\begin{align*}
		\xi^1 := (\xi^1_{K:K+1}, \xi^1_y) :=  ((X_{K:K+1}, Y)(\tau^1; \tau, \xi_{K:K+1}, \xi_y, U^1))
	\end{align*}
	as the exit position, which is $\cF_{\tau^1}$-measurable.

	Note that 
	\begin{equation}\label{imp_ineq}
		- \lambda_{K} + \partial_{K+1} \varphi^{p, \eta} - \partial_{K} \varphi^{p, \eta} > 0 \quad  \text{ on } \overline{B(\bar{x}_{K:K+1}, \bar{y},\varepsilon)}.
	\end{equation}
	By the mean value theorem, for any $(x_{K:K+1}, y) \in B(\bar{x}_{K:K+1}, \bar{y},\varepsilon)$ and $(x_K + h, x_{K+1} - h, y) \in B(\bar{x}_{K:K+1}, \bar{y},\varepsilon)$ with a small enough $h>0$, we have
	\begin{align*}
		& \varphi^{p, \eta}(t, x_K + h, x_{K+1} - h, y) - \varphi^{p, \eta}(t, x_K, x_{K+1}, y) + \lambda_K h \\
		& \quad = [- \partial_{K+1} \varphi^{p, \eta} + \partial_K \varphi^{p, \eta}](t, \tilde{x}_K, \tilde{x}_{K+1}, y) h + \lambda_K h < 0,
	\end{align*}
	where $(\tilde{x}_K, \tilde{x}_{K+1})$ is on the line segment between $x_{K:K+1}$ and $(x_K + h, x_{K+1} - h)$. This suggests transferring from $x_{K+1}$ to $x_K$.

	Denote $(\cI_K(x_K), \cI_{K+1}(x_{K+1}), y)$ as the intersection of the ray $\{ (x_K + h, x_{K+1} - h, y): h \geq 0 \}$ and $\partial B(\bar{x}_{K:K+1}, \bar{y},\varepsilon/2)$. Then $x_{K+1} - \cI_{K+1}(x_{K+1}) = \cI_K(x_K) - x_K$ and either one can be used when defining $U$ below. Let $U^2 := \{ (\alpha^2_{K:K+1}(t), L^2_K(t), M^2_K(t)) \}_{t \in [T_K, T]}$ be a suitable control for $v^n_{K}$ with the random initial condition $(T_K, \cI_K(\xi^1_K), \cI_{K+1}(\xi^1_{K+1}), \xi^1_y)$. Let $$U^3 := \{ (\alpha^3_{K:K+1}(t), L^3_K(t), M^3_K(t)) \}_{t \in [\tau^1, T]}$$ be a suitable control for $v^n_{K}$ with the random initial condition $(\tau^1, \xi^1)$. Let  $U^4 := \{ \alpha^4_{K+1}(t) \}_{t \in [T_K, T]}$ be a suitable control for $v^n_{K+1}$ with the random initial condition $(T_K, \xi^1_{K+1}, \xi^1_y)$. Finally, define a control $U := (\alpha_{K:K+1}, L_K, M_K)$ by
	\begin{align*}
		& (\alpha_{K:K+1}(t), L_K(t), M_K(t)) \\
		& = \one_{\{ \tau \leq t < \tau^1 \}} U^1(t) \\
		& \quad + \one_{A \cap \{ \tau^1 = T_K\} \cap \{\tau^1 \leq t \leq T\}} (\alpha^2_{K:K+1}(t),  L^2_K(t) - L^2_K(\tau^1-) + \xi^1_{K+1} - \cI_{K+1}(\xi^1_{K+1}) + L^1_K(\tau^1), \\
		& \hspace{4.5cm}  M^2_K(t) - M^2_K(\tau^1 -) + M^1_K(\tau^1)) \\
		& \quad + \one_{\{ \tau^1 < T_K\} \cap \{ \tau^1 \leq t \leq T\}} (\alpha^3_{K:K+1}(t), L^3_K(t) - L^3_K(\tau^1 -) + L^1_K(\tau^1), M^3_K(t) - M^3_K(\tau^1 -) + M^1_K(\tau^1) ) \\
		& \quad + \one_{A^c \cap \{ \tau^1 = T_K\} \cap \{ \tau^1 \leq t \leq T\}} (\alpha^4_{K+1}(t), L^1_K(\tau^1), M^1_K(\tau^1) ).
	\end{align*}

	By Lemma \ref{lem:glue}, the control $U \in \cU(\tau, \xi_{K:K+1}, \xi_y)$. When $A \cap \{ \tau^1 = T_K\}$ happens, there can be ``triple transactions" at time $t = \tau^1 = T_K$, including $(\Delta L^1_K(T_K), \Delta M^1_K(T_K))$, $(\xi^1_{K+1} - \cI_{K+1}(\xi^1_{K+1}), 0)$, and $(\Delta L^2_K(T_K), \Delta M^2_K(T_K))$. In particular, $(\xi^1_{K+1} - \cI_{K+1}(\xi^1_{K+1}), 0)$ brings the state to the boundary $\partial B(\bar{x}_{K:K+1}, \bar{y},\varepsilon/2)$. But it is still possible to transfer immediately again with $(\Delta L^2_K(T_K), \Delta M^2_K(T_K))$. Besides, the investment proportion $\alpha^2_K(t)$ in $U^2$ for the portfolio $K$ is not really needed, since the goal $K$ expires right now.  
	
	We verify that $U$ is suitable for $v^{p, \eta}_{K}$ with $(\tau, \xi_{K:K+1}, \xi_y)$. Similar to Case (1), the inequality \eqref{psiAineq1} holds with $\varphi^{p, \eta}$ instead. The proof uses \eqref{phip} with $\varphi^{p, \eta}$. \eqref{psiAineq2} is replaced by
	\begin{align*}
		& \one_{A \cap \{ \rho \geq \tau^1\}} \varphi^{p, \eta} (\tau^1, \xi^1)  \\
		& \geq \one_{A \cap \{ \rho \geq \tau^1\} \cap \{ \tau^1 < T_K\}} v^n_K (\tau^1, \xi^1)  \\
		& \quad + \one_{A \cap \{ \rho \geq \tau^1\} \cap \{ \tau^1 = T_K\}} \Big( \lambda_K (\xi^1_{K+1} - \cI_{K+1}(\xi^1_{K+1})) + \varphi^{p, \eta}(T_K, \cI_K(\xi^1_K), \cI_{K+1}(\xi^1_{K+1}), \xi^1_y) \Big)  \\
		& \geq \one_{A \cap \{ \rho \geq \tau^1\} \cap \{ \tau^1 < T_K\}} v^n_K (\tau^1, \xi^1) \\
		& \quad + \one_{A \cap \{ \rho \geq \tau^1\} \cap \{ \tau^1 = T_K\}} \Big( \lambda_K (\xi^1_{K+1} - \cI_{K+1}(\xi^1_{K+1})) + v^n_K(T_K, \cI_K(\xi^1_K), \cI_{K+1}(\xi^1_{K+1}), \xi^1_y) \Big).
	\end{align*} 
	Here, we use \eqref{phi_E}, \eqref{imp_ineq}, and $(\cI_K(\xi^1_K), \cI_{K+1}(\xi^1_{K+1}), \xi^1_y) \in \partial B(\bar{x}_{K:K+1}, \bar{y},\varepsilon/2)$.
	
	Similar to \eqref{Acvnineq} and \eqref{ineq:tau1}, the remaining proof of the supermartingale property when $\tau \in [T_{K-1}, T_K]$ follows from the definition of $U$. Other cases with $\tau \in [T_{k-1}, T_k]$, $k \neq K$ are straightforward to prove. Hence, $(v^n_1, \ldots, v^{p, \eta}_{K}, v^n_{K+1})$ is a stochastic supersolution, which leads to a contradiction.

	{\bf Step 3: Case (3)}. $- \theta_K - \partial_{K+1} \varphi(\bar{x}_{K:K+1}, \bar{y}) + \partial_{K} \varphi (\bar{x}_{K:K+1}, \bar{y}) > 0$. 
	
	This suggests transferring from $x_{K}$ to $x_{K+1}$. The proof is similar to Case (2) with the counterpart inequality to \eqref{imp_ineq}. The control when $A \cap \{ \tau^1 = T_K\}$ happens should be modified as follows. Denote $(\cN_K(x_K), \cN_{K+1}(x_{K+1}), y)$ as the intersection of the ray $\{ (x_K - h, x_{K+1} + h, y): h \geq 0 \}$ and $\partial B(\bar{x}_{K:K+1}, \bar{y},\varepsilon/2)$. Let $U^2 := \{ (\alpha^2_{K:K+1}(t), L^2_K(t), M^2_K(t)) \}_{t \in [T_K, T]}$ be a suitable control for $v^n_{K}$ with the random initial condition $(T_K, \cN_K(\xi^1_K), \cN_{K+1}(\xi^1_{K+1}), \xi^1_y)$. The controls $U^3$ and $U^4$ are defined similarly as in Case (2). Finally, define a control $U := (\alpha_{K:K+1}, L_K, M_K)$ by
	\begin{align*}
		& (\alpha_{K:K+1}(t), L_K(t), M_K(t)) \\
		& = \one_{\{ \tau \leq t < \tau^1 \}} U^1(t) \\
		& \quad + \one_{A \cap \{ \tau^1 = T_K\} \cap \{\tau^1 \leq t \leq T\}} (\alpha^2_{K:K+1}(t),  L^2_K(t) - L^2_K(\tau^1-) + L^1_K(\tau^1), \\
		& \hspace{4.5cm}  M^2_K(t) - M^2_K(\tau^1 -) + \xi^1_{K} - \cN_{K}(\xi^1_{K}) + M^1_K(\tau^1)) \\
		& \quad + \one_{\{ \tau^1 < T_K\} \cap \{ \tau^1 \leq t \leq T\}} (\alpha^3_{K:K+1}(t), L^3_K(t) - L^3_K(\tau^1 -) + L^1_K(\tau^1), M^3_K(t) - M^3_K(\tau^1 -) + M^1_K(\tau^1) ) \\
		& \quad + \one_{A^c \cap \{ \tau^1 = T_K\} \cap \{ \tau^1 \leq t \leq T\}} (\alpha^4_{K+1}(t), L^1_K(\tau^1), M^1_K(\tau^1) ).
	\end{align*}

	Then the remaining proof is similar to Case (2) above.

	{\bf Step 4}. The viscosity subsolution property in Definition \ref{def:vis_sub} (1), when $t \in [T_{K-1}, T_K)$: The proof is similar to the previous Step 3 and \cite{bayraktar2013stochastic,bayraktar2015stochastic}. We mainly report the construction of the stochastic supersolutions and suitable controls here.
	
	Let $(\bar{t}, \bar{x}_{K:K+1}, \bar{y}) \in [T_{K-1}, T_K) \times \cS_K$ and $\varphi(t, x_{K:K+1}, y) \in C^{1, 2}([T_{K-1}, T_K) \times \cS_K)$ be a test function such that $v_{K, +} - \varphi$ attains a strict local maximum of zero at $(\bar{t}, \bar{x}_{K:K+1}, \bar{y})$. We want to prove
	\begin{equation}
		\begin{aligned}
			\max \Big\{ & \beta \varphi - \partial_t \varphi - H(\bar{x}_{K:K+1}, \bar{y}, \partial \varphi, \partial^2 \varphi), \\
			& - \lambda_{K} + \partial_{K+1} \varphi - \partial_{K} \varphi, - \theta_K - \partial_{K+1} \varphi + \partial_{K} \varphi \Big\} (\bar{t}, \bar{x}_{K:K+1}, \bar{y}) \leq 0.
		\end{aligned}
	\end{equation}
	Assume on the contrary that the left-hand side is strictly positive.
	
	{\bf Step 4: Case (1)}. $\beta \varphi(\bar{t}, \bar{x}_{K:K+1}, \bar{y}) - \partial_t \varphi(\bar{t}, \bar{x}_{K:K+1}, \bar{y}) - H(\bar{x}_{K:K+1}, \bar{y}, \partial \varphi, \partial^2 \varphi) > 0$.
	
	Similar to Step 3, there exist $\varepsilon, \eta > 0$ and $v^n_K$, the $K$-th element of some $v^n := (v^n_1, \ldots, v^n_K, $ $v^n_{K+1}) \in \cV^+$, such that $\varphi^\eta := \varphi - \eta$ satisfies
	\begin{align*}
		& \beta \varphi^\eta - \partial_t \varphi^\eta - L^{\bar{\alpha}_{K:K+1}} \varphi^\eta > 0 \quad \text{ on } \overline{B(\bar{t}, \bar{x}_{K:K+1}, \bar{y}, \varepsilon)}, \\
		& \varphi^\eta \geq v^n_K \quad \text{ on } \overline{B(\bar{t}, \bar{x}_{K:K+1}, \bar{y}, \varepsilon)} \backslash B(\bar{t}, \bar{x}_{K:K+1}, \bar{y}, \varepsilon/2), \\
		& \varphi^\eta(\bar{t}, \bar{x}_{K:K+1}, \bar{y}) < v_{K, +}(\bar{t}, \bar{x}_{K:K+1}, \bar{y}).
	\end{align*}
	Here, the open ball $B(\bar{t}, \bar{x}_{K:K+1}, \bar{y}, \varepsilon) := (\bar{t}-\varepsilon, \bar{t}+\varepsilon) \times B(\bar{x}_{K:K+1}, \bar{y},\varepsilon)$.
	
	Define
	\begin{equation}\label{veta}
		v^{\eta}_{K} (t, x_{K:K+1}, y) := \left\{ 
		\begin{array}{ c l }
			v^n_K (t, x_{K:K+1}, y) \wedge \varphi^{\eta}(t, x_{K:K+1}, y), &\; \text{on } \overline{B(\bar{t}, \bar{x}_{K:K+1}, \bar{y}, \varepsilon)}, \\
			v^n_K (t, x_{K:K+1}, y), & \; \text{otherwise}.
		\end{array}
		\right.
	\end{equation}
	Consider the random initial condition $(\tau, \xi_{K:K+1}, \xi_y)$ with $\tau \in [T_{K-1}, T_K]$. Define the event 
	\begin{equation}\label{newA2}
		A := \{ (\tau, \xi_{K:K+1}, \xi_y) \in B(\bar{t}, \bar{x}_{K:K+1}, \bar{y}, \varepsilon/2) \} \cap \{\varphi^{\eta}(\tau, \xi_{K:K+1}, \xi_y) < v^n_K(\tau, \xi_{K:K+1}, \xi_y)\}.
	\end{equation}
	Define $U^1 := (\alpha^1_{K:K+1}, L^1_K, M^1_K)$ similarly to \eqref{U1} with $A$ in \eqref{newA2} instead. Let 
	\begin{align*}
		\tau^1 := \inf \{ t \in [\tau, T_K{\color{black}]} \, | \,  (t,  (X_{K:K+1}, Y)(t; \tau, \xi_{K:K+1}, \xi_y, U^1)) \notin B(\bar{t}, \bar{x}_{K:K+1}, \bar{y}, \varepsilon/2) \}.
	\end{align*}
	{\color{black} By choosing $\varepsilon > 0$ such that $\bar{t} + \varepsilon/2 < T_K$, the forward flow of time guarantees the process exits the neighborhood before $T_K$. This ensures the infimum is taken over a non-empty set, making $\tau^1$ finite without needing an explicit cap.} Denote $\xi^1 := (\xi^1_{K:K+1}, \xi^1_y)$ as the exit position. {\color{black} Define a control $U$ as in \eqref{eq:step3-1-U}.  We can verify that $U$ is suitable for $v^{\eta}_{K}$ with $(\tau, \xi_{K:K+1}, \xi_y)$. The proof is similar to Case (1) in Step 3. }

	{\bf Step 4: Case (2)}. $ - \lambda_{K} + \partial_{K+1} \varphi (\bar{t}, \bar{x}_{K:K+1}, \bar{y}) - \partial_{K} \varphi (\bar{t}, \bar{x}_{K:K+1}, \bar{y}) > 0$.
	
	There exist $\varepsilon, \eta > 0$ and $v^n_K$, the $K$-th element of some $v^n := (v^n_1, \ldots, v^n_K, v^n_{K+1}) \in \cV^+$, such that $\varphi^\eta := \varphi - \eta$ satisfies
	\begin{align*}
		& [- \lambda_{K} + \partial_{K+1} \varphi^\eta - \partial_{K} \varphi^\eta](t, x_{K:K+1}, y) > 0 \quad \text{ on } \overline{B(\bar{t}, \bar{x}_{K:K+1}, \bar{y}, \varepsilon)}, \\
		& \varphi^\eta \geq v^n_K \quad \text{ on } \overline{B(\bar{t}, \bar{x}_{K:K+1}, \bar{y}, \varepsilon)} \backslash B(\bar{t}, \bar{x}_{K:K+1}, \bar{y}, \varepsilon/2), \\
		& \varphi^\eta(\bar{t}, \bar{x}_{K:K+1}, \bar{y}) < v_{K, +}(\bar{t}, \bar{x}_{K:K+1}, \bar{y}).
	\end{align*}
	
	Define $v^{\eta}_{K} (t, x_{K:K+1}, y)$ as in \eqref{veta} and the event $A$ as in \eqref{newA2}. Consider the random initial condition $(\tau, \xi_{K:K+1}, \xi_y)$ with $\tau \in [T_{K-1}, T_K]$. Let $U^0 := (\alpha^0_{K:K+1}, L^0_K, M^0_K)$ be a suitable control for $v^n_{K}$ with the random initial condition $(\tau, \xi_{K:K+1}, \xi_y)$. Recall the function $\cI_{K+1}(x_{K+1})$ introduced in Case (2) of Step 3. Define a new control $U^1 := (\alpha^1_{K:K+1}, L^1_K, M^1_K)$ by
	\begin{equation}
		\begin{aligned}
			\alpha^1_{K:K+1}(t)  := & \alpha^0_{K:K+1}(t) \one_{\{ \tau \leq t\}}, \\
			(L^1_K(t), M^1_K(t)) := & \one_{A \cap \{ \tau \leq t\}}(\xi_{K+1} - \cI_{K+1}(\xi_{K+1}), 0) \\
			& + \one_{A^c \cap \{ \tau \leq t\}} (L^0_K(t) - L^0_K(\tau-), M^0_K(t) - M^0_K(\tau-)).
		\end{aligned}
	\end{equation}
	{\color{black}   The exit time $\tau^1$ and exit position $\xi^1$ are defined as in Case (1) of Step 4. Under the event $A$, the control $U^1$ jumps at time $\tau$, pushing the state to the boundary. Hence, under $A$, we have $\tau^1 = \tau < T_K$ and $\xi^1 = (\cI_K(\xi_K), \cI_{K+1}(\xi_{K+1}), \xi_y)$. The remaining proof follows similarly.
	}
	
	{\bf Step 4: Case (3)}. $ - \theta_{K} - \partial_{K+1} \varphi (\bar{t}, \bar{x}_{K:K+1}, \bar{y}) + \partial_{K} \varphi (\bar{t}, \bar{x}_{K:K+1}, \bar{y}) > 0$.
	
	The proof is symmetric to Case (2) above. The main difference is that $U^1$ uses
	\begin{equation}
		\begin{aligned}
			(L^1_K(t), M^1_K(t)) := & \one_{A \cap \{ \tau \leq t\}}(0, \xi_{K} - \cN_{K}(\xi_{K})) \\
			& + \one_{A^c \cap \{ \tau \leq t\}} (L^0_K(t) - L^0_K(\tau-), M^0_K(t) - M^0_K(\tau-)),
		\end{aligned}
	\end{equation} 
	where $\cN_K$ is defined in Case (3) of Step 3.
	
	{\bf Step 5}. With some modifications, the proof above also works when more than two goals are active. We omit the detail since it is almost the same.
\end{proof}

\begin{proof}[Proof of Lemma \ref{lem:boundary}.]
	Consider a given point $(x^o_{k:K+1}, y) \in \overline{\cS}^o_{k} \backslash \cS_k$. Let $i$ be the smallest index with $x^o_i > 0$. For a small $\zeta >0$, recall the definition of $Z_{k:K+1}(x^o_{k:K+1}, \zeta)$ in \eqref{z}. Note that $(Z_{k:K+1}(x^o_{k:K+1}, \zeta), y) \in \cS_k$ for small $\zeta > 0$.
	
	Given $\zeta >0$, we regard $(t, Z_{k:K+1}(x^o_{k:K+1}, \zeta), y) \in [T_{k-1}, T_k] \times \cS_k$ as the random initial condition. For every stochastic supersolution $v_k$, there exists a suitable control $U$ with respect to this random initial condition. As the value function $V_k$ takes the infimum over all admissible controls, we have
	\begin{align}
		V_k(t, x^o_{k:K+1}, y) \leq &  \sum_{i=k}^{K} \lambda_i \zeta \Delta L_i (x^o_{k:K+1}) +  \sum_{i=k}^{K} \theta_i \zeta \Delta M_i(x^o_{k:K+1}) \label{ineq:46} \\
		& + v_k (t, Z_{k:K+1}(x^o_{k:K+1}, \zeta), y), \nonumber
	\end{align}
	where $\Delta L_i(\cdot)$ and $\Delta M_i(\cdot)$ are defined in \eqref{transfer}. Taking the infimum over $v_k$, we can replace $v_k$ with the upper stochastic envelope $v_{k, +}$ in this inequality \eqref{ineq:46}.
	
	Letting $\zeta \rightarrow 0$, since $(Z_{k:K+1}(x^o_{k:K+1}, \zeta), y) \in \cS_k$, we have
	\begin{align*}
		\limsup_{\zeta \rightarrow 0} v_{k, +}(t, Z_{k:K+1}(x^o_{k:K+1}, \zeta), y) \leq \limsup_{\substack{ x'_{k:K+1} \rightarrow x^o_{k:K+1} \\ x'_j > 0, \, j = k,\ldots,K+1}} v_{k, +}(t, x'_{k:K+1}, y).
	\end{align*}
	The first two terms in \eqref{ineq:46} converge to zero when $\zeta \rightarrow 0$. Then the claim \eqref{ineq:uscbd} follows.
\end{proof}

\subsection{Results of the stochastic subsolution}\label{sec:proof_stosub}

	\begin{proof}[Proof of Lemma \ref{lem:no_mental}]
		By definition of the value function $J$ without mental accounting, we have	
		\begin{equation}
		\begin{aligned}
			\Gamma_k(t, 0, y) & \leq \sum^{K+1}_{i = k} w_i e^{-\beta (T_i-t)} G_i, \quad \forall\;  t \in [0, T], \; y \in \R^m, \; k = 1, \ldots, K+1, \\
			\Gamma_{K+1}(T, x_{K+1}, y) & \leq (G_{K+1} - x_{K+1})^+, \quad \forall \; (x_{K+1}, y) \in \cS_{K+1}.
		\end{aligned}
		\end{equation} 
		
		Next, we prove that $(\Gamma_1, \ldots, \Gamma_k, \ldots, \Gamma_{K+1})$ satisfies the condition (1) in Definition \ref{def:sto_sub}:
		\begin{equation}\label{eq:sub_ex}
			\Gamma_k(\tau, \xi_{k:K+1}, \xi_y) \leq \E \big[ \cM \big([\tau, \rho], \Gamma_{k:K+1}, (X_{k:K+1}, Y)(\cdot; \tau, \xi_{k:K+1}, \xi_y, U) \big) \big|\cF_\tau \big],
		\end{equation}
        where $\tau \in [T_{k-1}, T_k]$. 
        
        Indeed, mental accounting terms in the objective function are non-negative:
        \begin{equation*}
        	\sum^K_{i=k} \lambda_i \int^{T_i}_{\tau} e^{-\beta t} d L_i(t) \geq 0 \; \text{ and } \; \sum^K_{i=k} \theta_i \int^{T_i}_{\tau} e^{-\beta t} d M_i(t) \geq 0
        \end{equation*}
        Then \eqref{eq:sub_ex} holds once we show that $\Gamma$ satisfies one direction of the DPP. {\color{black} When $k=1, \ldots, K$,}
		\begin{align}
			& \Gamma_k(\tau, \xi_{k:K+1}, \xi_y) \label{ineq:Gam_k}\\
			& \leq \inf_{\substack{ U \in \cU(\tau, \xi_{k:K+1}, \xi_y) } } \E \Big[ e^{-\beta(\rho - \tau)} \Gamma_{k}(\rho, (X_{k:K+1}, Y)(\rho; \tau, \xi_{k:K+1}, \xi_y, U)) \one_{\{\tau \leq \rho < T_k\}} \nonumber \\		
			& \qquad \qquad \qquad \qquad \quad + \sum^{K-1}_{l=k} \Big\{ e^{-\beta(\rho - \tau)} \Gamma_{l+1}(\rho, (X_{l+1:K+1}, Y)(\rho; \tau, \xi_{k:K+1}, \xi_y, U)) \nonumber \\
			&\qquad \qquad \qquad \qquad \qquad \qquad \quad +  \sum^{l}_{i = k} w_i e^{-\beta (T_i - \tau)} (G_i - X_i(T_i; \tau, \xi_{k:K+1}, \xi_y, U))^+ \Big \} \one_{\{T_l \leq \rho < T_{l+1} \}} \nonumber \\
			&  \qquad \qquad \qquad \qquad \quad + \Big\{ e^{-\beta(\rho - \tau)} \Gamma_{K+1}(\rho, (X_{K+1}, Y)(\rho; \tau, \xi_{k:K+1}, \xi_y, U)) \nonumber \\
			&\qquad \qquad \qquad \qquad \qquad \quad +  \sum^{K}_{i = k} w_i e^{-\beta (T_i - \tau)} (G_i - X_i(T_i; \tau, \xi_{k:K+1}, \xi_y, U))^+ \Big \} \one_{\{T_K \leq \rho \leq T\}} \Big | \cF_{\tau} \Big]. \nonumber 
		\end{align}
		{\color{black} For $k = K + 1$, the inequality simplifies to
			\begin{equation} \label{ineq:Gam_K+1}
				\begin{aligned}
					\Gamma_{K+1}(\tau, \xi_{K+1}, \xi_y) \leq \inf_{\substack{ U \in \cU(\tau, \xi_{K+1}, \xi_y) } } \E \Big[ & e^{-\beta(\rho - \tau)} \Gamma_{K+1}(\rho, (X_{K+1}, Y)(\rho; \tau, \xi_{K+1}, \xi_y, U)) \Big | \cF_{\tau} \Big].
				\end{aligned}
		\end{equation}
		}
	Since $\Gamma$ is defined by $J$, inequalities \eqref{ineq:Gam_k} and \eqref{ineq:Gam_K+1} follow from the fact that $J$ satisfies the DPP and is continuous when $t \in [T_{k-1}, T_k)$, as shown in \citet[Theorem 4.1 and Lemma E.2]{capponi2024}.
\end{proof}
Besides, one can also apply stochastic Perron's method to the problem \eqref{obj:no_acc} without mental accounting. To avoid the technical difficulty from the state constraint $X(t) \geq 0$ in \eqref{eq:no_account}, we can modify the wealth dynamic as in \citet[Equation (3.5)]{capponi2024} and allow negative initial wealth:
\begin{equation}\label{eq:extended}
	\begin{aligned}
		d X(u) &=  r X(u) du + (\mu(Y(u)) - r \mathbf{1})^\top \alpha(u) X^+(u) du + X^+(u) \alpha^\top(u) \sigma(Y(u)) dW(u), \\
		X(t) &= x \in \R.
	\end{aligned}
\end{equation}
It extends the spatial domain of the problem \eqref{obj:no_acc} to $\R \times \R^m$, while the value function in the original spatial domain $[0, \infty) \times \R^m$ is unchanged. Since there is no gradient constraint, the proof of viscosity supersolution and subsolution properties is standard and similar to \cite{bayraktar2013stochastic,bayraktar2015stochastic}. Moreover, Lemma \ref{lem:boundary} related to the state constraint  is no longer needed after the extension \eqref{eq:extended}. The proof of comparison principle is similar to \cite{capponi2024} and \citet[Theorem 4.4.5]{pham2009book}. We mention that the truncation $\rho(x)$ in \cite{capponi2024} is not needed, since the sequence $\{x_n\}$ is bounded. The unique viscosity solution in the enlarged domain $[T_{k-1}, T_k) \times \R \times \R^m$, agrees with the value function when restricting to the original domain $[T_{k-1}, T_k) \times [0, \infty) \times \R^m$.

As a side remark, if there are gradient constraints due to mental accounting and the state dynamic is extended to \eqref{eq:extended}, we could not find a strict classical subsolution as in Lemma \ref{lem:strict_l}. Therefore, this paper does not use this extension technique in the mental accounting case.

\begin{proof}[Proof of Proposition \ref{prop:vissuper}]
	{\bf Step 1}. The boundary/terminal conditions in Definition \ref{def:vis_super} (4): By Lemma \ref{lem:no_mental}, the {\it value} function without mental accounting, i.e., $\lambda_{k} = \theta_k = 0$, $k=1, \ldots, K$, is a stochastic subsolution. Since $v_-$ is the supremum of stochastic subsolutions, $v_-$ satisfies the boundary condition \eqref{vissup_bd0} at $x_{k:K+1} = 0$. By a modification of the step 3 in \citet[Theorem 4.1]{bayraktar2013stochastic}, it follows that $v_{K+1, -}$ in $v_-$ satisfies the terminal condition \eqref{vissup_T} at $T$.
	
	{\bf Step 2}. The viscosity supersolution property in Definition \ref{def:vis_super} (3), when $t \in [T_K, T)$: Similarly, it follows from the proof of \citet[Theorem 4.1]{bayraktar2013stochastic}.
	
	{\bf Step 3}. The viscosity supersolution property in Definition \ref{def:vis_super} (2), when $t=T_K$: Let $(\bar{x}_{K:K+1}, \bar{y}) \in \overline{\cS}^o_K = \{ [0, \infty)^2 \backslash \{0\} \} \times \R^m$ and $\varphi(x_{K:K+1}, y) \in C^2(\overline{\cS}^o_K)$ be a test function such that $v_{K, -}(T_K, \cdot, \cdot) - \varphi(\cdot, \cdot)$ attains a strict local minimum of zero at $(\bar{x}_{K:K+1}, \bar{y})$. We want to prove
	\begin{equation}
		\begin{aligned}
			\max \Big\{ & \varphi(\bar{x}_{K:K+1}, \bar{y}) - w_K (G_K - \bar{x}_K)^+ - v_{K+1, -}(T_K, \bar{x}_{K+1}, \bar{y}), \\
			& - \lambda_{K} + \partial_{K+1} \varphi(\bar{x}_{K:K+1}, \bar{y}) - \partial_{K} \varphi(\bar{x}_{K:K+1}, \bar{y}), \\
			&- \theta_K - \partial_{K+1} \varphi(\bar{x}_{K:K+1}, \bar{y}) + \partial_{K} \varphi (\bar{x}_{K:K+1}, \bar{y}) \Big\} \geq 0.
		\end{aligned}
	\end{equation}
	Assume on the contrary that the left-hand side is strictly negative. Since $\varphi$ is $C^2$ and $v_{K+1, -}$ is LSC,
	\begin{equation}
		\begin{aligned}
			\max \Big\{ & \varphi(x_{K:K+1}, y) - w_K (G_K - x_K)^+ - v_{K+1, -}(T_K, x_{K+1}, y), \\
			& - \lambda_{K} + \partial_{K+1} \varphi(x_{K:K+1}, y) - \partial_{K} \varphi(x_{K:K+1}, y), \\
			&- \theta_K - \partial_{K+1} \varphi(x_{K:K+1}, y) + \partial_{K} \varphi (x_{K:K+1}, y) \Big\} < 0
		\end{aligned}
	\end{equation}
	on a small closed set $\overline{B_{\bar{\cS}}(\bar{x}_{K:K+1}, \bar{y}, \varepsilon)}$, which is the closure of  
	\begin{equation*}
		B_{\bar{\cS}}(\bar{x}_{K:K+1}, \bar{y}, \varepsilon) := B(\bar{x}_{K:K+1}, \bar{y}, \varepsilon) \cap \overline{\cS}_K.
	\end{equation*}
	Indeed, thanks to the LSC property of $v_{K+1, -}$ and the continuity of $\varphi$, the following set is open:
	\begin{equation*}
		\big\{ (x_{K:K+1}, y) \in \overline{\cS}_K | 0 < - \varphi(x_{K:K+1}, y) + w_K (G_K - x_K)^+ + v_{K+1, -}(T_K, x_{K+1}, y) \big\}.
	\end{equation*}
	Then, with a small enough $\varepsilon > 0$, we have
	\begin{equation*}
		\varepsilon \leq - \varphi(x_{K:K+1}, y) + w_K (G_K - x_K)^+ + v_{K+1, -}(T_K, x_{K+1}, y)
	\end{equation*}
	on $\overline{B_{\bar{\cS}}(\bar{x}_{K:K+1}, \bar{y}, \varepsilon)}$. Also, we can make $\varepsilon$ small such that $\overline{B_{\bar{\cS}}(\bar{x}_{K:K+1}, \bar{y}, \varepsilon)} \subset \overline{\cS}^o_K$. It is used to verify the boundary condition of $v^{p, \eta}_{K}$ in \eqref{vpeta} when $x_{K:K+1} = 0$.
	
	Similar to \citet[Theorem 4.1]{bayraktar2013stochastic} and the step 3 in Proposition \ref{prop:vissub}, we introduce a function
	\begin{equation*}
		h(\delta) := \inf_{\substack{ T_K - \delta \leq t \leq T_K, \; (x_{K:K+1}, y) \in \overline{\cS}_K \text{ and } \\ \frac{\varepsilon}{2} \leq | (x_{K:K+1}, y) -  (\bar{x}_{K:K+1}, \bar{y})| \leq \varepsilon}} \Big( v_{K, -}(t, x_{K:K+1}, y) - \varphi(x_{K:K+1}, y) \Big), \quad 0 < \delta < \delta_0,
	\end{equation*}
	with some small $\delta_0 > 0$. Since $v_{K, -} - \varphi$ is LSC as a function of $(t, x_{K:K+1}, y)$, the minimum in $h(\delta)$ is attained. Let a minimizer be $(t^\delta, x^\delta_{K:K+1}, y^\delta)$. By compactness, there exists a subsequence, still denoted as $\delta \searrow 0$, such that 
	\begin{align*}
		& (t^\delta, x^\delta_{K:K+1}, y^\delta) \rightarrow (T_K, x^*_{K:K+1}, y^*) \\ 
		& \text{and} \quad \frac{\varepsilon}{2} \leq | (x^*_{K:K+1}, y^*) -  (\bar{x}_{K:K+1}, \bar{y})| \leq \varepsilon, \; (x^*_{K:K+1}, y^*) \in \overline{\cS}_K.
	\end{align*}
	Thus,
	\begin{align*}
		\liminf_{\delta \searrow 0} h(\delta) & = \liminf_{\delta \searrow 0} \Big( v_{K, -}(t^\delta, x^\delta_{K:K+1}, y^\delta) - \varphi(x^\delta_{K:K+1}, y^\delta) \Big) \\
		& \geq v_{K, -} (T_K, x^*_{K:K+1}, y^*) - \varphi(x^*_{K:K+1}, y^*) \\
		& \geq \inf_{\substack{ (x_{K:K+1}, y) \in \overline{\cS}_K, \\ \frac{\varepsilon}{2} \leq | (x_{K:K+1}, y) -  (\bar{x}_{K:K+1}, \bar{y})| \leq \varepsilon}}  \Big( v_{K, -}(T_K, x_{K:K+1}, y) - \varphi(x_{K:K+1}, y) \Big) > 0. 
	\end{align*}
	It implies that there exists $\delta$ small enough such that $h(\delta) > 0$. For this fixed $\delta$, we define $\delta_h := h(\delta) > 0$. When $T_K - \delta \leq t \leq T_K$, $\frac{\varepsilon}{2} \leq | (x_{K:K+1}, y) -  (\bar{x}_{K:K+1}, \bar{y})| \leq \varepsilon$ and $(x_{K:K+1}, y) \in \overline{\cS}_K$, we have
	\begin{equation} 
		v_{K, -}(t, x_{K:K+1}, y) - \varphi(x_{K:K+1}, y) \geq \delta_h > 0.
	\end{equation}
	Define $\varphi^p(t, x_{K:K+1}, y) := \varphi(x_{K:K+1}, y) - p(T_K - t)$. When $T_K - \delta \leq t \leq T_K - \delta/2$, $|(x_{K:K+1}, y) -  (\bar{x}_{K:K+1}, \bar{y})| \leq \varepsilon/2$, $(x_{K:K+1}, y) \in \overline{\cS}_K$ and $p$ is large enough, we have
	\begin{align*}
		& v_{K, -}(t, x_{K:K+1}, y) - \varphi^p(t, x_{K:K+1}, y) \\
		& \quad = v_{K, -}(t, x_{K:K+1}, y) - \varphi(x_{K:K+1}, y) + p(T_K - t)  \\
		& \quad \geq \inf_{\substack{T_K - \delta \leq t \leq T_K - \delta/2, \\ (x_{K:K+1}, y) \in \overline{\cS}_K, \\  | (x_{K:K+1}, y) -  (\bar{x}_{K:K+1}, \bar{y})| \leq \frac{\varepsilon}{2} }} \Big( v_{K, -}(t, x_{K:K+1}, y) - \varphi(x_{K:K+1}, y) \Big) + p(T_K - t) \geq \delta_h.
	\end{align*}
	Combining these results, we derive that $v_{K, -}(t, x_{K:K+1}, y) - \varphi^p(t, x_{K:K+1}, y) \geq \delta_h$ on 
	\begin{align*}
		E_{\bar{\cS}}(\delta, \varepsilon) :=&  \big\{ (t, x_{K:K+1}, y) \big|\, t \in [T_K - \delta, T_K], \, (x_{K:K+1}, y) \in \overline{B_{\bar{\cS}}(\bar{x}_{K:K+1}, \bar{y}, \varepsilon)} \big\} \\
		& - \big\{ (t, x_{K:K+1}, y) \big|\, t \in (T_K - \delta/2, T_K], \, (x_{K:K+1}, y) \in B_{\bar{\cS}}(\bar{x}_{K:K+1}, \bar{y}, \varepsilon/2) \big\}.
	\end{align*}
	Moreover, by letting $p$ even larger, we have
	\begin{equation*}
		\beta \varphi^p - \partial_t \varphi^p - H(x_{K:K+1}, y, \partial \varphi^p, \partial^2 \varphi^p) < 0 \quad \text{ on } [T_K - \delta, T_K] \times \overline{B_{\bar{\cS}}(\bar{x}_{K:K+1}, \bar{y}, \varepsilon)}.
	\end{equation*}
	With a Dini's argument, there exists $v^{n_1}_K$, which is the $K$-th component of a stochastic subsolution $v^{n_1} := (v^{n_1}_1, \ldots,  v^{n_1}_{K}, v^{n_1}_{K+1})$, such that
	\begin{equation*}
		v^{n_1}_K(t, x_{K:K+1}, y) - \varphi^p(t, x_{K:K+1}, y) \geq \delta_h/2 \quad \text{ on } E_{\bar{\cS}}(\delta, \varepsilon).
	\end{equation*}
	There exists $v^{n_2}_{K+1}$, which is the last component of a stochastic subsolution $v^{n_2} := (v^{n_2}_1, \ldots,$ $ v^{n_2}_{K}, v^{n_2}_{K+1})$, such that
	\begin{equation*}
		\varepsilon/2 \leq - \varphi(x_{K:K+1}, y) + w_K (G_K - x_K)^+ + v^{n_2}_{K+1}(T_K, x_{K+1}, y)
	\end{equation*}
	on $\overline{B_{\bar{\cS}}(\bar{x}_{K:K+1}, \bar{y}, \varepsilon)}$.
	
	We take $v^n := (v^{n_1}_1 \vee v^{n_2}_1, \ldots,  v^{n_1}_{K} \vee v^{n_2}_{K}, v^{n_1}_{K+1} \vee v^{n_2}_{K+1})$, which is a stochastic subsolution by Lemma \ref{lem:two_sub}. Let $0 < \eta < \min\{ \delta_h/2, \varepsilon/2 \}$ be small enough and define $\varphi^{p, \eta} := \varphi^p + \eta$. In summary, the results above show that $\varphi^{p, \eta}$ and $v^n$ satisfy the following properties:
	\begin{enumerate}[label={(\arabic*)}]	
		\item On $[T_K - \delta, T_K] \times \overline{B_{\bar{\cS}}(\bar{x}_{K:K+1}, \bar{y}, \varepsilon)}$, we have
		\begin{align}
			& \beta \varphi^{p, \eta} - \partial_t \varphi^{p, \eta} - H(x_{K:K+1}, y, \partial \varphi^{p, \eta}, \partial^2 \varphi^{p, \eta}) < 0, \label{subsol} \\
			& - \lambda_{K} + \partial_{K+1} \varphi^{p, \eta} - \partial_{K} \varphi^{p, \eta} < 0, \quad - \theta_K - \partial_{K+1} \varphi^{p, \eta} + \partial_{K} \varphi^{p, \eta} < 0; \nonumber
		\end{align}
		\item On $\overline{B_{\bar{\cS}}(\bar{x}_{K:K+1}, \bar{y}, \varepsilon)}$ with $t = T_K$,   
		\begin{equation*}
			\varphi^{p, \eta}(T_K, x_{K:K+1}, y) \leq w_K (G_K - x_K)^+ + v^{n}_{K+1}(T_K, x_{K+1}, y); 
		\end{equation*}
		\item $v^n_K(t, x_{K:K+1}, y) \geq \varphi^{p, \eta}(t, x_{K:K+1}, y)$ on $E_{\bar{\cS}}(\delta, \varepsilon)$;
		\item $\varphi^{p, \eta}(T_K, \bar{x}_{K:K+1}, \bar{y}) > v_{K, -}(T_K, \bar{x}_{K:K+1}, \bar{y})$.
	\end{enumerate}
	
	Define
	\begin{equation}\label{vpeta}
		v^{p, \eta}_{K} (t, x_{K:K+1}, y) = \left\{ 
		\begin{array}{ c l }
			v^n_K (t, x_{K:K+1}, y) \vee \varphi^{p, \eta}(t, x_{K:K+1}, y) & \text{on } [T_K - \delta, T_K] \times \overline{B_{\bar{\cS}}(\bar{x}_{K:K+1}, \bar{y}, \varepsilon)}, \\
			v^n_K (t, x_{K:K+1}, y), & \text{otherwise}.
		\end{array}
		\right.
	\end{equation}
	We verify the submartingale property of $(v^n_1, \ldots, v^{p, \eta}_{K}, v^n_{K+1})$ as a stochastic subsolution. Consider a random initial condition $(\tau, \xi_{K:K+1}, \xi_y)$ with $\tau \in [T_{K-1}, T_K]$, $(\xi_{K:K+1}, \xi_y) \in \cF_\tau$ and $\p((\xi_{K:K+1}, \xi_y) \in \overline{\cS}^o_K) = 1$. Let $U:= (\alpha_{K:K+1}, L_K, M_K) \in \cU(\tau, \xi_{K:K+1}, \xi_y)$. Consider a stopping time $\rho \in [\tau, T]$.
	
	Define the event 
	\begin{align*}
		A := & \{ (\tau, \xi_{K:K+1}, \xi_y) \in (T_K - \delta/2, T_K] \times B_{\bar{\cS}}(\bar{x}_{K:K+1}, \bar{y}, \varepsilon/2) \} \\
		& \cap \{\varphi^{p, \eta}(\tau, \xi_{K:K+1}, \xi_y) > v^n_K(\tau, \xi_{K:K+1}, \xi_y)\}.
	\end{align*}
	Then $A \in \cF_\tau$. Let
	\begin{align*}
		\tau^1 := \inf \Big\{ t \in [\tau, T_K] \, \Big| \, &  (t,  (X_{K:K+1}, Y)(t; \tau, \xi_{K:K+1}, \xi_y, U)) \\
		& \notin (T_K - \delta/2, T_K] \times B_{\bar{\cS}}(\bar{x}_{K:K+1}, \bar{y}, \varepsilon/2) \Big\} {\color{black} \wedge T_K}.
	\end{align*}
	Denote 
	\begin{align*}
		\xi^1 := (\xi^1_{K:K+1}, \xi^1_y) :=  ((X_{K:K+1}, Y)(\tau^1; \tau, \xi_{K:K+1}, \xi_y, U))
	\end{align*}
	as the exit position, which is $\cF_{\tau^1}$-measurable. 
	
	Under the event $A$ and at the exit time $\tau^1$, there may be a transfer such that the post-jump position $\xi^1$ is not on $\partial B_{\bar{\cS}}(\bar{x}_{K:K+1}, \bar{y}, \varepsilon/2)$. Since \eqref{subsol} only holds locally, we need the truncation method in \cite{bayraktar2015stochastic} to avoid this technical issue. However, our problem is more complicated as the exit may be only triggered by time expiration. If $\tau^1 = T_K$, the exit position $\xi^1$ could either be in $\overline{B_{\bar{\cS}}(\bar{x}_{K:K+1}, \bar{y}, \varepsilon/2)}$ or $\xi^1 \notin \overline{B_{\bar{\cS}}(\bar{x}_{K:K+1}, \bar{y}, \varepsilon/2)}$. Loosely speaking, we should use the intermediate position in \cite{bayraktar2015stochastic}, but also avoid overshooting.
	
	Let $(\xi^{1-}_{K:K+1}, \xi^{1-}_y) := (X_{K:K+1}, Y)(\tau^1-; \tau, \xi_{K:K+1}, \xi_y, U)$ be the left-limit position before the exit. We have $(\xi^{1-}_{K:K+1}, \xi^{1-}_y) \in \overline{B_{\bar{\cS}}(\bar{x}_{K:K+1}, \bar{y}, \varepsilon/2)}$ when $A$ happens. With a slight abuse of notations, denote $(\cI_K(x_K), \cI_{K+1}(x_{K+1}), y)$ as the intersection of the ray $\{ (x_K + h, x_{K+1} - h, y): h \geq 0 \}$ and $\partial B_{\bar{\cS}}(\bar{x}_{K:K+1}, \bar{y}, \varepsilon/2)$. Let $(\cN_K(x_K), \cN_{K+1}(x_{K+1}), y)$ be the intersection of the ray $\{ (x_K - h, x_{K+1} + h, y): h \geq 0 \}$ and $\partial B_{\bar{\cS}}(\bar{x}_{K:K+1}, \bar{y}, \varepsilon/2)$. We introduce the intermediate position as
	\begin{align*}
		(\tilde{\xi}^1_{K:K+1}, \tilde{\xi}^1_y) := & \one_{A \cap \{ \Delta L_K(\tau^1) > 0\}} \Big(\min\{\cI_K(\xi^{1-}_K), \xi^1_K \}, \max\{\cI_{K+1}(\xi^{1-}_{K+1}), \xi^1_{K+1}\}, \xi^{1-}_y \Big) \\
		& + \one_{A \cap \{ \Delta M_K(\tau^1) > 0\}} \Big(\max\{\cN_K(\xi^{1-}_K), \xi^1_K \}, \min \{\cN_{K+1}(\xi^{1-}_{K+1}), \xi^1_{K+1}\}, \xi^{1-}_y \Big) \\
		& + \one_{A^c \cup {\color{black} (A \cap \{ \Delta L_K(\tau^1) = 0, \, \Delta M_K(\tau^1) = 0\})} } (\xi^1_{K:K+1}, \xi^1_y),
	\end{align*}
	which is $\cF_{\tau^1}$-measurable.
	
	Define 
	\begin{align*}
		(\Delta L^1_K(\tau^1), \Delta M^1_K(\tau^1)) := & \one_{A \cap \{ \Delta L_K(\tau^1) > 0\}} \Big( \xi^{1-}_{K+1} - \max\{\cI_{K+1}(\xi^{1-}_{K+1}), \xi^1_{K+1}\}, 0 \Big) \\
		& + \one_{A \cap \{ \Delta M_K(\tau^1) > 0\}} \Big(0, \xi^{1-}_K - \max\{\cN_K(\xi^{1-}_{K}), \xi^1_K \} \Big) \\
		& + \one_{A^c \cup {\color{black} (A \cap \{ \Delta L_K(\tau^1) = 0, \, \Delta M_K(\tau^1) = 0\})} } (\Delta L_K (\tau^1), \Delta M_K (\tau^1)).
	\end{align*}
	To interpret the case under $A \cap \{ \Delta L_K(\tau^1) > 0 \}$, we note that
	\begin{align*}
		\xi^{1-}_{K+1} - \max\{\cI_{K+1}(\xi^{1-}_{K+1}), \xi^1_{K+1}\} = & \min \{ \xi^{1-}_{K+1} - \cI_{K+1}(\xi^{1-}_{K+1}), \xi^{1-}_{K+1} -  \xi^1_{K+1} \} \\
		= &  \min \{ \xi^{1-}_{K+1} - \cI_{K+1}(\xi^{1-}_{K+1}), \Delta L_K(\tau^1) \}.
	\end{align*}
	It means that we avoid overshooting by not transferring higher than $\Delta L_K(\tau^1)$. If the first term is binding, then $\tilde{\xi}^1_{K:K+1} = (\cI_K(\xi^{1-}_K), \cI_{K+1}(\xi^{1-}_{K+1}))$. Otherwise, $\tilde{\xi}^1_{K:K+1} = \xi^1_{K:K+1}$. Another case under $A \cap \{ \Delta M_K(\tau^1) > 0 \}$ can be interpreted similarly. Hence, $(\tilde{\xi}^1_{K:K+1}, \tilde{\xi}^1_y) \in \overline{B_{\bar{\cS}}(\bar{x}_{K:K+1}, \bar{y}, \varepsilon/2)}$.
	
	Define a control $U^1 := (\alpha^1_{K:K+1}, L^1_K, M^1_K)$ by
	\begin{equation}\label{def:U1_vissup}
		\begin{aligned}
			\alpha^1_{K:K+1}(t)  := & \alpha_{K:K+1}(t) \one_{\{ \tau \leq t\}}, \\
			(L^1_K(t), M^1_K(t)) := & \one_{\{ \tau \leq t < \tau^1\}}(L_K(t), M_K(t)) \\
			& + \one_{\{ t \geq \tau^1 \}} (L_K(\tau^1- ) + \Delta L^1_K(\tau^1), M_K(\tau^1- ) + \Delta M^1_K(\tau^1)).
		\end{aligned}
	\end{equation}
	When $t \geq \tau^1$, $(L^1_K(t), M^1_K(t))$ only transact at $\tau^1$ and are unchanged afterward.
	
	Under the event $A$, $(X_{K:K+1}, Y)(t; \tau, \xi_{K:K+1}, \xi_y, U^1) \in \overline{B_{\bar{\cS}}(\bar{x}_{K:K+1}, \bar{y}, \varepsilon/2)}$ for all $t \in [\tau, \tau^1]$. By It\^o's formula for RCLL semimartingales, we have
	\begin{align}
		& e^{-\beta (\rho \wedge \tau^1 - \tau)} \varphi^{p, \eta} (\rho \wedge \tau^1 , (X_{K:K+1}, Y)(\rho \wedge \tau^1; \tau, \xi_{K:K+1}, \xi_y, U^1)) \nonumber \\
		& \qquad - \varphi^{p, \eta} (\tau, (X_{K:K+1}, Y)(\tau; \tau, \xi_{K:K+1}, \xi_y, U^1)) \nonumber \\
		& =  \int^{\rho \wedge \tau^1}_{\tau} e^{-\beta(s - \tau)}\big( - \beta \varphi^{p, \eta} + \partial_t \varphi^{p, \eta} + L^{\alpha^1_{K:K+1}} \varphi^{p, \eta} \big) ds \label{Ito} \\
		& \quad + \int^{\rho \wedge \tau^1}_{\tau} e^{-\beta(s - \tau)} \big( - \partial_{K+1} \varphi^{p, \eta} + \partial_K \varphi^{p, \eta} \big) dL^{1, c}_K(s) \nonumber \\
		& \quad + \int^{\rho \wedge \tau^1}_{\tau} e^{-\beta(s - \tau)} \big( \partial_{K+1} \varphi^{p, \eta} - \partial_K \varphi^{p, \eta} \big) dM^{1, c}_K(s) \nonumber \\
		& \quad + \int^{\rho \wedge \tau^1}_{\tau} e^{-\beta(s - \tau)} X_K(s) \partial_K \varphi^{p, \eta} \alpha^1_K(s)^\top \sigma(Y(s)) dW(s) \nonumber \\
		& \quad + \int^{\rho \wedge \tau^1}_{\tau} e^{-\beta(s - \tau)} X_{K+1}(s) \partial_{K+1} \varphi^{p, \eta} \alpha^1_{K+1}(s)^\top \sigma(Y(s)) dW(s) \nonumber\\
		& \quad  + \int^{\rho \wedge \tau^1}_{\tau} e^{-\beta(s - \tau)} \partial_y \varphi^{p, \eta} \sigma_Y(Y(s)) dW(s) \nonumber\\
		& \quad + \sum_{\tau < s \leq \rho \wedge \tau^1} e^{-\beta(s - \tau)} \big[ \varphi^{p, \eta} (s , (X_{K:K+1}, Y)(s; \tau, \xi_{K:K+1}, \xi_y, U^1)) \nonumber\\
		& \hspace{4cm} - \varphi^{p, \eta} (s-, (X_{K:K+1}, Y)(s-; \tau, \xi_{K:K+1}, \xi_y, U^1)) \big]. \nonumber
	\end{align} 
	$L^{1, c}_K$ and $M^{1, c}_K$ denote the continuous part of $L^1_K$ and $M^1_K$, respectively. \eqref{subsol} can be used to simplify the $ds$, $dL^{1, c}_K$ and $dM^{1, c}_K$ integrals. Since $X_{K:K+1}(\tau; \tau, \xi_{K:K+1}, \xi_y, U^1)$ is the post-jump position at $\tau$, the last term above does not include $s = \tau$. By \eqref{subsol} and the mean value theorem, if $\Delta L^1_K(s) > 0$ when $A$ happens and $s \in [\tau, \tau^1]$, then
	\begin{align*}
		& \varphi^{p, \eta} (s , (X_{K:K+1}, Y)(s; \tau, \xi_{K:K+1}, \xi_y, U^1)) - \varphi^{p, \eta} (s-, (X_{K:K+1}, Y)(s-; \tau, \xi_{K:K+1}, \xi_y, U^1)) \\
		& + \lambda_{K} \Delta L^1_K(s) \geq 0.
	\end{align*}
	If $\Delta M^1_K(s) > 0$ when $A$ happens and $s \in [\tau, \tau^1]$, then
	\begin{align*}
		& \varphi^{p, \eta} (s , (X_{K:K+1}, Y)(s; \tau, \xi_{K:K+1}, \xi_y, U^1)) - \varphi^{p, \eta} (s-, (X_{K:K+1}, Y)(s-; \tau, \xi_{K:K+1}, \xi_y, U^1)) \\
		& + \theta_{K} \Delta M^1_K(s) \geq 0.
	\end{align*}
	Moreover, $\Delta L^1_K(s)$ and $\Delta M^1_K(s)$ are not positive at the same time. Taking the $\cF_\tau$-conditional expectation on \eqref{Ito}, we obtain
	\begin{align}
		& \E \Big[ \one_A \Big\{ e^{-\beta (\rho \wedge \tau^1 - \tau)} \varphi^{p, \eta} (\rho \wedge \tau^1 , (X_{K:K+1}, Y)(\rho \wedge \tau^1; \tau, \xi_{K:K+1}, \xi_y, U^1))  \nonumber \\
		& \qquad \quad + \lambda_K \int^{\rho \wedge \tau^1}_{\tau+} e^{-\beta(s - \tau)} dL^1_K(s) + \theta_K \int^{\rho \wedge \tau^1}_{\tau+} e^{-\beta(s - \tau)} dM^1_K(s) \Big\} \Big| \cF_{\tau} \Big] \nonumber \\
		& \geq \one_A \varphi^{p, \eta} (\tau , (X_{K:K+1}, Y)(\tau; \tau, \xi_{K:K+1}, \xi_y, U^1)).    \label{tau1_vissup}
	\end{align}
	The $dL^1_K$ and $dM^1_K$ integrals do not include jumps at $\tau$. Instead, these jumps at $\tau$ are related to the right-hand side of \eqref{tau1_vissup}:
	\begin{align}
		& \one_A \varphi^{p, \eta} (\tau , (X_{K:K+1}, Y)(\tau; \tau, \xi_{K:K+1}, \xi_y, U^1)) \nonumber \\
		& \geq \one_A \big\{ - \lambda_K \Delta L^1_K(\tau) - \theta_K \Delta M^1_K(\tau) + \varphi^{p, \eta} (\tau, \xi_{K:K+1}, \xi_y) \big\} \nonumber \\
		& = \one_A \big\{ - \lambda_K \Delta L^1_K(\tau) - \theta_K \Delta M^1_K(\tau) + v^{p, \eta}_K (\tau, \xi_{K:K+1}, \xi_y) \big\}, \label{RHS}
	\end{align}
	where the last equality follows from the definition of $A$.
	
	Consider the left-hand side of \eqref{tau1_vissup}. If $\rho < \tau^1$, then $(X_{K:K+1}, Y)(\rho; \tau, \xi_{K:K+1}, \xi_y, U^1) \in \overline{B_{\bar{\cS}}(\bar{x}_{K:K+1}, \bar{y}, \varepsilon/2)}$. Moreover, $U=U^1$ when $t < \tau^1$. Thus,
	\begin{align}
		& \one_{A \cap \{ \rho < \tau^1 \}} \varphi^{p, \eta} (\rho , (X_{K:K+1}, Y)(\rho; \tau, \xi_{K:K+1}, \xi_y, U^1)) \nonumber \\
		& \quad \leq \one_{A \cap \{ \rho < \tau^1 \}} v^{p, \eta}_K (\rho , (X_{K:K+1}, Y)(\rho; \tau, \xi_{K:K+1}, \xi_y, U)). \label{lessrho}
	\end{align}

	To investigate the case with $\rho \geq \tau^1$, we consider three events:
		\begin{align*}
			Q_1 & := \{ \tau^1 < T_K\}, \\
			Q_2 & := \{ \tau^1 = T_K\} \cap \{ (\xi^1_{K:K+1}, \xi^1_y) \in B_{\bar{\cS}}(\bar{x}_{K:K+1}, \bar{y}, \varepsilon/2) \}, \\
			Q_3 & := \{ \tau^1 = T_K \} \cap \{ (\xi^1_{K:K+1}, \xi^1_y) \notin B_{\bar{\cS}}(\bar{x}_{K:K+1}, \bar{y}, \varepsilon/2) \}.
		\end{align*}
		Both $Q_1$ and $Q_3$ mean that $(\xi^1_{K:K+1}, \xi^1_y) \notin B_{\bar{\cS}}(\bar{x}_{K:K+1}, \bar{y}, \varepsilon/2)$. The event $Q_2$ indicates that the exit is triggered by the time only. 	Next, we introduce $U^2 := (\alpha^2_{K:K+1}, L^2_K, M^2_K)$ with
		\begin{equation}\label{def:U2_vissup}
			(\alpha^2_{K:K+1}(t), L^2_K(t), M^2_K(t)) = \left\{ 
			\begin{array}{ c l }
				(0, 0, 0), & \; t < \tau^1, \\
				(0, L_K(\tau^1)- L^1_K(\tau^1), M_K(\tau^1)- M^1_K(\tau^1)), &\; t \geq \tau^1.
			\end{array}
			\right.
		\end{equation}
		Under $Q_1$ and $Q_3$, $U^2$ brings $(\tilde{\xi}^1_{K:K+1}, \tilde{\xi}^1_y)$ to $(\xi^1_{K:K+1}, \xi^1_y)$ immediately at $\tau^1$ and stays inactive afterward. For $Q_2$, we have $U^2 = (0, 0, 0)$. 
		
		Since $(\tau^1, \tilde{\xi}^1_{K:K+1}, \tilde{\xi}^1_y) \in E_{\bar{\cS}}(\delta, \varepsilon)$ under $Q_1$, we have
		\begin{align}
			& \one_{A \cap \{ \rho \geq \tau^1 \} \cap Q_1} \varphi^{p, \eta}(\tau^1, \tilde{\xi}^1_{K:K+1}, \tilde{\xi}^1_y) \nonumber \\
			& \leq \one_{A \cap \{ \rho \geq \tau^1 \} \cap Q_1} v^n_K(\tau^1, \tilde{\xi}^1_{K:K+1}, \tilde{\xi}^1_y) \nonumber \\
			& \leq \one_{A \cap \{ \rho \geq \tau^1 \} \cap Q_1} \E \Big[ \lambda_K \Delta L^2_K(\tau^1) + \theta_K \Delta M^2_K (\tau^1) + v^n_K (\tau^1, (X_{K:K+1}, Y)(\tau^1; \tau^1, \tilde{\xi}^1_{K:K+1}, \tilde{\xi}^1_y, U^2)) \Big| \cF_{\tau^1} \Big] \nonumber \\
			& = \one_{A \cap \{ \rho \geq \tau^1 \} \cap Q_1} \big[ \lambda_K \Delta L^2_K(\tau^1) + \theta_K \Delta M^2_K (\tau^1) + v^n_K (\tau^1, \xi^1_{K:K+1}, \xi^1_y) \big].   \label{Q1}
		\end{align}
		The first inequality uses the property (3) above. The second inequality follows from the submartingale property of $v^n_K$ with the random initial condition $(\tau^1, \tilde{\xi}^1_{K:K+1}, \tilde{\xi}^1_y)$ and the control $U^2$. 
		
		With the property (2) and $(\tilde{\xi}^1_{K:K+1}, \tilde{\xi}^1_y) = (\xi^1_{K:K+1}, \xi^1_y)$ on $Q_2$ given $A$, we obtain
		\begin{align}
			& \one_{A \cap \{ \rho \geq \tau^1 \} \cap Q_2} \varphi^{p, \eta}(\tau^1, \tilde{\xi}^1_{K:K+1}, \tilde{\xi}^1_y) \nonumber \\
			& \leq \one_{A \cap \{ \rho \geq \tau^1 \} \cap Q_2} \big[ w_K(G_K - \tilde{\xi}^1_K)^+ + v^n_{K+1}(T_K, \tilde{\xi}^1_{K+1}, \tilde{\xi}^1_y) \big] \nonumber \\
			& = \one_{A \cap \{ \rho \geq \tau^1 \} \cap Q_2} \big[ w_K(G_K - \xi^1_K)^+ + v^n_{K+1}(T_K, \xi^1_{K+1}, \xi^1_y) \big]. \label{Q2}
		\end{align}

		Similarly,
		\begin{align}
			& \one_{A \cap \{ \rho \geq \tau^1 \} \cap Q_3} \varphi^{p, \eta}(\tau^1, \tilde{\xi}^1_{K:K+1}, \tilde{\xi}^1_y) \nonumber \\
			& \leq \one_{A \cap \{ \rho \geq \tau^1 \} \cap Q_3} v^n_K(\tau^1, \tilde{\xi}^1_{K:K+1}, \tilde{\xi}^1_y) \nonumber \\
			& \leq \one_{A \cap \{ \rho \geq \tau^1 \} \cap Q_3} \big[ \lambda_K \Delta L^2_K(\tau^1) + \theta_K \Delta M^2_K (\tau^1) + w_K(G_K - \xi^1_K)^+ + v^n_{K+1}(T_K, \xi^1_{K+1}, \xi^1_y) \big]. \label{Q3}	
		\end{align}
		The first inequality uses the property (3) and the fact that $(\tau^1, \tilde{\xi}^1_{K:K+1}, \tilde{\xi}^1_y) \in E_{\bar{\cS}}(\delta, \varepsilon)$ under $Q_3$. The second inequality applies the submartingale property of $v^n_K$ from $\tau^1$ to $\tau^1$ under the control $U^2$. 		

		Putting \eqref{tau1_vissup}, \eqref{RHS}, \eqref{lessrho}, \eqref{Q1}, \eqref{Q2}, and \eqref{Q3} together, we obtain
		\begin{align}
			\E \Big[ &  \one_{A \cap \{ \rho < \tau^1 \}} e^{-\beta (\rho - \tau)} v^{p, \eta}_K (\rho , (X_{K:K+1}, Y)(\rho; \tau, \xi_{K:K+1}, \xi_y, U)) \nonumber \\
			& + \one_{A \cap \{ \rho \geq \tau^1 \} \cap \{ \tau^1 < T_K \}} e^{-\beta (\tau^1 - \tau)} v^n_K (\tau^1, \xi^1_{K:K+1}, \xi^1_y) \nonumber \\
			& + \one_{A \cap \{ \rho \geq \tau^1 \} \cap \{ \tau^1 = T_K \}} e^{-\beta (T_K - \tau)} \big( w_K(G_K - \xi^1_K)^+ + v^n_{K+1}(T_K, \xi^1_{K+1}, \xi^1_y) \big)  \nonumber \\
			& + \one_A \Big(\lambda_K \int^{\rho \wedge \tau^1}_{\tau} e^{-\beta(s - \tau)} dL_K(s) + \theta_K \int^{\rho \wedge \tau^1}_{\tau} e^{-\beta(s - \tau)} dM_K(s) \Big) \Big| \cF_{\tau} \Big] \nonumber \\
			\geq & \one_A v^{p, \eta}_K (\tau, \xi_{K:K+1}, \xi_y). \label{Aineq}
		\end{align}
		Here, we use the fact that the sum of $(L^1_K, M^1_K)$ and $(L^2_K, M^2_K)$ equals $(L_K, M_K)$ from $\tau$ to $\rho \wedge \tau^1$.
		
		Under $A^c$, the submartingale property of $v^n_K$ with the random initial condition $(\tau, \xi_{K:K+1}, \xi_y)$ and the control $U$ leads to
		\begin{align}
			\E \Big[ &  \one_{A^c \cap \{ \rho < \tau^1 \}} e^{-\beta (\rho - \tau)} v^{n}_K (\rho , (X_{K:K+1}, Y)(\rho; \tau, \xi_{K:K+1}, \xi_y, U)) \nonumber \\
			& + \one_{A^c \cap \{ \rho \geq \tau^1 \} \cap \{ \tau^1 < T_K \}} e^{-\beta (\tau^1 - \tau)} v^n_K (\tau^1, \xi^1_{K:K+1}, \xi^1_y) \nonumber \\
			& + \one_{A^c \cap \{ \rho \geq \tau^1 \} \cap \{ \tau^1 = T_K \}} e^{-\beta (T_K - \tau)} \big( w_K(G_K - \xi^1_K)^+ + v^n_{K+1}(T_K, \xi^1_{K+1}, \xi^1_y) \big)  \nonumber \\
			& + \one_{A^c} \Big(\lambda_K \int^{\rho \wedge \tau^1}_{\tau} e^{-\beta(s - \tau)} dL_K(s) + \theta_K \int^{\rho \wedge \tau^1}_{\tau} e^{-\beta(s - \tau)} dM_K(s) \Big) \Big| \cF_{\tau} \Big] \nonumber \\
			\geq & \one_{A^c} v^n_K (\tau, \xi_{K:K+1}, \xi_y) = \one_{A^c} v^{p, \eta}_K (\tau, \xi_{K:K+1}, \xi_y), \label{Acineq}
		\end{align}
		where the last equality follows from the definition of the event $A$ and the property (3). Noting that $v^{p, \eta}_K \geq v^n_K$ everywhere, \eqref{Aineq} and \eqref{Acineq} lead to
		\begin{align}
			\E \Big[ &  \one_{\{ \rho < \tau^1 \}} e^{-\beta (\rho - \tau)} v^{p, \eta}_K (\rho , (X_{K:K+1}, Y)(\rho; \tau, \xi_{K:K+1}, \xi_y, U)) \nonumber \\
			& + \one_{\{ \rho \geq \tau^1 \} \cap \{ \tau^1 < T_K \}} e^{-\beta (\tau^1 - \tau)} v^n_K (\tau^1, \xi^1_{K:K+1}, \xi^1_y) \nonumber \\
			& + \one_{\{ \rho \geq \tau^1 \} \cap \{ \tau^1 = T_K \}} e^{-\beta (T_K - \tau)} \big( w_K(G_K - \xi^1_K)^+ + v^n_{K+1}(T_K, \xi^1_{K+1}, \xi^1_y) \big)  \nonumber \\
			& + \lambda_K \int^{\rho \wedge \tau^1}_{\tau} e^{-\beta(s - \tau)} dL_K(s) + \theta_K \int^{\rho \wedge \tau^1}_{\tau} e^{-\beta(s - \tau)} dM_K(s) \Big| \cF_{\tau} \Big] \nonumber \\
			\geq & v^{p, \eta}_K (\tau, \xi_{K:K+1}, \xi_y). \label{AAc}
		\end{align}
		
		Similar to \citet[p.109]{bayraktar2015stochastic}, we introduce another control $U^3:=(\alpha^3_{K:K+1}, L^3_K,$ $ M^3_K)$ to avoid double-counting the transaction of $U$ at $\tau^1$:
		\begin{equation}
			\begin{aligned}
				\alpha^3_{K:K+1}(t)  & := \alpha_{K:K+1}(t), \\
				(L^3_K(t), M^3_K(t)) & := (L_K(t) - \one_{\{\tau^1 \leq t \}} \Delta L_K(\tau^1), M_K(t) - \one_{\{\tau^1 \leq t \}} \Delta M_K(\tau^1)). 
			\end{aligned}
		\end{equation}
		The control $U^3$ makes
		\begin{align}
			(X_{K:K+1}, Y)(t; \tau^1, \xi^1_{K:K+1}, \xi^1_y, U^3) = (X_{K:K+1}, Y)(t; \tau, \xi_{K:K+1}, \xi_y, U), \quad t \geq \tau^1. \label{U3_U}
		\end{align}
		
		When $\tau^1 < T_K$, we have
		\begin{align}
			& \one_{\{ \rho \geq \tau^1 \} \cap \{\tau^1 < T_K\}} v^n_K (\tau^1, \xi^1_{K:K+1}, \xi^1_y) \nonumber \\
			& \leq \E \Big[\one_{\{ \rho \geq \tau^1 \} \cap \{\tau^1 < T_K\}} \cM \big([\tau^1, \rho], (v^n_K, v^n_{K+1}), (X_{K:K+1}, Y)(\cdot; \tau^1, \xi^1_{K:K+1}, \xi^1_y, U^3) \big) \Big|\cF_{\tau^1} \Big] \nonumber \\
			&  \leq \E \Big[\one_{\{ \rho \geq \tau^1 \} \cap \{\tau^1 < T_K\}} \cM \big([\tau^1, \rho], (v^{p, \eta}_K, v^n_{K+1}), (X_{K:K+1}, Y)(\cdot; \tau^1, \xi^1_{K:K+1}, \xi^1_y, U^3) \big) \Big|\cF_{\tau^1} \Big] \nonumber \\
			& = \E \Big[\one_{\{ \rho \geq \tau^1 \} \cap \{\tau^1 < T_K\}} \cM \big([\tau^1, \rho], (v^{p, \eta}_K, v^n_{K+1}), (X_{K:K+1}, Y)(\cdot; \tau, \xi_{K:K+1}, \xi_y, U) \big) \Big|\cF_{\tau^1} \Big]. \label{Q1M}
		\end{align}
		The first inequality comes from the submartingale property of $v^n_K$ with the admissible control $U^3$. The second inequality uses the fact that $v^{p, \eta}_K \geq v^n_K$. The last equality is due to \eqref{U3_U}.
		
		When $\tau^1 = T_K$, we have
		\begin{align*}
			& \one_{\{ \rho \geq \tau^1 \} \cap \{\tau^1 = T_K\}} \big( w_K(G_K - \xi^1_K)^+ + v^n_{K+1}(T_K, \xi^1_{K+1}, \xi^1_y) \big) \\
			& \leq \E \Big[ \one_{\{ \rho \geq \tau^1 \} \cap \{\tau^1 = T_K\}} \big\{ w_K(G_K - \xi^1_K)^+ \\
			& \hspace{3.8cm} + \cM \big([\tau^1, \rho], v^n_{K+1}, (X_{K+1}, Y)(\cdot; \tau^1, \xi^1_{K+1}, \xi^1_y, U^3) \big) \big\}  \Big| \cF_{\tau^1} \Big] \\
			& = \E \Big[ \one_{\{ \rho \geq \tau^1 \} \cap \{\tau^1 = T_K\}} \big\{ w_K(G_K - \xi^1_K)^+ \\
			& \hspace{3.8cm} + \cM \big([\tau^1, \rho], v^n_{K+1}, (X_{K+1}, Y)(\cdot; \tau, \xi_{K:K+1}, \xi_y, U) \big) \big\}  \Big| \cF_{\tau^1} \Big].
		\end{align*}
	The last equality comes from \eqref{U3_U}. Note that $X_{K+1}$ may also depend on $X_K$ and thus $\xi_K$, due to potential transactions between goals.
		
	Summing them up with \eqref{AAc}, we get the desired result.

	{\bf Step 4}. The viscosity supersolution property in Definition \ref{def:vis_super} (1), when $t \in [T_{K-1}, T_K)$: The main idea is similar to Step 3 above. We give the detail about stochastic subsolutions and the decomposition of controls.
	
	Let $(\bar{t}, \bar{x}_{K:K+1}, \bar{y}) \in [T_{K-1}, T_K) \times \overline{\cS}^o_K$ and $\varphi(t, x_{K:K+1}, y) \in C^{1, 2}([T_{K-1}, T_K) \times \overline{\cS}^o_K)$ be a test function such that $v_{K, -} - \varphi$ attains a strict local minimum of zero at $(\bar{t}, \bar{x}_{K:K+1}, \bar{y})$. We want to prove
	\begin{equation}
		\begin{aligned}
			\max \Big\{ & \beta \varphi - \partial_t \varphi - H(\bar{x}_{K:K+1}, \bar{y}, \partial \varphi, \partial^2 \varphi), \\
			& - \lambda_{K} + \partial_{K+1} \varphi - \partial_{K} \varphi, - \theta_K - \partial_{K+1} \varphi + \partial_{K} \varphi \Big\} (\bar{t}, \bar{x}_{K:K+1}, \bar{y}) \geq 0.
		\end{aligned}
	\end{equation}
	Assume on the contrary that the left-hand side is strictly negative. By the continuity of SDE parameters and Berge's maximum theorem, the Hamiltonian $$H(x_{K:K+1}, y, \partial \varphi(t, x_{K:K+1}, y), \partial^2 \varphi(t, x_{K:K+1}, y))$$ is continuous in $(t, x_{K:K+1}, y)$. 
	
	Denote 
	\begin{equation*}
		B_{\bar{\cS}} (\bar{t}, \bar{x}_{K:K+1}, \bar{y}, \varepsilon) := (\bar{t}-\varepsilon, \bar{t}+\varepsilon) \times B_{\bar{\cS}}(\bar{x}_{K:K+1}, \bar{y}, \varepsilon).
	\end{equation*} 
	Similar to Step 3, there exist $\varepsilon, \eta > 0$ and $v^n_K$, the $K$-th element of some $v^n := (v^n_1, \ldots, v^n_K, v^n_{K+1}) \in \cV^-$, such that $\varphi^\eta := \varphi + \eta$ satisfies
	
	\begin{enumerate}[label={(\arabic*)}]	
		\item on $\overline{B_{\bar{\cS}} (\bar{t}, \bar{x}_{K:K+1}, \bar{y}, \varepsilon)}$,
		\begin{equation}
			\begin{aligned}
				\max \Big\{ & \beta \varphi^\eta - \partial_t \varphi^\eta - H(x_{K:K+1}, y, \partial \varphi^\eta, \partial^2 \varphi^\eta), \\
				& - \lambda_{K} + \partial_{K+1} \varphi^\eta - \partial_{K} \varphi^\eta, - \theta_K - \partial_{K+1} \varphi^\eta + \partial_{K} \varphi^\eta \Big\} < 0;
			\end{aligned}
		\end{equation}
		\item $\varphi^\eta \leq v^n_K \quad \text{ on } \overline{B_{\bar{\cS}} (\bar{t}, \bar{x}_{K:K+1}, \bar{y}, \varepsilon)} \backslash B_{\bar{\cS}} (\bar{t}, \bar{x}_{K:K+1}, \bar{y}, \varepsilon/2)$;
		\item $\varphi^\eta(\bar{t}, \bar{x}_{K:K+1}, \bar{y}) > v_{K, -}(\bar{t}, \bar{x}_{K:K+1}, \bar{y})$.
	\end{enumerate}
	Also, we can make $\varepsilon$ small such that $\overline{B_{\bar{\cS}} (\bar{t}, \bar{x}_{K:K+1}, \bar{y}, \varepsilon)} \subset [T_{K-1}, T_K) \times \overline{\cS}^o_K$.
	
	Define
	\begin{equation}
		v^{\eta}_{K} (t, x_{K:K+1}, y) := \left\{ 
		\begin{array}{ c l }
			v^n_K (t, x_{K:K+1}, y) \vee \varphi^{\eta}(t, x_{K:K+1}, y) &\; \text{on } \overline{B_{\bar{\cS}} (\bar{t}, \bar{x}_{K:K+1}, \bar{y}, \varepsilon)}, \\
			v^n_K (t, x_{K:K+1}, y), & \; \text{otherwise}.
		\end{array}
		\right.
	\end{equation}
	Consider the random initial condition $(\tau, \xi_{K:K+1}, \xi_y)$ with $\tau \in [T_{K-1}, T_K]$. Define the event 
	\begin{equation*}
		A := \{ (\tau, \xi_{K:K+1}, \xi_y) \in B_{\bar{\cS}} (\bar{t}, \bar{x}_{K:K+1}, \bar{y}, \varepsilon/2) \} \cap \{\varphi^{\eta}(\tau, \xi_{K:K+1}, \xi_y) > v^n_K(\tau, \xi_{K:K+1}, \xi_y)\}.
	\end{equation*}
	Let 
	\begin{align*}
		\tau^1 := \inf \{ t \in [\tau, T_K{\color{black}]} \, | \,  (t,  (X_{K:K+1}, Y)(t; \tau, \xi_{K:K+1}, \xi_y, U^1)) \notin B_{\bar{\cS}} (\bar{t}, \bar{x}_{K:K+1}, \bar{y}, \varepsilon/2) \}.
	\end{align*}
	 {\color{black} By choosing $\varepsilon > 0$ such that $\bar{t} + \varepsilon/2 < T_K$, the exit time $\tau^1$ is finite.} Denote $\xi^1 := (\xi^1_{K:K+1}, \xi^1_y)$ as the exit position. $U^1$ and $U^2$ are defined as in \eqref{def:U1_vissup} and \eqref{def:U2_vissup}, respectively. There is no need to introduce events $Q_i, i=1,2,3$. The remaining proof is similar to but simpler than Step 3.
	
	{\bf Step 5}. When more than two goals are active, the claim follows from a similar argument with some direct modifications. We omit the detail here.
\end{proof}

\subsection{Results of the comparison principle}\label{sec:proof_comp}

	\begin{proof}[Proof of Lemma \ref{lem:strict_l}]
		The derivatives of $l(\cdot)$ are
		\begin{align*}
			& \partial_t l = - \gamma l, \quad \partial_i l = - c_1 e^{\gamma(T_k - t)} \Big(1 + \sum^{K+1}_{i=k} a_i x_i \Big)^{q-1} a_i q, \quad \partial_y l = - 2 c_2 e^{\gamma (T_k - t)} y, \\
			& \partial^2_{ij} l = - c_1 e^{\gamma(T_k - t)} \Big(1 + \sum^{K+1}_{i=k} a_i x_i \Big)^{q-2} a_i a_j q(q-1), \quad \partial^2_{iy} l = 0, \quad \partial^2_{yy} l = - 2 c_2 e^{\gamma(T_k - t)} I_m. 
		\end{align*}
		
		First, we consider the infimum in the Hamiltonian. Since $\mu(y)$ and $\alpha_i$ are bounded, there exists a generic constant $C_\mu > 0$, such that
		\begin{align*}
			 &  \sum^{K+1}_{i=k}  (\mu(y) - r \one )^\top \alpha_i a_i x_i \leq  C_\mu \sum^{K+1}_{i=k} a_i x_i.
		\end{align*} 
		Then
		\begin{align*}
			&  \sum^{K+1}_{i=k}  (\mu(y) - r \one )^\top \alpha_i x_i \partial_i l \geq  - c_1 e^{\gamma(T_k - t)} \Big(1 + \sum^{K+1}_{i=k} a_i x_i \Big)^{q} C_\mu q.
		\end{align*} 
		 For the quadratic term, we introduce the following notations:
		 \begin{align*}
		 	A & := (\alpha_k, \ldots, \alpha_{K+1}) \in \R^{n \times (K-k+2)}, \\
		 	B & := (\mu(y) - r \one, \ldots, \mu(y) - r \one) \in \R^{n \times (K-k+2)}, \\
		 	X & := \text{diag}(x_{k:K+1}) \in \R^{(K-k+2) \times (K-k+2)}, \\
		 	\text{diag}(\partial_x l) & := \text{diag}((\partial_k l, \ldots, \partial_{K+1}l)) \in \R^{(K-k+2) \times (K-k+2)}.
		 \end{align*}
		 Since $\partial^2_{iy} l = 0$, we have 
		 \begin{align*}
		 	\tr[\Sigma(\alpha_{k:K+1}, x_{k:K+1}, y) \partial^2 l] = -2 c_2 e^{\gamma(T_k - t)} \tr[ \sigma_Y(y) \sigma_Y(y)^\top] + \tr[\Sigma_{xx} \partial^2_{xx} l],
		 \end{align*}
		 where $\partial^2_{xx} l$ is the Hessian matrix with the second derivative of $l$ on $x_{k:K+1}$ and 
		 \begin{align*}
		 	\Sigma_{xx} := (\alpha x)^\top_{k:K+1} \sigma(y) \sigma(y)^\top (\alpha x)_{k:K+1} = X A^\top \sigma(y) \sigma(y)^\top A X. 
		 \end{align*}
		 Since $\sigma_Y(y)$ is Lipschitz, there exists $C_{\sigma_Y} > 0$ such that
		 \begin{equation*}
		 	\tr[ \sigma_Y(y) \sigma_Y(y)^\top] \leq C_{\sigma_Y} (1 + |y|^2).
		 \end{equation*}
		 Entries of $X\partial^2_{xx}l X$ are given by
		 \begin{align*}
		 	 - c_1 e^{\gamma(T_k - t)} \Big(1 + \sum^{K+1}_{i=k} a_i x_i \Big)^{q-2} a_i a_j q(q-1) x_i x_j, \quad i, j \in \{k, \ldots, K+1\}.
		 \end{align*}
		  As $\alpha_i$ and $\sigma(y)$ are bounded, entries of $A^\top \sigma(y) \sigma(y)^\top A$ are bounded. There exists a constant $C_{\sigma} > 0$ such that
			\begin{align}
				\tr[\Sigma_{xx} \partial^2_{xx}l] \geq - C_\sigma q |q - 1| c_1 e^{\gamma(T_k - t)} \Big(1 + \sum^{K+1}_{i=k} a_i x_i \Big)^{q}.
			\end{align}
		 In summary, there exists a generic constant $C_{\sigma, \sigma_Y, \mu, q}$ depending on $q$, $\mu(y)$, $\sigma(y)$, and $\sigma_Y(y)$, such that
		\begin{align*}
			& \inf_{\alpha_{k:K+1} \in \cA^{K-k+2}}  \Big\{  \sum^{K+1}_{i=k}  (\mu(y) - r \one )^\top \alpha_i x_i \partial_i l + \frac{1}{2} \tr\left[\Sigma(\alpha_{k:K+1}, x_{k:K+1}, y) \partial^2 l  \right] \Big\} \\
			& \quad \geq - C_{\sigma, \sigma_Y, \mu, q} c_2  e^{\gamma(T_k - t)} (1 + |y|^2) - C_{\sigma, \sigma_Y, \mu, q} c_1 e^{\gamma(T_k - t)} \Big(1 + \sum^{K+1}_{i=k} a_i x_i \Big)^q,
		\end{align*} 
		Therefore, we obtain
		\begin{align*}
			& \beta l - \partial_t l - H(x_{k:K+1}, y, \partial l, \partial^2 l) \\
			& \leq (\beta + \gamma) l + r q c_1 e^{\gamma(T_k - t)} \Big(1 + \sum^{K+1}_{i=k} a_i x_i \Big)^q + C_{\mu_Y} c_2 e^{\gamma(T_k - t)} (1 + |y|^2) \\
			& \quad  + C_{\sigma, \sigma_Y, \mu, q} c_2  e^{\gamma(T_k - t)} (1 + |y|^2) + C_{\sigma, \sigma_Y, \mu, q} c_1 e^{\gamma(T_k - t)} \Big(1 + \sum^{K+1}_{i=k} a_i x_i \Big)^q \\
			& \leq \left[\beta + \gamma - rq - C_{\mu_Y} - C_{\sigma, \sigma_Y, \mu, q} \right]l,
		\end{align*}
		where $C_{\mu_Y}$ is a generic constant for the linear growth rate of $\mu_Y(y)$.
		
		We can choose $\gamma$ large enough, such that the right-hand side is strictly negative. Importantly, $\gamma$ only needs to be bigger than a constant independent of $c_1$ and $c_2$.
		
		For the gradient constraints in \eqref{F}, we have
		\begin{align*}
			\partial_{K+1} l - \partial_i l = (a_i - a_{K+1}) c_1 q e^{\gamma(T_k - t)} \Big(1 + \sum^{K+1}_{i=k} a_i x_i \Big)^{q-1}.
		\end{align*}
		Since $(1 + \sum^{K+1}_{i=k} a_i x_i )^{q-1} \leq 1$, we can set 
		$$ - \theta_i < (a_i - a_{K+1}) < \lambda_i, $$
		and choose $c_1 > 0$ small enough, such that $0 < c_1 q e^{\gamma(T_k - t)} < 1$. Then
		\begin{equation}\label{eq:grad}
			- \lambda_i + \partial_{K+1} l - \partial_i l < 0, \quad - \theta_i - \partial_{K+1} l + \partial_i l < 0.
		\end{equation}
	\end{proof}

\begin{proof}[Proof of Lemma \ref{lem:conti}]
	
	The claim \eqref{conti} is trivial for $(x_{k:K+1}, y) \in \cS_k$ since $g(x_{k:K+1})=0$. 
	
	Consider fixed $t \in [T_{k-1}, T_k]$ and $(x^o_{k:K+1}, y) \in \overline{\cS}^o_{k} \backslash \cS_k$. We define the corresponding transfer plan $(\zeta \Delta L_{k:K}(x^o_{k:K+1}), \zeta \Delta M_{k:K}(x^o_{k:K+1}))$ in \eqref{transfer} and the wealth state $Z_{k:K+1}(x^o_{k:K+1}, \zeta)$ in \eqref{z} after the transaction. First, it follows from the assumption that
	\begin{equation}
		u(t, x^o_{k:K+1}, y) = \limsup_{\substack{ x'_{k:K+1} \rightarrow x^o_{k:K+1} \\ x'_j > 0, \, j = k,\ldots,K+1}} u(t, x'_{k:K+1}, y).
	\end{equation}
	To provide an upper bound of the right-hand side, we note that the viscosity subsolution property \eqref{u_sub} of $u$ holds since $(t, x'_{k:K+1}, y) \in [T_{k-1}, T_k] \times \cS_k$. Moreover, elements of $x'_{k:K+1} - \zeta g(x^o_{k:K+1})$ are positive for any sufficiently small $\zeta > 0$. We obtain that
	\begin{align*}
		u(t, x'_{k:K+1}, y) \leq & u(t, x'_{k:K+1} - \zeta g(x^o_{k:K+1}), y) \\
		& + \sum^K_{j=k} \lambda_j \zeta \Delta L_j(x^o_{k:K+1}) + \sum^K_{j=k} \theta_j \zeta \Delta M_j(x^o_{k:K+1}),
	\end{align*}
	when $\zeta \in [0, \varepsilon)$ for some sufficiently small $\varepsilon>0$. Hence,
	\begin{align*}
		& u(t, x^o_{k:K+1}, y) \\
		&= \limsup_{\substack{ x'_{k:K+1} \rightarrow x^o_{k:K+1} \\ x'_j > 0, \, j = k,\ldots,K+1}} u(t, x'_{k:K+1}, y) \\
		& \leq \limsup_{\substack{ x'_{k:K+1} \rightarrow x^o_{k:K+1} \\ x'_j > 0, \, j = k,\ldots,K+1}} \Big\{ u(t, x'_{k:K+1} - \zeta g(x^o_{k:K+1}), y) + \sum^K_{j=k} \lambda_j \zeta \Delta L_j(x^o_{k:K+1}) + \sum^K_{j=k} \theta_j \zeta \Delta M_j(x^o_{k:K+1}) \Big\} \\
		& \leq u(t, x^o_{k:K+1} - \zeta g(x^o_{k:K+1}), y) + \sum^K_{j=k} \lambda_j \zeta \Delta L_j(x^o_{k:K+1}) + \sum^K_{j=k} \theta_j \zeta \Delta M_j(x^o_{k:K+1}),
	\end{align*}
	where the last inequality follows from the USC property of $u$. Next, we let $\zeta \downarrow 0$ and obtain
	\begin{align*}
		& u(t, x^o_{k:K+1}, y) \\
		& \leq \liminf_{\zeta \downarrow 0} \Big\{ u(t, x^o_{k:K+1} - \zeta g(x^o_{k:K+1}), y) + \sum^K_{j=k} \lambda_j \zeta \Delta L_j(x^o_{k:K+1}) + \sum^K_{j=k} \theta_j \zeta \Delta M_j(x^o_{k:K+1}) \Big\} \\
		& \leq \limsup_{\zeta \downarrow 0} \Big\{ u(t, x^o_{k:K+1} - \zeta g(x^o_{k:K+1}), y) + \sum^K_{j=k} \lambda_j \zeta \Delta L_j(x^o_{k:K+1}) + \sum^K_{j=k} \theta_j \zeta \Delta M_j(x^o_{k:K+1}) \Big\} \\
		& \leq \limsup_{\zeta \downarrow 0} u(t, x^o_{k:K+1} - \zeta g(x^o_{k:K+1}), y) \\
		& \quad + \limsup_{\zeta \downarrow 0}\Big\{ \sum^K_{j=k} \lambda_j \zeta \Delta L_j(x^o_{k:K+1}) + \sum^K_{j=k} \theta_j \zeta \Delta M_j(x^o_{k:K+1}) \Big\} \\
		& \leq u(t, x^o_{k:K+1}, y),
	\end{align*}
	where we apply the USC property of $u$ again in the last inequality. Then all inequalities must be equalities and the claim follows.
\end{proof}

\begin{proof}[Proof of Proposition \ref{p:comp_btw}]
    For the ease of presentation, we consider $k = K$ only, since other cases are similar. 
    
    {\bf Step 1}. Denote $s := (t, x_{K:K+1}, y)$. Assume on the contrary that $ u_K(\bar{s}) - v_K(\bar{s}) > 0$ for some $\bar{s} = (\bar{t}, \bar{x}_{K:K+1}, \bar{y}) \in [T_{K-1}, T_K] \times \overline{\cS}_{K}$. Let $l(s)$ be the strict subsolution in Lemma \ref{lem:strict_l}. Consider a constant $\delta \in (0, 1)$ which is close enough to 1 and small positive constants $\varepsilon, \zeta > 0$, such that
	\begin{equation}
		\delta u_K(\bar{s}) - v_K(\bar{s})  + 2 \varepsilon l(\bar{s}) - 2 \zeta^2 > 0. 
	\end{equation}
	Hence,
	\begin{align}\label{shat}
		\sup_{s \in [T_{K-1}, T_K] \times \overline{\cS}_{K}} \Big\{ \delta u_K(s) - v_K(s) + 2 \varepsilon l(s) - 2 \zeta^2 \Big\} > 0.
	\end{align}
	Since $\delta u_K(s) - v_K(s)$ is USC and bounded, $l \rightarrow -\infty$ when $|(x_{K:K+1}, y) | \rightarrow \infty$ in $\overline{\cS}_K$, the maximum in \eqref{shat} is attained in a compact subset of $[T_{K-1}, T_K] \times \overline{\cS}_{K}$ at some $\hat{s} = (\hat{t}, \hat{x}_{K:K+1}, \hat{y})$. It follows from the boundary/terminal conditions that $\hat{t} < T_K$ and $\hat{x}_{K:K+1} \neq (0, 0)$. Moreover, \eqref{shat} indicates that $\hat{s}$ does not depend on $\zeta$.
	
	Define a map $g(\hat{s}): [T_{K-1}, T_K] \times\overline{\cS}^o_{K} \rightarrow \R^{3+m}$ as
	\begin{equation}\label{eq:ghats}
		g(\hat{s}) = g(\hat{t}, \hat{x}_{K:K+1}, \hat{y}) = \left\{ 
		\begin{array}{ c l }
			(0, 0, 0, 0_m), & \quad \hat{x}_{K} > 0 \text{ and } \hat{x}_{K+1} > 0, \\
			(0, -1, 1, 0_m), & \quad \hat{x}_{K} = 0 \text{ and } \hat{x}_{K+1} > 0, \\
			(0, 1, -1, 0_m), & \quad \hat{x}_{K} > 0 \text{ and } \hat{x}_{K+1} = 0.
		\end{array}
		\right.
	\end{equation}
	Then $\hat{s} - \zeta g(\hat{s}) \in [T_{K-1}, T_K) \times \cS_K$ for small $\zeta > 0$.
	
	Next, for each $n > 0$, define
	\begin{align}
		\Phi_n(s, s') & := \delta u_K(s) - v_K(s') + \varepsilon l(s) + \varepsilon l(s') - \psi_n(s, s'), \label{Phi} \\
		\psi_n(s, s') & := |n(s - s') + \zeta g(\hat{s})|^2 + \zeta | s' - \hat{s} |^2.
	\end{align}
	Similarly, thanks to the fact that $\delta u_K(s) - v_K(s')$ is USC and bounded, $l \rightarrow -\infty$ when $|(x_{K:K+1}, y) | \rightarrow \infty$ in $\overline{\cS}_K$, the maximum of $\Phi_n$ in \eqref{Phi} is attained at 
	$$(s_n, s'_n) = ((t_n, x_{K:K+1, n}, y_n), (t'_n, x'_{K:K+1, n}, y'_n))$$
	in some compact subset of $[T_{K-1}, T_K] \times \overline{\cS}_{K} \times [T_{K-1}, T_K] \times \overline{\cS}_{K}$. Moreover, in view of \eqref{shat},
	\begin{align}
		\Phi_n(s_n, s'_n) \geq \Phi_n(\hat{s}, \hat{s}) > 0.
	\end{align}
	If we set $s = \hat{s} - \zeta g(\hat{s})/n$ and $s' = \hat{s}$, the inequality $\Phi_n(s_n, s'_n) \geq \Phi_n(\hat{s} - \zeta g(\hat{s})/n, \hat{s})$ yields
	\begin{align}
		& \delta u_K(s_n) - v_K(s'_n) + \varepsilon l(s_n) + \varepsilon l(s'_n) - |n(s_n - s'_n) + \zeta g(\hat{s})|^2 - \zeta | s'_n - \hat{s} |^2 \nonumber \\
		& \quad \geq \delta u_K(\hat{s} - \zeta g(\hat{s})/n) - v_K(\hat{s})  + \varepsilon l(\hat{s} - \zeta g(\hat{s})/n) + \varepsilon l(\hat{s}).
	\end{align}
	It reduces to
	\begin{align}
		& |n(s_n - s'_n) + \zeta g(\hat{s})|^2 + \zeta | s'_n - \hat{s} |^2 \nonumber \\
		& \quad \leq \delta u_K(s_n) - v_K(s'_n) + \varepsilon l(s_n) + \varepsilon l(s'_n) - \delta u_K(\hat{s} - \zeta g(\hat{s})/n) + v_K(\hat{s}) \label{ubd} \\
		& \qquad  - \varepsilon l(\hat{s} - \zeta g(\hat{s})/n) - \varepsilon l(\hat{s}). \nonumber 
	\end{align}
	Since $u_K$ and $v_K$ are bounded and $l(\hat{s} - \zeta g(\hat{s})/n) \rightarrow l(\hat{s})$ when $n \rightarrow \infty$ by the continuity of $l$, the right-hand side of \eqref{ubd} is bounded for all $n \geq 1$. Then $\{s_n\}^\infty_{n=1}$ and $\{s'_n\}^\infty_{n=1}$ are bounded. Up to a subsequence, both $\{s_n\}^\infty_{n=1}$ and $\{s'_n\}^\infty_{n=1}$ are convergent. They must converge to the same limit, denoted as $s^*$, since $\{|n(s_n - s'_n) + \zeta g(\hat{s})|^2 \}^\infty_{n=1}$ is bounded. To show that $s^* = \hat{s}$, we use the upper semicontinuity to derive
	\begin{equation*}
		\limsup_{n \rightarrow \infty} \Big\{ \delta u_K(s_n) - v_K(s'_n) \Big\} \leq \delta u_K(s^*) - v_K(s^*).
	\end{equation*}
	By the condition \eqref{modi_bd_u} and Lemma \ref{lem:conti},
	\begin{equation*}
		\lim_{n \rightarrow \infty} u_K(\hat{s} - \zeta g(\hat{s})/n) = u_K(\hat{s}).
	\end{equation*}
	Putting together, \eqref{ubd} becomes
	\begin{align}
		& \limsup_{n \rightarrow \infty} |n(s_n - s'_n) + \zeta g(\hat{s})|^2 + \zeta | s'_n - \hat{s} |^2 \nonumber \\
		& \quad \leq \delta u_K(s^*) - v_K(s^*) + 2 \varepsilon l(s^*) - [\delta u_K(\hat{s}) - v_K(\hat{s}) + 2 \varepsilon l(\hat{s})] \\
		& \quad \leq 0, \nonumber
	\end{align}
	where the last inequality follows from the fact that $\hat{s}$ is a maximum point in \eqref{shat}. Therefore,
	\begin{equation}\label{s_limit} 
		\lim_{n \rightarrow \infty} s_n = \lim_{n \rightarrow \infty} s'_n = \hat{s} = s^*, \quad \lim_{n \rightarrow \infty} |n(s_n - s'_n) + \zeta g(\hat{s})|^2 = 0.
	\end{equation}
	It leads to 
	\begin{equation}
		s_n = s'_n - \frac{\zeta g(\hat{s}) + o(1)}{n}.
	\end{equation}
	For sufficiently large $n$, $\hat{s} \in [T_{K-1}, T_K) \times \overline{\cS}^o_K$ implies that $s'_n \in [T_{K-1}, T_K) \times \overline{\cS}^o_K$ and $s_n \in [T_{K-1}, T_K) \times \cS_K$. It ensures that the viscosity sub/supersolution properties of $u_k$ and $v_k$ can be applied at $s_n$ and $s'_n$, respectively.

	{\bf Step 2}. Since $(s_n, s'_n)$ is a maximum point of $\Phi_n(s, s')$, by applying Ishii's lemma to the USC function $\delta u_K(s) + \varepsilon l(s)$ and the LSC function $ v_K(s') - \varepsilon l(s')$, we derive that there exist $(m+2)\times(m+2)$ symmetric matrices $M_n$ and $N_n$ satisfying
	\begin{align}
		& \left( \frac{\partial \psi_n(s_n, s'_n)}{\partial s}, M_n \right) \in \bar{J}^{2, +}_{[T_{K-1}, T_K) \times \cS_K} \big[ \delta u_K + \varepsilon l \big] (s_n), \\
		& \left( - \frac{\partial \psi_n(s_n, s'_n)}{\partial s'}, N_n \right) \in \bar{J}^{2, -}_{[T_{K-1}, T_K) \times \overline{\cS}^o_K} \big[ v_K - \varepsilon l \big](s'_n),
	\end{align} 
	and 
	\begin{equation}
		\begin{pmatrix}
			M_n & 0\\
			0 & - N_n
		\end{pmatrix} 
		\leq \Xi_n + \frac{1}{n^2} (\Xi_n)^2,
	\end{equation}
	where 
	\begin{align}
		& \frac{\partial \psi_n(s_n, s'_n)}{\partial s} = \left( \frac{\partial \psi_n(s_n, s'_n)}{\partial t}, \frac{\partial \psi_n(s_n, s'_n)}{\partial x_{K:K+1}}, \frac{\partial \psi_n(s_n, s'_n)}{\partial y} \right), \\
		& \frac{\partial \psi_n(s_n, s'_n)}{\partial s'} = \left( \frac{\partial \psi_n(s_n, s'_n)}{\partial t'}, \frac{\partial \psi_n(s_n, s'_n)}{\partial x'_{K:K+1}}, \frac{\partial \psi_n(s_n, s'_n)}{\partial y'} \right).
	\end{align}
	The Hessian matrix $\Xi_n$ is given by
	\begin{align*}
		\Xi_n = 2n^2 \begin{pmatrix}
			I_{2+m} & - I_{2+m}\\
			-I_{2+m} & I_{2+m}
		\end{pmatrix} 
		+ 2 \zeta \begin{pmatrix}
			0 & 0\\
			0 & I_{2+m}
		\end{pmatrix}.
	\end{align*}
	
	Since $l$ is $C^{1, 2}$ smooth, we can define
	\begin{align*}
		p_{t, n} & := \frac{1}{\delta} \left( \frac{\partial \psi_n(s_n, s'_n)}{\partial t} -  \varepsilon \frac{\partial l(s_n)}{\partial t} \right), \quad p_{K, n} := \frac{1}{\delta} \left( \frac{\partial \psi_n(s_n, s'_n)}{\partial x_K} -  \varepsilon \frac{\partial l(s_n)}{\partial x_K} \right),  \\
		p_{K+1, n} & := \frac{1}{\delta} \left( \frac{\partial \psi_n(s_n, s'_n)}{\partial x_{K+1}} -  \varepsilon \frac{\partial l(s_n)}{\partial x_{K+1}} \right), \quad p_{y, n} := \frac{1}{\delta} \left( \frac{\partial \psi_n(s_n, s'_n)}{\partial y} -  \varepsilon \frac{\partial l(s_n)}{\partial y} \right), \\
		A_n & := \frac{1}{\delta} \left( M_n -  \varepsilon \frac{\partial^2 l(s_n)}{\partial (x_{K:K+1}, y)^2} \right), \\
		q_{t, n} & := - \frac{\partial \psi_n(s_n, s'_n)}{\partial t'} + \varepsilon \frac{\partial l(s'_n)}{\partial t'}, \quad q_{K, n} := - \frac{\partial \psi_n(s_n, s'_n)}{\partial x'_K} +  \varepsilon \frac{\partial l(s'_n)}{\partial x'_K}, \\
		q_{K+1, n} & := - \frac{\partial \psi_n(s_n, s'_n)}{\partial x'_{K+1}} + \varepsilon \frac{\partial l(s'_n)}{\partial x'_{K+1}}, \quad q_{y', n} := - \frac{\partial \psi_n(s_n, s'_n)}{\partial y'} + \varepsilon \frac{\partial l(s'_n)}{\partial y'}, \\
		B_n & := N_n +  \varepsilon \frac{\partial^2 l(s'_n)}{\partial (x'_{K:K+1}, y')^2}.
	\end{align*}
	Then 
	\begin{align*}
		& \left( p_{t, n}, (p_{K:K+1, n}, p_{y, n}),  A_n\right) \in \bar{J}^{2, +}_{[T_{K-1}, T_K) \times \cS_K} \big[ u_K\big](s_n), \\
		& \left( q_{t, n}, (q_{K:K+1, n}, q_{y', n}),  B_n \right) \in \bar{J}^{2, -}_{[T_{K-1}, T_K) \times \overline{\cS}^o_K} \big[ v_K \big](s'_n).
	\end{align*} 
	The semijets definition of viscosity solution yields
	\begin{equation}
		\begin{aligned}
			\max \Big\{ & \beta u_K(s_n) - p_{t, n} - H(x_{K:K+1, n}, y_n, (p_{K:K+1, n}, p_{y, n}), A_n), \\
			& - \lambda_{K} + p_{K+1, n} - p_{K, n}, - \theta_{K} - p_{K+1, n} + p_{K, n}\Big\} \leq 0
		\end{aligned}
	\end{equation}
	and 
	\begin{equation}
		\begin{aligned}
			\max \Big\{ & \beta v_K(s'_n) - q_{t, n} - H(x'_{K:K+1, n}, y'_n, (q_{K:K+1, n}, q_{y', n}), B_n), \\
			& - \lambda_{K} + q_{K+1, n} - q_{K, n}, - \theta_{K} - q_{K+1, n} + q_{K, n}\Big\} \geq 0.
		\end{aligned}
	\end{equation}
	
        {\bf Step 3}. We consider three scenarios.
	
	{\bf Case 1}. Suppose that $- \lambda_{K} + q_{K+1, n} - q_{K, n} \geq 0$ holds for infinitely many $n$'s. Then
	\begin{align*}
		0 \geq & \delta (- \lambda_{K} + p_{K+1, n} - p_{K, n}) - (- \lambda_{K} + q_{K+1, n} - q_{K, n}) \\
		= & (1 - \delta) \lambda_K + \frac{\partial \psi_n(s_n, s'_n)}{\partial x_{K+1}} + \frac{\partial \psi_n(s_n, s'_n)}{\partial x'_{K+1}} - \frac{\partial \psi_n(s_n, s'_n)}{\partial x_{K}} - \frac{\partial \psi_n(s_n, s'_n)}{\partial x'_{K}} \\
		& -  \varepsilon \frac{\partial l(s_n)}{\partial x_{K+1}} + \varepsilon \frac{\partial l(s_n)}{\partial x_{K}}  -  \varepsilon \frac{\partial l(s'_n)}{\partial x'_{K+1}} + \varepsilon \frac{\partial l(s'_n)}{\partial x'_{K}}.
	\end{align*}
	We note that
	\begin{align*}
		\frac{\partial \psi_n(s_n, s'_n)}{\partial x_{K:K+1}} + \frac{\partial \psi_n(s_n, s'_n)}{\partial x'_{K:K+1}} = 2 \zeta (x'_{K:K+1, n} - \hat{x}_{K:K+1}) \rightarrow 0 \text{ as } n \rightarrow \infty.
	\end{align*}
	Moreover,
	\begin{align*}
		- \frac{\partial l(s_n)}{\partial x_{K+1}} + \frac{\partial l(s_n)}{\partial x_{K}} > - \lambda_{K}, \quad - \frac{\partial l(s'_n)}{\partial x'_{K+1}} + \frac{\partial l(s'_n)}{\partial x'_{K}} > - \lambda_{K}
	\end{align*}
	by the strict subsolution property of $l$. Since the first-order derivative of $l$ is continuous, the inequality becomes
	\begin{equation}
		0  \geq  (1 - \delta) \lambda_K  -  2 \varepsilon \frac{\partial l(\hat{s})}{\partial x_{K+1}} + 2 \varepsilon \frac{\partial l(\hat{s})}{\partial x_{K}} > (1 - \delta - 2 \varepsilon) \lambda_{K}, \text{ as } n \rightarrow \infty.
	\end{equation}
	We apply the strict subsolution property of $l$ at $\hat{s}$ in the second inequality. It contradicts the fact that $(1 - \delta - 2 \varepsilon) > 0$ and $\lambda_{K} \geq 0$.
	
	{\bf Case 2}. Suppose that $- \theta_{K} - q_{K+1, n} + q_{K, n} \geq 0$ holds for infinitely many $n$'s. Similarly, we can find a contradiction by $0 > (1 - \delta - 2 \varepsilon) \theta_{K}$.
	
	{\bf Case 3}. $\beta v_K(s'_n) - q_{t, n} - H(x'_{K:K+1, n}, y'_n, (q_{K:K+1,n}, q_{y', n}), B_n) \geq 0$ for infinitely many $n$'s. 
	
	For convenience, we define an operator 
	\begin{align*}
		& \cL^{\alpha_{k:K+1}} (x_{k:K+1}, y, \partial V_k, \partial^2 V_k) \\
		& := \sum^{K+1}_{i=k}  (\mu(y) - r \one )^\top \alpha_i x_i \partial_i V_k + \frac{1}{2} \tr\left[\Sigma(\alpha_{k:K+1}, x_{k:K+1}, y) \partial^2 V_k  \right].
	\end{align*} 
	By continuity, the infimum in $H(x'_{K:K+1, n}, y'_n, (q_{K:K+1, n}, q_{y', n}), B_n)$ is attained by some $\alpha^*_{K:K+1}$.  Then
	\begin{align}
		0 \leq & \beta v_K(s'_n) - q_{t, n} - H(x'_{K:K+1, n}, y'_n, (q_{K:K+1, n}, q_{y', n}), B_n) \nonumber \\
		& - \left[ \beta u_K(s_n) - p_{t, n} - H(x_{K:K+1, n}, y_n, (p_{K:K+1, n}, p_{y, n}), A_n) \right] \delta \nonumber \\
		\leq & \beta [v_K(s'_n) - \delta u_K(s_n)] + \frac{\partial \psi_n(s_n, s'_n)}{\partial t} -  \varepsilon \frac{\partial l(s_n)}{\partial t} + \frac{\partial \psi_n(s_n, s'_n)}{\partial t'} - \varepsilon \frac{\partial l(s'_n)}{\partial t'} \nonumber \\
		& - \sum^{K+1}_{i=K} r x'_{i, n} \left(- \frac{\partial \psi_n(s_n, s'_n)}{\partial x'_i} + \varepsilon \frac{\partial l(s'_n)}{\partial x'_i} \right) - \mu_Y(y'_n)^\top \left(- \frac{\partial \psi_n(s_n, s'_n)}{\partial y'} + \varepsilon \frac{\partial l(s'_n)}{\partial y'} \right)  \nonumber \\
		& - \cL^{\alpha^*_{K:K+1}} \left[ x'_{K:K+1, n}, y'_n, - \frac{\partial \psi_n(s_n, s'_n)}{\partial (x'_{K:K+1}, y')} + \varepsilon \frac{\partial l(s'_n)}{\partial (x'_{K:K+1}, y')}, N_n +  \varepsilon \frac{\partial^2 l(s'_n)}{\partial (x'_{K:K+1}, y')^2} \right]  \nonumber \\
		& + \sum^{K+1}_{i=K} r x_{i,n} \left( \frac{\partial \psi_n(s_n, s'_n)}{\partial x_i} - \varepsilon \frac{\partial l(s_n)}{\partial x_i} \right) + \mu_Y(y_n)^\top \left( \frac{\partial \psi_n(s_n, s'_n)}{\partial y} - \varepsilon \frac{\partial l(s_n)}{\partial y} \right)  \nonumber \\
		& + \cL^{\alpha^*_{K:K+1}} \left[ x_{K:K+1, n}, y_n, \frac{\partial \psi_n(s_n, s'_n)}{\partial (x_{K:K+1}, y)} - \varepsilon \frac{\partial l(s_n)}{\partial (x_{K: K+1}, y)}, M_n -  \varepsilon \frac{\partial^2 l(s_n)}{\partial (x_{K:K+1}, y)^2} \right], \label{Ham1}
	\end{align}
	where we used the fact that $\delta$ cancels with $1/\delta$ in $p_{K:K+1, n}$ and $A_n$, since $\cL^{\alpha_{K:K+1}}$ is linear in derivatives.
	
	
	A direct calculation simplifies the right-hand side of \eqref{Ham1} as
	\begin{align*}
		& \beta [v_K(s'_n) - \delta u_K(s_n)] + 2 \zeta (t'_n - \hat{t})  \\
		& + \sum^{K+1}_{i=K} r x'_{i,n} \frac{\partial \psi_n(s_n, s'_n)}{\partial x'_i} + \mu_Y(y'_n)^\top \frac{\partial \psi_n(s_n, s'_n)}{\partial y'} \\
		& + \cL^{\alpha^*_{K:K+1}} \left[ x'_{K:K+1, n}, y'_n, \frac{\partial \psi_n(s_n, s'_n)}{\partial (x'_{K:K+1}, y')}, -N_n \right]  \\
		& + \sum^{K+1}_{i=K} r x_{i,n}  \frac{\partial \psi_n(s_n, s'_n)}{\partial x_i} + \mu_Y(y_n)^\top \frac{\partial \psi_n(s_n, s'_n)}{\partial y} \\
		& + \cL^{\alpha^*_{K:K+1}} \left[ x_{K:K+1, n}, y_n, \frac{\partial \psi_n(s_n, s'_n)}{\partial (x_{K:K+1}, y)}, M_n \right] \\
		& - \varepsilon \frac{\partial l(s_n)}{\partial t} - \varepsilon \sum^{K+1}_{i=K} r x_{i, n} \frac{\partial l(s_n)}{\partial x_i}  - \varepsilon \mu_Y(y_n)^\top \frac{\partial l(s_n)}{\partial y} \\
		&  - \varepsilon \cL^{\alpha^*_{K:K+1}} \left[ x_{K:K+1, n}, y_n, \frac{\partial l(s_n)}{\partial (x_{K:K+1}, y)}, \frac{\partial^2 l(s_n)}{\partial (x_{K:K+1}, y)^2} \right] \\
		& - \varepsilon \frac{\partial l(s'_n)}{\partial t'} - \varepsilon \sum^{K+1}_{i=K} r x'_{i, n} \frac{\partial l(s'_n)}{\partial x'_i}  - \varepsilon \mu_Y(y'_n)^\top \frac{\partial l(s'_n)}{\partial y'} \\
		&  - \varepsilon \cL^{\alpha^*_{K:K+1}} \left[ x'_{K:K+1, n}, y'_n, \frac{\partial l(s'_n)}{\partial (x'_{K: K+1}, y')}, \frac{\partial^2 l(s'_n)}{\partial (x'_{K:K+1}, y')^2} \right].
	\end{align*}

	
	We discuss upper bounds for each term separately. Define matrices
	\begin{align}
		& C_n := \begin{pmatrix}
			(\alpha^* x_n)^\top_{K:K+1} \sigma(y_n)  \\
			\sigma_Y(y_n)
		\end{pmatrix} \in \R^{(2+m) \times d}, 
		\quad D_n := \begin{pmatrix}
			(\alpha^* x'_n)^\top_{K:K+1} \sigma(y'_n)  \\
			\sigma_Y(y'_n)
		\end{pmatrix} \in \R^{(2+m) \times d},
		\nonumber
	\end{align} 
	Then
	\begin{align}
		& \tr\left[\Sigma(\alpha^*_{K:K+1}, x_{K:K+1, n}, y_n) M_n \right] - \tr\left[\Sigma(\alpha^*_{K:K+1}, x'_{K:K+1, n}, y'_n) N_n \right] \nonumber \\
		& \leq 10 n^2 | C_n - D_n|^2 + \left(10 \zeta +  \frac{4\zeta^2}{n^2} \right)  \tr[D^\top_n D_n] - 8 \zeta \tr[C^\top_n D_n]. \label{Sigdiff} 
	\end{align}
	By first letting $n \rightarrow \infty$ and then $\zeta \rightarrow 0$, the right-hand side of \eqref{Sigdiff} converges to zero, thanks to the bounded Lipschitz $\sigma(\cdot)$, Lipschitz $\sigma_Y(\cdot)$, \eqref{s_limit}, and the boundedness of $\{s_n\}$ and $\{s'_n\}$ obtained from \eqref{s_limit}.
	
	We also have
	\begin{align*}
		& x'_{i,n} \frac{\partial \psi_n(s_n, s'_n)}{\partial x'_i} + x_{i,n}  \frac{\partial \psi_n(s_n, s'_n)}{\partial x_i} \\
		& \leq 2n|x_{i, n} - x'_{i, n}|\times|n(x_{i, n} - x'_{i, n}) + \zeta g(\hat{x}_{i})| + 2 \zeta x'_{i, n} |x'_{i, n} - \hat{x}_i| \\
		& \rightarrow 0, \text{ as } n \rightarrow \infty, \\
		& \mu_Y(y'_n)^\top \frac{\partial \psi_n(s_n, s'_n)}{\partial y'} + \mu_Y(y_n)^\top \frac{\partial \psi_n(s_n, s'_n)}{\partial y} \\
		& \leq c \left(  n^2|y_n - y'_n|^2 + \zeta | y'_n - \hat{y} | \right) \\
		& \rightarrow 0, \text{ as } n \rightarrow \infty,
	\end{align*}
	where $c$ is a generic constant, and
	\begin{align*}
		& (\mu(y_n) - r \one )^\top \alpha^*_i x_{i, n} \frac{\partial \psi_n(s_n, s'_n)}{\partial x_i} + (\mu(y'_n) - r \one )^\top \alpha^*_i x'_{i, n} \frac{\partial \psi_n(s_n, s'_n)}{\partial x'_i}   \\
		& \leq c (n |x_{i, n} - x'_{i, n}| + n |y_n - y'_n|) \times|n(x_{i, n} - x'_{i, n}) + \zeta g(\hat{x}_{i})| + c \zeta |x'_{i, n} - \hat{x}_i| \\
		& \rightarrow 0, \text{ as } n \rightarrow \infty.
	\end{align*}
        Here, $g(\hat{x}_i)$ denotes the entry of $g(\hat{s})$ in \eqref{eq:ghats} corresponding to $\hat{x}_i$. Hence, $g(\hat{x}_i)$ takes values in $\{1, -1, 0\}$.
    
	First, when $\beta > 0$, thanks to the estimates above and the fact that $l(s)$ is a strict subsolution, \eqref{Ham1} reduces to
	\begin{align*}
		0 < & \beta [v_K(s'_n) - \delta u_K(s_n)] - \beta \varepsilon l(s_n) - \beta \varepsilon l(s'_n) + 2 \zeta |t'_n - \hat{t}|\\
		& + \sum^{K+1}_{i=K} \left( 2n|x_{i, n} - x'_{i, n}|\times|n(x_{i, n} - x'_{i, n}) + \zeta g(\hat{x}_{i})| + 2 \zeta x'_{i, n} |x'_{i, n} - \hat{x}_i| \right) \\
		& + c \left(  n^2|y_n - y'_n|^2 + \zeta | y'_n - \hat{y} | \right) \\
		& + \sum^{K+1}_{i=K} \Big\{ c (n |x_{i, n} - x'_{i, n}| + n |y_n - y'_n|) \times|n(x_{i, n} - x'_{i, n}) + \zeta g(\hat{x}_{i})| + c \zeta |x'_{i, n} - \hat{x}_i| \Big\} \\
		& + 10 n^2 | C_n - D_n|^2 + \left(10 \zeta +  \frac{4\zeta^2}{n^2} \right)  \tr[D^\top_n D_n] - 8 \zeta \tr[C^\top_n D_n] \\
		\leq & \beta [v_K(\hat{s}) - \delta u_K(\hat{s}) - 2 \varepsilon l(\hat{s}) + 2 \zeta^2 - \psi_n(s_n, s'_n)] + 2 \zeta |t'_n - \hat{t}|\\
		& + \sum^{K+1}_{i=K} \left( 2n|x_{i, n} - x'_{i, n}|\times|n(x_{i, n} - x'_{i, n}) + \zeta g(\hat{x}_{i})| + 2 \zeta x'_{i, n} |x'_{i, n} - \hat{x}_i| \right) \\
		& + c \left(  n^2|y_n - y'_n|^2 + \zeta | y'_n - \hat{y} | \right) \\
		& + \sum^{K+1}_{i=K} \Big\{ c (n |x_{i, n} - x'_{i, n}| + n |y_n - y'_n|) \times|n(x_{i, n} - x'_{i, n}) + \zeta g(\hat{x}_{i})| + c \zeta |x'_{i, n} - \hat{x}_i| \Big\} \\
		& + 10 n^2 | C_n - D_n|^2 + \left(10 \zeta +  \frac{4\zeta^2}{n^2} \right)  \tr[D^\top_n D_n] - 8 \zeta \tr[C^\top_n D_n].
	\end{align*}
	The second inequality is from
	\begin{align*}
		\Phi_n (s_n, s'_n) \geq \delta u_K(\hat{s}) - v_K(\hat{s}) + 2 \varepsilon l(\hat{s}) - 2 \zeta^2,
	\end{align*}
	since $(s_n, s'_n)$ is a maximum point of $\Phi_n$. By first letting $n \rightarrow \infty$ and then $\zeta \rightarrow 0$, it leads to
	\begin{align}\label{eq:contrad}
		0 \leq \beta [v_K(\hat{s}) - \delta u_K(\hat{s}) - 2 \varepsilon l(\hat{s})]. 
	\end{align}
	Since $\beta > 0$, \eqref{eq:contrad} contradicts \eqref{shat}.
	
	When $\beta = 0$, we have
	\begin{align*}
		0 \leq 	& - \varepsilon \frac{\partial l(s_n)}{\partial t} - \varepsilon \sum^{K+1}_{i=K} r x_{i, n} \frac{\partial l(s_n)}{\partial x_i}  - \varepsilon \mu_Y(y_n)^\top \frac{\partial l(s_n)}{\partial y} \\
		&  - \varepsilon \cL^{\alpha^*_{K:K+1}} \left[ x_{K:K+1, n}, y_n, \frac{\partial l(s_n)}{\partial (x_{K:K+1}, y)}, \frac{\partial^2 l(s_n)}{\partial (x_{K:K+1}, y)^2} \right] \\
		& - \varepsilon \frac{\partial l(s'_n)}{\partial t'} - \varepsilon \sum^{K+1}_{i=K} r x'_{i, n} \frac{\partial l(s'_n)}{\partial x'_i}  - \varepsilon \mu_Y(y'_n)^\top \frac{\partial l(s'_n)}{\partial y'} \\
		&  - \varepsilon \cL^{\alpha^*_{K:K+1}} \left[ x'_{K:K+1, n}, y'_n, \frac{\partial l(s'_n)}{\partial (x'_{K: K+1}, y')}, \frac{\partial^2 l(s'_n)}{\partial (x'_{K:K+1}, y')^2} \right] \\
		& + 2 \zeta |t'_n - \hat{t}| + \sum^{K+1}_{i=K} \left( 2n|x_{i, n} - x'_{i, n}|\times|n(x_{i, n} - x'_{i, n}) + \zeta g(\hat{x}_{i})| + 2 \zeta x'_{i, n} |x'_{i, n} - \hat{x}_i| \right) \\
		& + c \left(  n^2|y_n - y'_n|^2 + \zeta | y'_n - \hat{y} | \right) \\
		& + \sum^{K+1}_{i=K} \Big\{ c (n |x_{i, n} - x'_{i, n}| + n |y_n - y'_n|) \times|n(x_{i, n} - x'_{i, n}) + \zeta g(\hat{x}_{i})| + c \zeta | x'_{i, n} - \hat{x}_i| \Big\}\\
		& + 10 n^2 | C_n - D_n|^2 + \left(10 \zeta +  \frac{4\zeta^2}{n^2} \right)  \tr[D^\top_n D_n] - 8 \zeta \tr[C^\top_n D_n].
	\end{align*}
	
	First letting $n \rightarrow \infty$ and then $\zeta \rightarrow 0$, by the continuity of derivatives of $l(s)$, we have
	\begin{align*}
		0 \leq 	& - 2 \varepsilon \frac{\partial l(\hat{s})}{\partial t} - 2 \varepsilon \sum^{K+1}_{i=K} r \hat{x}_{i} \frac{\partial l(\hat{s})}{\partial x_i}  - 2 \varepsilon \mu_Y(\hat{y})^\top \frac{\partial l(\hat{s})}{\partial y} \\
		&  - 2 \varepsilon \cL^{\alpha^*_{K:K+1}} \left[ \hat{x}_{K:K+1}, \hat{y}, \frac{\partial l(\hat{s})}{\partial (x_{K:K+1}, y)}, \frac{\partial^2 l(\hat{s})}{\partial (x_{K:K+1}, y)^2} \right].
	\end{align*}
	However, the strict subsolution property of $l$ at $\hat{s}$ indicates that the right-hand side is strictly smaller than zero, which is a contradiction.
\end{proof}

	\begin{proof}[Proof of Proposition \ref{prop:Tk}]
		Consider $k = K$ for ease of presentation. 
		
		With a slight abuse of notations, denote $s := (x_{K:K+1}, y)$. Assume on the contrary that $ u_K(\bar{s}) - v_K(\bar{s}) > 0$ for some $\bar{s} = (\bar{x}_{K:K+1}, \bar{y}) \in \overline{\cS}_{K}$. A strict subsolution is
		\begin{equation}
			l(x_{K:K+1}, y) := - c_1 \Big(1 + \sum^{K+1}_{i=K} a_i x_i \Big)^q - c_2 (1+|y|^2),
		\end{equation}
		where 
		\begin{equation}
			\begin{aligned}
				c_1, c_2 >0, \; a_i >0, \; q \in (0, 1), \; - \theta_K < (a_K - a_{K+1}) < \lambda_K, \quad 0 < c_1 q < 1,
			\end{aligned} 
		\end{equation}
		and $c_2$ is large enough to guarantee $l(x_{K:K+1}, y) - w_K (G_K - x_K)^+ - f(x_{K+1}, y) < 0$, thanks to the boundedness of $f$.
		
		Since $f$ is bounded, there exists a constant $C_f > 0$ such that $|f(x_{K+1}, y)| \leq C_f$. Consider a constant $\delta \in (0, 1)$ which is close enough to 1 and small positive constants $\varepsilon, \zeta > 0$, such that
		\begin{equation}
			\delta u_K(\bar{s}) - v_K(\bar{s})  + 2 \varepsilon l(\bar{s}) - 2 \zeta^2 - (1 - \delta) C_f > 0. 
		\end{equation}
		Hence,
		\begin{align}\label{shat_terminal}
			\sup_{s \in \overline{\cS}_{K}} \Big\{ \delta u_K(s) - v_K(s) + 2 \varepsilon l(s) - 2 \zeta^2 - (1 - \delta) C_f  \Big\} > 0,
		\end{align}
		and the maximum is attained in a compact subset of $\overline{\cS}_{K}$ at some $\hat{s} = (\hat{x}_{K:K+1}, \hat{y})$.
		
		Similar to the proof of Proposition \ref{p:comp_btw}, we can obtain a contradiction if $- \lambda_{K} + q_{K+1, n} - q_{K, n} \geq 0$ or $- \theta_{K} - q_{K+1, n} + q_{K, n} \geq 0$ holds for infinitely many $n$'s, where $q_{K+1, n}$ and $q_{K, n}$ are defined similarly as in the proof of Proposition \ref{p:comp_btw}. Hence, we only focus on the third case: $v_K(x'_{K:K+1, n}, y'_n) - w_K (G_K - x'_{K, n})^+ - f(x'_{K+1, n}, y'_n) \geq 0$ holds for infinitely many $n$'s. Then
		\begin{align}
			0 \leq & v_K(s'_n) - w_K (G_K - x'_{K, n})^+ - f(x'_{K+1, n}, y'_n)  \label{comTk1}\\
			& - \big( u_K(s_n) - w_K (G_K - x_{K, n})^+ - f(x_{K+1, n}, y_n)  \big) \delta. \nonumber
		\end{align}
		
		As in the proof of Proposition \ref{p:comp_btw}, $(s_n, s'_n)$ is a maximum point of $\Phi_n(s, s')$, defined in the same way as \eqref{Phi}. We have
		\begin{align}
			&\delta u_K(s_n) - v_K(s'_n) + \varepsilon l(s_n) + \varepsilon l(s'_n) - \psi_n(s_n, s'_n) \nonumber \\
			& \geq \delta u_K(\hat{s}) - v_K(\hat{s}) + 2 \varepsilon l(\hat{s}) - 2 \zeta^2. \label{comTk2}
		\end{align} 
		Therefore, \eqref{comTk1} and \eqref{comTk2} lead to
		\begin{align*}
			0 \leq & \delta w_K (G_K - x_{K, n})^+ - w_K (G_K - x'_{K, n})^+ - f(x'_{K+1, n}, y'_n) + \delta f(x_{K+1, n}, y_n) \\
			& + \varepsilon l(s_n) + \varepsilon l(s'_n) - \psi_n(s_n, s'_n) - \delta u_K(\hat{s}) + v_K(\hat{s}) - 2 \varepsilon l(\hat{s}) + 2 \zeta^2.
		\end{align*}
		When $n \rightarrow \infty$, by the continuity of $l$ and $f$, we have
		\begin{align*}
			0 & \leq v_K(\hat{s}) - \delta u_K(\hat{s}) + (\delta - 1) w_K (G_K - \hat{x}_K)^+ + (\delta - 1) f(\hat{x}_{K+1}, \hat{y}) + 2 \zeta^2\\
			& \leq v_K(\hat{s}) - \delta u_K(\hat{s}) + (1 - \delta) C_f + 2 \zeta^2,
		\end{align*}
		which contradicts \eqref{shat_terminal}.
	\end{proof}

\begin{proof}[Proof of Proposition \ref{prop:comp_last}]
	Denote $s := (t, x_{K+1}, y)$. Assume on the contrary that $ u_{K+1}(\bar{s}) - v_{K+1}(\bar{s}) > 0$ for some $\bar{s} = (\bar{t}, \bar{x}_{K+1}, \bar{y}) \in [T_{K}, T] \times \overline{\cS}_{K+1}$. A strict subsolution $l(s)$ is given by
		\begin{equation}
			l(t, x_{K+1}, y) := - c_1 e^{\gamma(T - t)} \Big(1 + a_{K+1} x_{K+1} \Big)^q - c_2 e^{\gamma(T - t)} (1 + |y|^2),
		\end{equation}
		with $c_1, c_2, \gamma > 0$, $q \in (0, 1)$, and $\gamma$ large enough. Similar to Lemma \ref{lem:strict_l}, it is direct to verify that $l(s)$ is a strict subsolution of \eqref{Vlast} on $[T_{K}, T) \times \overline{\cS}_{K+1}$.

		Consider a small positive constant $\varepsilon > 0$, such that
		\begin{equation}
			u_{K+1}(\bar{s}) - v_{K+1}(\bar{s})  + 2 \varepsilon l(\bar{s}) > 0. 
		\end{equation}
		
		For each $n \geq 0$, define
		\begin{align}
			\Phi_n(s, s') & := u_{K+1}(s) - v_{K+1}(s') + \varepsilon l(s) + \varepsilon l(s') - \psi_n(s, s'), \label{Phi_K+1} \\
			\psi_n(s, s') & := \frac{n}{2} |s - s'|^2.
		\end{align}
		Since $u_{K+1}(s) - v_{K+1}(s')$ is USC and bounded, and $l \rightarrow -\infty$ when $|(x_{K+1}, y) | \rightarrow \infty$ in $\overline{\cS}_{K+1}$, the maximum of $\Phi_n$ in \eqref{Phi_K+1} is attained at 
		$$(s_n, s'_n) = ((t_n, x_{K+1, n}, y_n), (t'_n, x'_{K+1, n}, y'_n))$$
		in some compact subset of $[T_{K}, T] \times \overline{\cS}_{K+1} \times [T_{K}, T] \times \overline{\cS}_{K+1}$. By the compactness, we can extract a subsequence, still denoted as $(s_n, s'_n)$, such that $(s_n, s'_n) \rightarrow (\hat{s}, \hat{s}')$. Moreover,
		\begin{align*}
			\Phi_n(s_n, s'_n) & \geq \sup_{s \in [T_K, T] \times \overline{\cS}_{K+1}} \Phi_0(s, s) \geq u_{K+1}(\bar{s}) - v_{K+1}(\bar{s}) + 2 \varepsilon l(\bar{s}) > 0.
		\end{align*}
		
		It reduces to
		\begin{align}
			& \frac{n}{2} | s_n - s'_n |^2 \leq u_{K+1}(s_n) - v_{K+1}(s'_n) + \varepsilon l(s_n) + \varepsilon l(s'_n) -  \sup_{s \in [T_K, T] \times \overline{\cS}_{K+1}} \Phi_0(s, s). \label{ubd_K+1}
		\end{align}
		Since $u_{K+1}$ and $v_{K+1}$ are bounded, $l(s_n)$ and $l(s'_n)$ converge to $l(\hat{s})$ by the continuity of $l$, and $\sup_{s \in [T_K, T] \times \overline{\cS}_{K+1}} \Phi_0(s, s) > 0$, the right-hand side of \eqref{ubd_K+1} is bounded for all $n \geq 1$. Then $| s_n - s'_n |^2 \rightarrow 0$. Hence $\hat{s} = \hat{s}'$. Putting this back into \eqref{ubd_K+1}, we obtain
		\begin{equation}\label{eq:limsup_K+1}
			\begin{aligned} 
				0 & \leq \limsup_{n \rightarrow \infty} \frac{n}{2}| s_n - s'_n|^2 \\
				& \leq \limsup_{n \rightarrow \infty} \left\{ u_{K+1}(s_n) - v_{K+1}(s'_n) + \varepsilon l(s_n) + \varepsilon l(s'_n) \right\} - \sup_{s \in [T_K, T] \times \overline{\cS}_{K+1}} \Phi_0(s, s) \\
				& \leq u_{K+1}(\hat{s}) - v_{K+1}(\hat{s}) + 2 \varepsilon l(\hat{s}) - \sup_{s \in [T_K, T] \times \overline{\cS}_{K+1}} \Phi_0(s, s) \leq 0, 
			\end{aligned}
		\end{equation} 
		where the third inequality holds since $u_{K+1} - v_{K+1}$ is USC and $l$ is continuous. Therefore,
		\begin{equation}\label{slim_K+1} 
			\lim_{n \rightarrow \infty} n|s_n - s'_n|^2 = 0.
		\end{equation}
		
		We note that $\Phi_n(s_n, s'_n)$ is monotone in $n$ since
		\begin{equation}
			\Phi_n(s_n, s'_n) \geq \Phi_n(s_{n+1}, s'_{n+1}) \geq \Phi_{n+1}(s_{n+1}, s'_{n+1}).
		\end{equation}
		Then the following limit exists and is given by
		\begin{equation}\label{shat_last}
			\begin{aligned}
				\lim_{n \rightarrow \infty} \Phi_n (s_n, s'_n) = & \lim_{n \rightarrow \infty} (u_{K+1}(s_n) - v_{K+1}(s'_n)) + \lim_{n \rightarrow \infty}  (\varepsilon l(s_n) + \varepsilon l(s'_n)) \\
				& - \lim_{n \rightarrow \infty} \frac{n}{2} |s_n - s'_n|^2 \\
				= &  u_{K+1}(\hat{s}) - v_{K+1}(\hat{s}) + 2 \varepsilon l(\hat{s})  = \sup_{s \in [T_K, T] \times \overline{\cS}_{K+1}} \Phi_0(s, s)  > 0,
			\end{aligned}
		\end{equation}
		where \eqref{eq:limsup_K+1} and \eqref{slim_K+1} are used.
		
		By the boundary/terminal conditions, we have $\hat{t} < T$ and $\hat{x}_{K+1} > 0$. Then for sufficiently large $n$, $s_n$ and $s'_n$ are also in $[T_K, T) \times \cS_{K+1}$. Hence, the viscosity sub/supersolution properties of $u_{K+1}$ and $v_{K+1}$ can be applied at $s_n$ and $s'_n$, respectively.

		Since $(s_n, s'_n)$ is a maximum point of $\Phi_n(s, s')$, by applying Ishii's lemma to the USC function $u_{K+1}(s) + \varepsilon l(s)$ and the LSC function $ v_{K+1}(s') - \varepsilon l(s')$, we derive that there exist $(m+1) \times (m + 1)$ symmetric matrices $M_n$ and $N_n$ satisfying
		\begin{align}
			& \left( \frac{\partial \psi_n(s_n, s'_n)}{\partial s}, M_n \right) \in \bar{J}^{2, +}_{[T_K, T) \times \cS_{K+1}} \big[ u_{K+1} + \varepsilon l \big](s_n), \\
			& \left( - \frac{\partial \psi_n(s_n, s'_n)}{\partial s'}, N_n \right) \in \bar{J}^{2, -}_{[T_K, T) \times \cS_{K+1}} \big[ v_{K+1} - \varepsilon l \big](s'_n),
		\end{align} 
		and 
		\begin{equation}\label{eq:Ishii}
			\begin{pmatrix}
				M_n & 0\\
				0 & - N_n
			\end{pmatrix} 
			\leq \Xi_n + \frac{1}{n} (\Xi_n)^2,
		\end{equation}
		where the Hessian matrix $\Xi_n$ is given by
		\begin{align*}
			\Xi_n = n \begin{pmatrix}
				I_{1+m} & - I_{1+m}\\
				-I_{1+m} & I_{1+m}
			\end{pmatrix}.
		\end{align*}
		
		Since $l$ is $C^{1, 2}$ smooth, we can define
		\begin{align*}
			p_{t, n} & :=  \frac{\partial \psi_n(s_n, s'_n)}{\partial t} -  \varepsilon \frac{\partial l(s_n)}{\partial t} = n(t_n - t'_n) - \varepsilon \frac{\partial l(s_n)}{\partial t}, \\
			p_{K+1, n} & := \frac{\partial \psi_n(s_n, s'_n)}{\partial x_{K+1}}  -  \varepsilon \frac{\partial l(s_n)}{\partial x_{K+1}} = n (x_{K+1, n} - x'_{K+1, n}) -  \varepsilon \frac{\partial l(s_n)}{\partial x_{K+1}}, \\
			p_{y, n} & :=  \frac{\partial \psi_n(s_n, s'_n)}{\partial y} -  \varepsilon \frac{\partial l(s_n)}{\partial y} = n (y_n - y'_n) -  \varepsilon \frac{\partial l(s_n)}{\partial y}, \\
			q_{t, n} & := - \frac{\partial \psi_n(s_n, s'_n)}{\partial t'} + \varepsilon \frac{\partial l(s'_n)}{\partial t'} =  n(t_n - t'_n) + \varepsilon \frac{\partial l(s'_n)}{\partial t'}, \\
			q_{K+1, n} & := - \frac{\partial \psi_n(s_n, s'_n)}{\partial x'_{K+1}} + \varepsilon \frac{\partial l(s'_n)}{\partial x'_{K+1}} = n(x_{K+1, n} - x'_{K+1, n}) + \varepsilon \frac{\partial l(s'_n)}{\partial x'_{K+1}}, \\
			q_{y', n} &:= - \frac{\partial \psi_n(s_n, s'_n)}{\partial y'} + \varepsilon \frac{\partial l(s'_n)}{\partial y'} = n(y_n - y'_n) + \varepsilon \frac{\partial l(s'_n)}{\partial y'}, \\
			A_n & := M_n -  \varepsilon \frac{\partial^2 l(s_n)}{\partial (x_{K+1}, y)^2}, \quad B_n := N_n +  \varepsilon \frac{\partial^2 l(s'_n)}{\partial (x'_{K+1}, y')^2}.
		\end{align*}
		Then 
		\begin{align*}
			& \left( p_{t, n}, (p_{K+1, n}, p_{y, n}),  A_n\right) \in \bar{J}^{2, +}_{[T_K, T) \times \cS_{K+1}} \big[ u_{K+1} \big](s_n), \\
			& \left( q_{t, n}, (q_{K+1, n}, q_{y', n}),  B_n \right) \in \bar{J}^{2, -}_{[T_K, T) \times \cS_{K+1}} \big[ v_{K+1}\big](s'_n).
		\end{align*} 
		The semijets definition of viscosity solution yields
		\begin{equation}
			\begin{aligned} 
				& \beta u_{K+1}(s_n) - p_{t, n} - H(x_{K+1, n}, y_n, (p_{K+1, n}, p_{y, n}), A_n) \leq 0, \\
				& \beta v_{K+1}(s'_n) - q_{t, n} - H(x'_{K+1, n}, y'_n, (q_{K+1, n}, q_{y', n}), B_n) \geq 0.
			\end{aligned} 
		\end{equation}
		
		We argue by contradiction. Suppose $\beta v_{K+1}(s'_n) - q_{t, n} - H(x'_{K+1, n}, y'_n, (q_{K+1, n}, q_{y', n}), B_n) \geq 0$ for infinitely many $n$'s. By continuity, the infimum in $H(x'_{K+1, n}, y'_n, (q_{K+1, n}, q_{y', n}), B_n)$ is attained at some $\alpha^*_{K+1}$.  Then
		\begin{align}
			0 \leq & \beta v_{K+1}(s'_n) - q_{t, n} - H(x'_{K+1, n}, y'_n, (q_{K+1, n}, q_{y', n}), B_n) \nonumber \\
			& - \left[ \beta u_{K+1}(s_n) - p_{t, n} - H(x_{K+1, n}, y_n, (p_{K+1, n}, p_{y, n}), A_n) \right] \nonumber \\
			\leq & \beta [v_{K+1}(s'_n) - u_{K+1}(s_n)] - \varepsilon \frac{\partial l(s_n)}{\partial t} - \varepsilon \frac{\partial l(s'_n)}{\partial t'}  \label{ineq:last} \\
			& - r \varepsilon x_{K+1, n} \frac{\partial l(s_n)}{\partial x_{K+1}} - r \varepsilon x'_{K+1, n} \frac{\partial l(s'_n)}{\partial x'_{K+1}}  - \varepsilon \mu_Y(y_n)^\top \frac{\partial l(s_n)}{\partial y}  - \varepsilon \mu_Y(y'_n)^\top \frac{\partial l(s'_n)}{\partial y'}  \nonumber \\
			& +  r n (x_{K+1, n} - x'_{K+1,n})^2 + n [ \mu_Y(y_n) - \mu_Y(y'_n)]^\top(y_n - y'_n) \nonumber \\
			& - (\mu(y'_n) - r\one)^\top \alpha^*_{K+1} x'_{K+1, n} \left( n (x_{K+1, n} - x'_{K+1, n}) + \varepsilon \frac{\partial l(s'_n)}{\partial x'_{K+1}} \right) \nonumber \\
			& - \frac{1}{2} \tr[\Sigma(\alpha^*_{K+1}, x'_{K+1, n}, y'_n) N_n] - \frac{\varepsilon}{2} \tr \left[\Sigma(\alpha^*_{K+1}, x'_{K+1, n}, y'_n) \frac{\partial^2 l(s'_n)}{\partial (x'_{K+1}, y')^2} \right] \nonumber \\
			& + (\mu(y_n) - r\one)^\top \alpha^*_{K+1} x_{K+1, n} \left( n (x_{K+1, n} - x'_{K+1, n}) - \varepsilon \frac{\partial l(s_n)}{\partial x_{K+1}} \right) \nonumber \\
			& + \frac{1}{2} \tr[\Sigma(\alpha^*_{K+1}, x_{K+1, n}, y_n) M_n] - \frac{\varepsilon}{2} \tr \left[\Sigma(\alpha^*_{K+1}, x_{K+1, n}, y_n) \frac{\partial^2 l(s_n)}{\partial (x_{K+1}, y)^2} \right].   \nonumber
		\end{align}

		By the Lipschitz property of $\mu_Y(y)$,
		\begin{align*}
			n [ \mu_Y(y_n) - \mu_Y(y'_n)]^\top(y_n - y'_n) \leq cn | y_n - y'_n|^2,
		\end{align*}
		where $c$ is the Lipschitz constant.
		
		Similarly, by the boundedness of $\{x_{K+1, n}\}$, $\{ x'_{K+1, n}\}$, $\mu(y)$, and the Lipschitz property of $\mu(y)$, we have
		\begin{equation}
			\begin{aligned}
				& (\mu(y_n) - r\one)^\top \alpha^*_{K+1} x_{K+1, n} n (x_{K+1, n} - x'_{K+1, n}) \\
				& \quad  - (\mu(y'_n) - r\one)^\top \alpha^*_{K+1} x'_{K+1, n} n (x_{K+1, n} - x'_{K+1, n}) \\
				& \leq c n |x_{K+1, n} - x'_{K+1, n} |^2 + c n | y_n - y'_n|^2.
			\end{aligned}
		\end{equation}

		With \eqref{eq:Ishii}, Lipschitz property of $\sigma_Y(y)$, and the bounded Lipschitz property of $\sigma(y)$, we have
		\begin{align}
			& \tr\left[\Sigma(\alpha^*_{K+1}, x_{K+1, n}, y_n) M_n \right] - \tr\left[\Sigma(\alpha^*_{K+1}, x'_{K+1, n}, y'_n) N_n \right] \nonumber \\
			& \leq cn \left( |x_{K+1, n} - x'_{K+1, n}|^2 + |y_n - y'_n|^2 \right).
		\end{align}
		
		First, when $\beta > 0$, thanks to the estimates above and the fact that $l(s)$ is a strict subsolution, \eqref{ineq:last} reduces to
		\begin{align}
			0 < & \beta [v_{K+1}(s'_n) - u_{K+1}(s_n)] - \beta \varepsilon l(s_n) - \beta \varepsilon l(s'_n) + c n |x_{K+1, n} - x'_{K+1, n} |^2 + c n | y_n - y'_n|^2 \nonumber \\
			\leq & \beta [v_{K+1}(\hat{s}) - u_{K+1}(\hat{s})] - 2 \beta \varepsilon l(\hat{s}) + c n | s_n - s'_n|^2. \label{ineq:beta1}
		\end{align}
		The second inequality is from
		\begin{align*}
			\Phi_n (s_n, s'_n) \geq u_{K+1}(\hat{s}) - v_{K+1}(\hat{s}) + 2 \varepsilon l(\hat{s}),
		\end{align*}
		since $(s_n, s'_n)$ is a maximum point of $\Phi_n$. When $n \rightarrow \infty$, \eqref{ineq:beta1} leads to
		\begin{align}\label{eq:con}
			0 \leq \beta [v_{K+1}(\hat{s}) - u_{K+1}(\hat{s}) - 2 \varepsilon l(\hat{s})]. 
		\end{align}
		Since $\beta > 0$, \eqref{eq:con} contradicts \eqref{shat_last}.
		
		When $\beta = 0$, we have
		\begin{align*}
			0 \leq & - \varepsilon \frac{\partial l(s_n)}{\partial t} - \varepsilon \frac{\partial l(s'_n)}{\partial t'} \\
			& - \varepsilon r  x_{K+1, n} \frac{\partial l(s_n)}{\partial x_{K+1}} - \varepsilon r x'_{K+1, n} \frac{\partial l(s'_n)}{\partial x'_{K+1}} - \varepsilon \mu_Y(y_n)^\top \frac{\partial l(s_n)}{\partial y}  - \varepsilon \mu_Y(y'_n)^\top \frac{\partial l(s'_n)}{\partial y'}  \\
			&  - \varepsilon (\mu(y'_n) - r\one)^\top \alpha^*_{K+1} x'_{K+1, n}  \frac{\partial l(s'_n)}{\partial x'_{K+1}}  - \frac{\varepsilon}{2} \tr \left[\Sigma(\alpha^*_{K+1}, x'_{K+1, n}, y'_n) \frac{\partial^2 l(s'_n)}{\partial (x'_{K+1}, y')^2} \right]  \\
			&  - \varepsilon (\mu(y_n) - r\one)^\top \alpha^*_{K+1} x_{K+1, n}  \frac{\partial l(s_n)}{\partial x_{K+1}} - \frac{\varepsilon}{2} \tr \left[\Sigma(\alpha^*_{K+1}, x_{K+1, n}, y_n) \frac{\partial^2 l(s_n)}{\partial (x_{K+1}, y)^2} \right] \\
			& + c n | s_n - s'_n|^2.   
		\end{align*}
		
		Letting $n \rightarrow \infty$, by the continuity of derivatives of $l(s)$, we have
		\begin{align*}
			0 \leq 	& - 2 \varepsilon \frac{\partial l(\hat{s})}{\partial t} - 2 \varepsilon r \hat{x}_{K+1} \frac{\partial l(\hat{s})}{\partial x_{K+1}}  - 2 \varepsilon \mu_Y(\hat{y})^\top \frac{\partial l(\hat{s})}{\partial y} \\
			&  - 2 \varepsilon \cL^{\alpha^*_{K+1}} \left[ \hat{x}_{K+1}, \hat{y}, \frac{\partial l(\hat{s})}{\partial (x_{K+1}, y)}, \frac{\partial^2 l(\hat{s})}{\partial (x_{K+1}, y)^2} \right].
		\end{align*}
		However, the strict subsolution property of $l$ at $\hat{s}$ indicates that the right-hand side is strictly smaller than zero, which is a contradiction.
	\end{proof}
	
\end{document}